\documentclass[twoside,11pt]{article}

 \usepackage{hyperref}
\hypersetup{
    colorlinks=true,
    linkcolor=blue,
    filecolor=magenta,      
    urlcolor=cyan,
      citecolor=blue,
    }

 \usepackage[round]{natbib}

 \usepackage{amssymb, bbm,authblk}
\usepackage{amsmath}
\usepackage{amsthm}
\usepackage{graphicx}
 \usepackage{algorithm}
\usepackage{algpseudocode}
\usepackage{tikz}
\usetikzlibrary{decorations.pathreplacing}
\usepackage{geometry}
\usepackage{cleveref}
\algnewcommand\algorithmicinput{\textbf{INPUT:}}
\algnewcommand\INPUT{\item[\algorithmicinput]}
\algnewcommand\algorithmicoutput{\textbf{OUTPUT:}}
\algnewcommand\OUTPUT{\item[\algorithmicoutput]}

\newtheorem{theorem}{Theorem}
\newtheorem{proposition}{Proposition}

\newtheorem{lemma}{Lemma}

\newtheorem{definition}{Definition}
\newtheorem{remark}{Remark}

\newcommand{\diag}{\mathrm{diag}}
\newcommand{\op}{\mathrm{op}}

\newcommand{\rank}{\mathrm{rank}}
\newcommand{\sign}{\mathrm{sign}}
\newcommand{\kse}{\kappa_{\max}^{s,e}}

\DeclareMathOperator*{\argmax}{argmax}
\DeclareMathOperator*{\usvt}{USVT}

\newtheorem{assumption}{Assumption}

\date{\vspace{-5ex}}

\begin{document}
	
  	\title{Optimal Change Point Detection and Localization in Sparse Dynamic Networks}
 \author[1] {Daren Wang}
\author[2]{Yi Yu}
\author[3]{Alessandro Rinaldo}
 \affil[1]{\small Department of Statistics,  University of Chicago}
 \affil[2]{\small Department of Statistics,
				University of Warwick}
\affil[3]{\small Department of Statistics and Data Science,  Carnegie Mellon University}

\maketitle

\begin{abstract}
We study the problem of change point localization in dynamic networks models.  We assume that we observe a sequence of independent adjacency matrices of the same size, each corresponding to a realization of an unknown inhomogeneous  Bernoulli model. The underlying distribution of the adjacency matrices are piecewise constant, and may change over a subset of the  time points, called change points. We are concerned with recovering the unknown number and positions of the change points. In our model setting we allow for all the model parameters to change with the total number of time points, including the network size,  the minimal spacing between consecutive change points, the magnitude of the smallest change and the degree of sparsity of the networks.
We first identify a region of impossibility in the space of the model parameters such that no change point estimator is provably consistent if the data are generated according to parameters falling in that region. We propose a computationally-simple  algorithm for network change point localization, called Network Binary Segmentation, that relies on weighted averages of the adjacency matrices. We show that  Network Binary Segmentation is consistent over a range of the model parameters that nearly cover the complement of the impossibility region, thus demonstrating the existence of a phase transition for the problem at hand.  Next, we devise a more sophisticated algorithm based on singular value thresholding, called Local Refinement, that delivers more accurate estimates of the change point locations.  Under appropriate conditions, Local Refinement guarantees a minimax optimal  rate for network change point localization while remaining computationally feasible.

\end{abstract}
\

\section{Introduction}\label{section:introduction}
The analysis of network is a fundamental task in statistics due to the increasing popularity of network data generated from various scientific areas, social sciences, emerging industries, as well as everyday life.  
Over the last decade, most of the advances in the area of statistical network analysis have revolved around {\it static network models}, where the properties of the data generating process are inferred from a single realization of the network. 
For this type of problems, a large collection of results  of computational, methodological and theoretical nature exist. 

In contrast to the basic premise of the static network modeling framework, many modern network data sets consist instead of multiple network realizations indexed by time, so that both the number of nodes and the connectivity structure of the network exhibit time-varying features.  Such a {\it dynamic network modeling} setting is naturally more complex and challenging, as it is necessary to additionally formalize and model the underlying temporal dynamic. While there is a vast body of work on dynamic network models  \citep[see, e.g.,][]{BarabasiAlbert1999} in the broader scientific literature, theoretical results on such models are comparatively scarce in the statistical literature, with many of the contributions being fairly recent (see \Cref{sec-related-work} below for some literature review).


In this article we  are concerned with a discrete time network dynamic setting in which the set of nodes is fixed but the edge probabilities are time-varying.  We assume that we observe a sequence of $T$ independent and possibly sparse networks of constant size whose distributions may change at $K < T$  unknown time points, or change points. We impose minimal restrictions on the number and locations of the possible change points and especially on the nature of the  distributional changes that may occur at those times. In particular, most popular static network models can fit into our framework.
Our goal is to detect whether any such change has taken place, and to accurately estimate the time  of the corresponding change point. Importantly, we are not interested in estimating the underlying data-generating distributions. As our analysis will reveal, although we only consider a fairly straightforward form of network dynamics, the associated inference problem is rather  subtle and far from trivial.  Furthermore, if one is interested in the underlying distributions, then static network estimation methods can be applied to the sample means of the adjacency matrices between two consecutive change point estimators.

\subsection{Problem setup}

To set up the problem, we assume a sequence of $T$ independent adjacency matrices of size $n$, each  from a possibly sparse inhomogeneous Bernoulli network model, defined next.  

\begin{definition}[Inhomogeneous Bernoulli networks]\label{def-1}
	A network with node set $\{1, \ldots, n\}$  is an inhomogeneous Bernoulli network if its adjacency matrix $A \in \mathbb{R}^{n\times n}$ satisfies
		\[
		A_{ij} = A_{ji} = \begin{cases}
			1, & \mbox{nodes $i$ and $j$ are connected by an edge},\\
			0, & \mbox{otherwise};
		\end{cases}
		\]
		and $\{A_{ij}, i < j\}$ are independent Bernoulli random variables with $\mathbb{E}(A_{ij}) = \Theta_{ij}$.  
\end{definition}

\Cref{def-1} covers a wide range of models for undirected networks, including the Erd\H{o}s--R\'enyi random graph \citep{ErdosRenyi1959}, the stochastic block model \citep{HollandEtal1983}, the degree corrected block model \citep{KarrerNewman2011} and the random dot product model \citep{YoungScheinerman2007}, etc. It is worth pointing out that although  we are only considering undirected networks, our results extend straightforwardly to directed networks, i.e.~asymmetric  adjacency matrices. Additionally, for technical convenience, we are allowing self-loops, even though networks with no loops can be easily accommodated; see \Cref{section:SBM} below. 
Finally, discussions on the possible relaxations on the independence and Bernoulli assumptions can be found in Section~\ref{sec-disc}.

We further assume that the probability distributions of the networks change only over an unknown subset of the time points, called change points. 
We formalize our setting below.


\begin{assumption}[Change point dynamic network model]\label{assume:model}
	Let $\left\{ A(t)\right\}_{t = 1}^T$ be a sequence of $n \times n$ adjacency matrices of independent inhomogeneous Bernoulli networks with means $ \left\{\Theta(t)\right\}_{t=1}^T$ satisfying the following properties.
\begin{enumerate}
	\item  The sparsity parameter  
		\begin{equation}\label{eq-rho-def-1} 
			\rho : = \max_{t = 1, \ldots, T} \|\Theta(t)\|_{\infty}
		\end{equation}
	is such that
	\begin{equation}\label{eq-assume-sparse}
	\rho n  \ge \log(n),
	\end{equation}
	where $\|\cdot\|_{\infty}$ denotes the entrywise maximum norm of a matrix.
	\item  There exists a sequence  $(\eta_0,\ldots,\eta_{K+1})$ of time points, called change points, such that  $1 = \eta_0 < \eta_1 < \ldots < \eta_K \leq T < \eta_{K+1} = T + 1$ and, for $t=2,\ldots,T$, 
	\[
	\Theta(t) \neq \Theta(t-1) \quad  \text{if and only if} \quad t \in \{ \eta_1,\ldots,\eta_K\}.
	\]
	\end{enumerate}
We let   
	\[
	\Delta := \min_{k = 1, \ldots, K+1} \{\eta_k-\eta_{k-1}\} \leq T
	\]
be the minimal spacing between two consecutive change points and set
	\begin{equation}\label{eq-as1-frob}
		\kappa_0 : = \frac{\min_{k=1,\ldots,K} \|\Theta (\eta_{k}) -\Theta (\eta_{k}-1 ) \|_{\mathrm{F}}}{n\rho} \in (0,1],	
	\end{equation}	
to be the normalized magnitude of the smallest changes in the data generating distribution, where $\|\cdot\|_{\mathrm{F}}$ denotes the Frobenius norm.

\end{assumption}

A few comments on our modeling assumptions are in order. First, we rely on the Frobenius norm of the difference between two consecutive expected adjacency matrices at a change point to quantify the magnitude of the corresponding distributional change. This is a fairly general metric, able to capture both  ``dense'' changes caused by small variations spread across many edge probabilities as well by ``sparse'' changes due to large difference only along few coordinates. Next, the quantity $\kappa_0 \in (0,1]$ appearing in \eqref{eq-as1-frob} measures the  size of the smallest distributional change in the model in a manner that is independent of the choice of the other parameters. Indeed, the terms $\|\Theta (\eta_{k} ) -\Theta (\eta_{k}-1 ) \|_{\mathrm{F}}$'s depend on both the sparsity parameter $\rho$ and the size of the networks $n$. To avoid such confounding, and using the fact that $\max_k \|\Theta(\eta_k) - \Theta(\eta_{k}-1)\|_{\mathrm{F}} \leq n\rho$, setting $\kappa_0$ as in \eqref{eq-as1-frob} yields a scale-free parameter in $(0, 1]$ that is independent of both $\rho$ and $n$.

The model described above is defined by the parameters $\Delta$, $\kappa_0$, $n$ and $\rho$. We adopt a high-dimensional framework whereby $T$ grows unbounded and all the defining parameters are allowed to change as a function of $T$.  The number of change points $K$ also may change with $T$, but since $K \leq \frac{T}{\Delta}$ by definition, we will capture any dependence on $K$ only through $\Delta$.
 We refer to any relationship among all the model parameters  $(\Delta, \kappa_0, n, \rho)$ and $T$ that holds as $T \rightarrow \infty$ as a {\bf scaling}. For ease of readability we will not make the dependence on $T$ explicit in our notation.

 We are concerned with the problem of estimating the unknown number and unknown locations of the change points based on one observation of a sequence $(A(1), \ldots, A(T))$ of adjacency matrices satisfying the above assumptions. More precisely, for a given  scaling  of the model parameters, we aim to construct an estimator of $(\eta_1,\ldots,\eta_{K})$ of the form
\begin{equation}\label{eq:estimator}
(A(1),\ldots,A(T)) \mapsto (\hat{\eta}_1,\ldots, \hat{\eta}_{\hat{K}}) \subset (2,\ldots,T )
\end{equation}
and with $\hat{\eta}_1 < \hat{\eta}_2 <  \ldots <  \hat{\eta}_{\hat{K}}$  satisfying the following notion of localization consistency.

\begin{definition}[Consistent localization]\label{def:consistency}
A change point estimator of the form \eqref{eq:estimator} is consistent if, with probability tending to $1$ as $T \rightarrow \infty$, 
\begin{equation}\label{eq:consistency}
\widehat{K} = K \quad \text{and} \quad \max_{k=1,\ldots,K} | \hat{\eta}_{k} -  \eta_k| \leq \epsilon, 
\end{equation}
where  $\epsilon = \epsilon(T,\Delta,\kappa_0,\rho,n)$ is such that
\begin{equation}\label{eq:consistency.def}
 \frac{\epsilon}{\Delta} \rightarrow  0.
\end{equation}
The term $\epsilon$ is called the {\bf localization error} of the estimator and the sequence  $\left\{ \frac{\epsilon}{\Delta} \right\}$ the {\bf localization rate}.
\end{definition}

Thus, we will deem a change point estimator consistent if, with high probability as the number of time points grows, its localization error is a vanishing fraction of the minimal distance between consecutive change points. 
The limiting  probability (in $T$) of the event in \eqref{eq:consistency}  and the value of the localization error $\epsilon$   depend on the choice of the scaling. For instance, it is intuitively clear that scalings in which all parameters decrease with $T$ will lead to a sequence of change point problems of increasing difficulty.  


Our main goal is to derive conditions on the scaling of the model parameters that guarantee the feasibility of consistent estimation of the change points and to derive computationally efficient estimators that are consistent and in fact optimal, in the sense of achieving the minimax localization rate. Throughout, we will specify any scaling regime among the parameters by expressing them as functions of the quantity 
\begin{equation}\label{eq:SNR}
 \sqrt{\rho} \kappa_0,
\end{equation}
which can be considered as a uniform lower bound on the signal-to-noise ratio for any network change point model satisfying  \Cref{assume:model}. Indeed, the above quantity is the minimal magnitude of the signal jump, namely $\kappa_0 n \rho$, divided by $n\sqrt{\rho}$, which is an upper bound on the total standard deviation of the entries of $A$.


\subsection{List of contributions}
The main theoretical contribution  of the paper is the identification of three regions inside the space of model parameters corresponding to different types of scaling or regimes: (i) an impossibility regimes, where no change point localization algorithm is guaranteed to be consistent (see \Cref{sec-2.1}); (ii) a feasibility regime, described in \Cref{assume:phase}), for which we demonstrate the existence of a polynomial-time, consistent change point estimator (see \Cref{sec-NBS}); and (iii) a subset of (ii), described in \Cref{assume:phase 2}, for which we further show that change point localization can be achieved at a nearly minimax optimal rate, again using a polynomial-time algorithm (see \Cref{sec-local}).  The partition of scaling regimes, represented pictorially in  \Cref{fig:rg}, is relatively sharp, in the sense that regimes (i) and (ii) are only off by any diverging factor in $T$. 

\begin{figure}
\begin{tikzpicture}
\draw (0,0) ellipse (3cm and 2cm);	
\draw (0,0) ellipse (2.5cm and 1.5cm);	
\draw (0,0) ellipse (2cm and 1cm);	
\draw (0,0) ellipse (1cm and 0.5cm);	

\node (A) at (-2, 1) {};
\node (A2) at (-3, 1.4) {};
\node (A3) at (-4, 1.4) {$\substack{\mbox{Infeasible regime}\\ (\Cref{sec-2.1})}$};
\node (B) at (0.3, 0.6) {};
\node (B2) at (3, 1.4) {};
\node (B3) at (4.6, 1.4) {$\substack{\mbox{Consistent localization} \\ (\Cref{sec-NBS})}$};
\node (C) at (-0.3, -0.3) {};
\node (C2) at (-3, -1.4) {};
\node (C3) at (-4.5, -1.4) {$\substack{\mbox{Optimal localization} \\ (\Cref{sec-local}) }$};
\node (D) at (0.2, -0.2) {};
\node (D2) at (3, -1.4) {};
\node (D3) at (4.5, -1.4) {$\substack{\mbox{An example: SBM} \\ (\Cref{section:SBM})}$};
\draw[->]
	(A) edge (A2) (B) edge (B2) (C) edge (C2) (D) edge (D2);
\draw [decorate,decoration={brace, mirror},xshift= 0.3cm,yshift=0pt]
(0,0) -- (0, 1.2) node [black, midway]{};	
\draw [decorate,decoration={brace, mirror},xshift= -0.3cm,yshift=0pt]
(0,0.2) -- (0, -0.8) node [black, midway]{};	
	
\end{tikzpicture}	
\caption{Reading guide.}\label{fig:rg}
\end{figure}

To be specific, our contributions are as follows.
\begin{itemize}
\item We first demonstrate the existence of a phase transition for the problem at hand by giving nearly matching necessary and sufficient conditions on the scaling of the model parameters and $T$ for consistent estimation of the change points. Specifically, under the low signal-to-noise scaling
\begin{equation}\label{eq:phase}
\rho \kappa^2_0 \lesssim \frac{\log(T)}{n\Delta},
\end{equation}
no algorithm is guaranteed to be consistent (in the minimax sense: there exists a change point problem setting compatible with the above assumption such that  any algorithm will have a localization rate uniformly bounded away from $0$). 
On the other hand, if for any $\xi>0$\footnote{In fact,  $\xi$ is allowed to be zero if $n$ diverges with $T$.  More generally, in that case, we may replace the term $\log^{\xi}(T)$ with any other quantity one diverging in $T$.}, 
\begin{equation}\label{eq:phase.consistent}
	\rho \kappa^2_0 \gtrsim   \frac{ \log^{2 + 2 \xi}(T)}{\Delta n},
	\end{equation}
	 we demonstrate a computationally-efficient procedure, called Network Binary Segmentation (NBS) (see \Cref{algorithm:MWBS} below) that is provably consistent. The procedure combines sample splitting with the randomized search strategy implemented in the wild binary segmentation (WBS) algorithm of \cite{fryzlewicz2014wild}. To show the consistency of the NBS we have generalized in non-trivial ways the analysis in \cite{venkatraman1992consistency} to allow for vector- and matrix-valued CUSUM statistics; we believe that such generalization may be applied to other change point detection problems and is of independent interest. 

The NBS is consistent under nearly the weakest possible conditions, since it leads to a vanishing localization rate under the scaling \eqref{eq:phase.consistent} which, save for a $\log^{1+2\xi}(T)$ term, matches the phase transition boundary in \eqref{eq:phase}.    Remarkably, no structural assumptions on the distributions of the networks themselves  are used. Indeed, in deriving the bound \eqref{eq:phase}, we construct a worst-case class of distributions consisting of dynamic networks satisfying stochastic block models. This reveals that, under the scaling in which the NBS is analyzed,  imposing extra network  structural assumptions do not necessarily lead to easier change point detection problems. This is in stark contrast with many other network problems, such as graphon estimation, clustering and testing, where some structural conditions on the edge probabilities are always necessary.
For instance, \cite{GaoEtal2015} showed that, when the number of communities $r$ in a stochastic block model is of order $n$, the minimax lower bound under the normalized mean squared error loss for graphon estimators is of order $1$.  The dynamic version optimality is shown in \cite{pensky2016dynamic}.

\item  In our second set of results, we seek to investigate conditions under which structural assumptions do help with our change point localization task. Towards that end, we introduce additional  assumptions on the model defined in \Cref{assume:model} by requiring that each difference $\Theta (\eta_{k} ) -\Theta (\eta_{k}-1)$, $k=1,\ldots,K$,  be a matrix of rank at most $r \leq n$, an additional parameter that is also allowed to change with $T$. Such low rank condition is relative mild and is satisfied by many instances of the stochastic block model. 
Then, with this assumption in place and under the  stronger  scaling
\begin{equation}\label{minimax.scaling}
\rho \kappa_0^2  \gtrsim \frac{\log^{2 + 2 \xi} (T) }{\Delta}\frac{r}{n},
\end{equation}
we are able to devise a computationally-efficient and consistent change points estimator with localization error of the order
\begin{equation}\label{eq:lr.rate}
\epsilon \lesssim \frac{\log^2(T)}{ \kappa_0^2 n^2\rho}.
\end{equation}

The proposed procedure takes as input the estimates of the change point locations from any reasonable (not necessarily consistent nor optimal) estimator, including the NBS, and further improves their accuracy to deliver the above localization rate. 
At its core, the LR algorithm relies on exactly
$K$ (this, we recall, being the number of change points) separate applications of the universal singular value thresholding  procedure of \cite{Chatterjee2015}. 
%
%
 Furthermore, we show that the localization rate afforded by the LR algorithm, given in \eqref{eq:lr.rate}, is in fact nearly minimax rate-optimal, aside for the $\log^{2}(T)$ term. Interestingly, the expression of the rate \eqref{eq:lr.rate} is essentially identical to the optimal localization rate for covariance and mean change point estimation, adjusted for the differences in the model settings \citep[e.g.][]{WangEtal2017}.

More discussions on the gap between the scalings \eqref{eq:phase.consistent} and \eqref{minimax.scaling}, and on the comparisons with \cite{WangEtal2017} are provided later in the paper.

\item We apply the LR algorithm to the problem of change point detection for sequence of networks from stochastic block models and derive optimal localization rates. For networks without self-loops -- a common feature of network models -- a technical complication arises in treating the expected adjacency matrix from a stochastic model as a low-rank matrix.  When the network has no self-loops, the diagonal entries of the expected adjacency matrix are set to be zero, which in general would prevent the low-rank assumption.  In fact, this complication is often ignored in the existing literature.   In this case, we show that with a very mild additional assumption, we are still able to recover the nearly optimal localization rate \eqref{eq:lr.rate}. In our analysis we borrow tools and ideas from several areas, including change point detection, network analysis and graphon estimation. 
\end{itemize}


The rest of this paper is organized as follows. \Cref{sec-related-work} summarize some of the related literature. In \Cref{section:consistent testing}, we first identify the scalings for which consistent localization is impossible and then present the NBS change point estimator, which we show to be consistent under almost any scaling outside this impossibility regime. In \Cref{sec-local} we develop the more sophisticated algorithm LR, which we then show to be almost minimax rate-optimal under an additional low-rank assumption. We further demonstrate in \Cref{section:SBM} how our procedure is applicable to the dynamic stochastic block model. \Cref{sec:simulations} presents few illustrative simulations that verify the effectiveness of our procedures. Finally, we conclude with more discussions including potential future work directions in \Cref{sec-disc}.  The proofs of our results are presented in the the appendix and supplementary material.

\subsection{Related work}\label{sec-related-work}

Dynamic network is a topical area which is intensely studied across different disciplines.  The relevant papers listed in this section are by no means exhaustive.  Readers may refer to \cite{CarringtonEtal2005}, \cite{GoldenbergEtal2010}, \cite{BoccalettiEtal2014}, \cite{Kolaczyk2017} for more comprehensive reviews. 

In terms of the invariant quantities, most of the existing work focus on a fixed set of nodes across time, but there are also exceptions. For instance, \cite{BarabasiAlbert1999} allowed for time-varying nodes and edges, \cite{Crane2015} assumed infinite population at every time point and allowed for random observations at different time points, to name but a few.  In terms of the network models imposed for every time point, \cite{Snijders2002} explored dynamic exponential random graph models, \cite{TangEtal2013} studied a dynamic version of random dot product models, \cite{HoEtal2011} extended the mixed membership models to a dynamic one, \cite{XuZheng2009}, \cite{SewellChen} among others considered dynamic latent space models, and dynamic stochastic block models have also been extensively studied.  

Among the work on dynamic stochastic block models, \cite{Xu2015} proposed a stochastic block transition model using a hidden Markov-type approach; \cite{XuHero2014} proposed to track dynamic stochastic block models using Gaussian approximation and an extended Kalman filter algorithm; \cite{MatiasMiele2017} integrated a Markov chain determined group labels evolving process; \cite{PenskyZhang:17} exploited kernel-based smoothing techniques dealing with the evolving block structures; \cite{BhattacharyyaChatterjee} focused on time-varying stochastic block model and variants thereof with time-independent community labels, applied spectral clustering on an averaged version of adjacency matrices, and achieved consistent community detection.  \cite{bhattacharjee2018change} dealt with a change point detection problem in a one-change-point stochastic block model sequences and focused on recovering underlying models, which resulted in a cost of sub-optimal change point detection.  \cite{WangEtal2014} used two types of scan statistics investigating change point detection on time-varying stochastic block model sequences, emphasizing testing connectivity matrices changes.  \cite{CribbenYu2017} proposed an eigen-space based statistics testing the community structures changes in stochastic block model sequences.  \cite{LiuEtal2018} proposed a loss function based on the eigen-space to track the changes of the community structures in stochastic block model sequences.  Both \cite{CribbenYu2017} and \cite{LiuEtal2018} have roots in subspace tracking in signal processing, but both lack theoretical justifications. \cite{ChuChen2017} proposed a test statistics for general data type including network sequences, and their method focuses on the testing perspective.  \cite{zhao2019change} provided a two-step algorithm, which first estimates the networks and then uses a moving window to detect change points.  The results thereof are consistent yet optimal. Another consistent yet optimal result on network change point detection problems is derived in Chapter 5 in \cite{Mukherjee}.

\subsection{Notation}\label{sec-notation}
For any $A \in \mathbb R^{n\times n}$, let $A_{ij}$ be the $(i, j)$th entry of $A$, $A_{i*}$ and $A_{*j}$ the $i$th row and $j$th column of $A$.  Let $\kappa_i(A) $ be the $i$th eigenvalue of $A$ with ordering $|\kappa_1(A)| \ge |\kappa_2(A)| \ge \ldots  \ge |\kappa_n(A)|$, and $\|A\|_{\mathrm{op}} = |\kappa_1(A)|$ be the operator norm of $A$.   Let $\|A\|_{\infty} = \max_{1 \leq i, j \leq n}|A_{ij}|$ be the entrywise maximum norm.  
In addition, for any $B \in \mathbb R^{n\times n}$, let $(A,B) =\sum_{1\le i,j\le n} A_{ij}B_{ij}$ be the inner product of $A$ and $B$ in the matrix space, and $\| A\|_{\mathrm{F}} = \sqrt{(A,A)}$ be the Frobenius norm of $A$.   For any vector $v \in \mathbb{R}^p$, let $v_i$ be the $i$th entry of $v$, $\|v\|$ and $\|v\|_{\infty}$ be the $\ell_2$- and entrywise maximum norms of $v$, respectively.  For any set $S$, let $S^c$ be its complement.  

For any positive functions of $n$, namely $f(n)$ and $g(n)$, denote $f(n) \lesssim g(n)$, if there exist constants $C > 0$ and $n_0$ such that $f(n) \leq Cg(n)$ for any $n \geq n_0$; denote $f(n) \gtrsim g(n)$, if $g(n) \lesssim f(n)$; and denote $f(n) \asymp g(n)$, if $f(n) \lesssim g(n)$ and $f(n) \gtrsim g(n)$.

We now recall the definition of cumulative sum (CUSUM) statistic \citep{page1954continuous}.

\begin{definition}[CUSUM statistics]\label{def-cusum}
	For a collection of any type of data $\{X(t)\}_{t=1}^T$, any pair of time points $(s, e) \subset \{0, \ldots, T\}$ with $s < e -1$, and any time point $t = s+1, \ldots, e-1$, let the CUSUM statistics be
	\[
	\widetilde X^{s,e}(t) = \sqrt {\frac{e-t}{(e-s)(t-s)} } \sum_{i=s+1}^t X_i  -\sqrt {\frac{t-s}{(e-s)(e-t)} } \sum_{i=t+1}^e X_i.
	\]
\end{definition}
Since the CUSUM statistic is linear in its arguments, we have that, for any $0 \leq s < t < e \leq T$, 
	\[
	\mathbb{E}(\widetilde A^{s,e}(t))   = \widetilde \Theta^{s,e}(t). 
	\]

\section{Consistent localization}\label{section:consistent testing}

In  this section we study the conditions under which consistent estimation of the change point locations for the model described in \Cref{assume:model} is feasible. Specifically, we derive a phase transition in the space of the model parameters that separates parameter scalings for which there exists some algorithm with a vanishing localization rate from the ones for which no estimator is consistent. To be precise, when we say that consistent localization is impossible for a given scaling, we mean it in a minimax sense that there exists {\it some} change point model satisfying  \Cref{assume:model} for which no estimator of the change points is consistent.

\subsection{The impossibility regime}\label{sec-2.1}
Below we establish an information-theoretic lower bound, which demonstrates that, if 
\begin{equation}\label{eq:no.consistent}
\rho \kappa^2_0 \lesssim  \frac{\log(T)}{n\Delta},
\end{equation}
then no consistent estimator of the change points  exists. The proof constructs two sequences of mixtures of stochastic block models with two communities of all possible sizes that cannot be reliably discriminated under the above scaling, and then employs the convex version of Le Cam's Lemma \citep[see, e.g.][]{yu1997assouad} to conclude that any change point estimator must have a localization rate bounded away from zero. As a by-product of our lower bound construction, we also see that imposing additional structural assumptions on the edge probabilities (such as that of a stochastic block model with a bounded number of communities and therefore low rank) does not necessarily lead to a consistent estimator under the scaling in \eqref{eq:no.consistent}. The details are given in \Cref{sec-lower-proof}.

\begin{lemma}\label{lemma:lower bound testing}
Let $\{A(t)\}_{t=1}^T$ be a sequence of independent inhomogeneous Bernoulli networks satisfying Assumption~1 with $K = 2$ (i.e. there exist two and only two change points).  Let $P_{\kappa_0, \Delta, n, \rho}^T$ denote the corresponding joint distribution.  Consider the class of distributions 
		\[
			\mathcal{P} = \left\{ P^T_{\kappa_0,\Delta,n ,\rho}: \, \Delta = \min \biggl\{\bigg\lfloor \frac{\log(T)}{n\rho\kappa_0^2} \bigg\rfloor, \, \lfloor T/4 \rfloor \biggr\},\, \rho < 1/2, \,\kappa_0 \leq 1 \right\}.
		\]
	Then there exists a $T_0$, such that for all $T \geq T_0$,
		\[
			\inf_{\hat \eta} \sup_{ P \in \mathcal{P}} \mathbb{E}_P(H(\widehat{\eta}, \eta(P))) \geq \Delta/2,
		\] 
		where the infimum is over all estimators $\widehat{\eta} = \{\widehat{\eta_k}\}^{\widehat{K}}_{k=1}$ of the change point locations, $\eta(P)$ is the set of the change points of $P \in \mathcal{P}$ and $H(\cdot,\cdot)$ denotes the Hausdorff distance.
\end{lemma}


\subsection{Network Binary Segmentation}\label{sec-NBS}

In our next result, we show that parameter scalings of the form given in \eqref{eq:no.consistent} are essentially the only ones for which consistent change point estimation is infeasible, thus proving the existence of a phase transition in the space of parameters. In particular, we will derive an algorithm (see \Cref{algorithm:MWBS} below) that will return a consistent estimator provided the following signal-to-noise condition is met.

\begin{assumption}  \label{assume:phase}
	For a constant $C_\alpha>0$ and any $\xi>0$, we have that
	\begin{equation}\label{eq:snr}
\kappa_0 \sqrt{\rho} \ge C_\alpha  \sqrt{\frac{1}{n\Delta}}  \log^{1 + \xi}(T).
\end{equation}
\end{assumption}

Recalling \eqref{eq:no.consistent}, our results cover all  parameter scalings, aside from a $\log^{\xi}(T)$ term, where $\xi > 0$ can be arbitrarily small.
%
When the size of the networks $n$ diverges with $T$, arguably a very natural asymptotic regime, one can take $\xi$ in \Cref{assume:phase} to be zero.  In fact, in this case the  signal-to-noise ratio condition \eqref{eq:snr} can be weakened to be of the form $\kappa_0 \sqrt{\rho} \ge C_\alpha  \sqrt{\frac{1}{n\Delta}}  \log(T) e_T$, for any  sequence of positive numbers $\{ e_T\}_{T=1,2,\ldots}$ diverging to infinity arbitrarily slowly.

To appreciate how \Cref{assume:phase}  is compatible with a broad range of network change point scenarios and is therefore fairly mild, we highlight the following two extreme cases.
\begin{itemize}
\item Assume a non-sparse setting (i.e. $\rho \asymp 1$). If the minimal spacing $\Delta $ is of order $\log^{2+2\xi} (T)$, 
then \Cref{assume:phase} demands that $n\kappa_0 \succeq n^{1/2}$.  This means that the edge probabilities need to change for at least  $\sqrt{n}$ order many nodes.
\item On the other hand, in the sparse setting  where $\rho$ is chosen to be $\log(n)/n$ as in \eqref{eq-assume-sparse}, if $\Delta \asymp T$ (so that the number of change points is bounded), then \Cref{assume:phase} only  requires $\kappa_0$ to be at least of the order  
	\[
	\frac{\log^{1 + \xi}(T)} {\sqrt{ T \log(n)}}.
	\]
	Thus  
	$\kappa_0$ is allowed to vanish with $T$, even for fixed $n$.
\end{itemize}

 We now introduce the procedure Network Binary Segmentation (NBS), detailed in \Cref{algorithm:MWBS}, for consistent estimation under nearly the worst possible scaling of \Cref{assume:phase}.  
	
\begin{algorithm}[!ht]
	\begin{algorithmic}
		\INPUT Two independent samples $\{A(t)\}_{t=1}^{T}, \{B(t)\}_{t=1}^{T} \in \mathbb{R}^{n\times n}$, $\tau_1$.
		\For{$m = 1, \ldots, M$}  
			\State $[s_m', e_m'] \leftarrow [s,e]\cap [\alpha_m,\beta_m]$
			\State $(s_m,e_m) \leftarrow [s_m' + 64^{-1} (e'_m - s'_m),e_m' - 64^{-1} (e_m' - s_m')] $  
			\If{$e_m - s_m \ge 1$}
				\State $b_{m} \leftarrow \arg\max_{t = s_m+1, \ldots, e_m-1}  ( \widetilde A^{s_m,e_m}( t),\widetilde B^{s_m,e_m} (t))$
				\State $a_m \leftarrow ( \widetilde A^{s_m,e_m} (b_m),\widetilde B^{s_m,e_m} (b_m))$
			\Else 
				\State $a_m \leftarrow -1$	
			\EndIf
		\EndFor
		\State $m^* \leftarrow \arg\max_{m = 1, \ldots, M} a_{m}$
		\If{$a_{m^*} > \tau_1$}
			\State add $b_{m^*}$ to the set of estimated change points
			\State NBS$((s, b_{m*}),\{ (\alpha_m,\beta_m)\}_{m=1}^M, \tau_1)$
			\State NBS$((b_{m*}+1,e),\{ (\alpha_m,\beta_m)\}_{m=1}^M,\tau_1)$
			
		\EndIf  
		\OUTPUT The set of estimated change points.
		\caption{Network Binary Segmentation. NBS$((s, e),$ $\{ (\alpha_m,\beta_m)\}_{m=1}^M, \tau_1$)} \label{algorithm:MWBS}
	\end{algorithmic}
\end{algorithm}

The NBS is a novel algorithm that builds on the traditional machinery developed for the univariate mean change point detection problem.
	 The cornerstones of the NBS are the CUSUM statistics $\widetilde{A}^{s_m, e_m}(t)$ and $\widetilde{B}^{s_m, e_m}(t)$ (see \Cref{def-1}). However, instead of searching for the maximum CUSUM statistics directly, as it is traditionally done in the binary segmentation and its more modern variants \citep[see, e.g.][] {vostrikova1981detection,fryzlewicz2014wild, wang2016high}, the NBS maximizes the inner product of two CUSUM statistics based on two independent samples.  This is due to the fact that each entry of the adjacency matrix is a Bernoulli random variable, and for any Bernoulli random variable $X$, it holds that $X^2 = X$.  As a result,  $\|\widetilde{A}^{s_m, e_m}(t)\|_{\mathrm{F}}^2$ cannot serve as a good estimator of $\|\widetilde{\Theta}^{s_m, e_m}(t)\|_{\mathrm{F}}^2$.  In practice, these two independent samples can be acquired by splitting the data into, say, odd and even time points.	 
	 	  In addition, every random interval $(s_m', e_m')$ provided to the algorithm is shrunk by a constant fraction of its original length.  This is done in order to avoid false positives around newly-found change points, a correction usually performed in  WBS-style algorithm: see, e.g., the parameter  $\delta$ used in Algorithm~3 in \cite{WangEtal2017} and the parameter $\beta$ used in Algorithm~4 \cite{wang2016high}.  Note that in our paper, however, the amount of shrinking does not depend on unknown quantities.
	
An interesting and possibly surprising feature of the NBS algorithm is that it merely relies on network CUSUM statistics -- weighted sample averages of adjacency matrices (see \Cref{def-cusum}) -- and does not rely on any  network or graphon estimation procedures, which are computationally costly. Though the NBS is not estimating any network parameters at all, 
it is still able achieve consistent network change point detection for a fairly large  class of models in a fast fashion.
In our next result we show that the NBS yields in fact a consistent estimator the change points.

\begin{theorem}\label{thm-1}
Assume the model described in \Cref{assume:model} and the condition of \Cref{assume:phase}. There exist absolute positive constants $C_R > 3/2$, $C_\beta$, $c_2 \in (0,1)$, $c$, $c_T$ and $C_1$ such that, letting $\{(\alpha_m,\  \beta_m) \}_{m=1}^M$  $\subset (0, T)$ be a collection of random intervals whose end points are drawn independently and uniformly from $\{1,\ldots, T\}$ and such that
		\begin{equation}\label{eq-thm1-statement-cr}
			\max_{m = 1, \ldots, M} (\beta_m -\alpha_m)\le C_R \Delta,
		\end{equation}  
		and
\begin{equation}\label{eq:input.parameter}
			C_\beta\rho n \log^{3/2}(T)< \tau <c_2\kappa_0^2n^2\rho^2 \Delta 
		\end{equation}		
guarantees that the collection of the estimated change points $\mathcal B=\{\hat \eta_k\}_{k=1}^{\widehat K}$	returned by the NBS procedure with input parameters $(0, T)$, $\{(\alpha_m,\beta_m)\}_{m=1}^M$ and $\tau$ will satisfy	
		\begin{align}\label{eq-thm1-result}
			& \mathbb{P}\Bigl\{  \widehat K =K ; \quad \max_{k=1,\ldots,K} |\eta_k-\hat \eta_k| \le  \epsilon \Bigr\} \ge  1 -\exp\left( \log\left(\frac{T}{\Delta}\right)-M\frac{\Delta}{4C_R T} \right) -  (6T^{3-c_T}+ 2T^{3-c}),
		\end{align}
where 
\begin{equation}\label{eq:epsilon.thm1}
			\epsilon = C_1 \log(T)\left( \frac{\sqrt \Delta }{\kappa_0n\rho } + \frac{ \sqrt{ \log(T)} }{\kappa_0^2n\rho}\right).
		\end{equation}
\end{theorem}

The constants in the theorem statement and their hierarchy of dependencies can be explicitly tracked in the proof; in particular, we require that the signal-to-noise ratio constant $C_\alpha$ in \Cref{assume:phase} to be sufficiently large. See the remark at the beginning of the proof of \Cref{thm-1} in \Cref{app:main-proof}.

To see how \Cref{thm-1} implies that the NBS is consistent according to \Cref{def:consistency}, we plug in the inequalities 
\[
\sqrt{\rho} \kappa_0 \geq \frac{C_\alpha \log^{1 + \xi} (T)}{\sqrt{n\Delta}} \quad \text{and} \quad \rho \geq \frac{\log (n)}{n},
\] 
stemming from Assumptions~\ref{assume:phase} and \ref{assume:model}, respectively, into the bound \eqref{eq:epsilon.thm1} on the localization error to get that
	\begin{align}
	\frac{\epsilon}{\Delta} & =C_1 \log(T)\left( \frac{\sqrt \Delta }{\kappa_0n\rho } + \frac{ \sqrt{ \log(T)}}{\kappa_0^2n\rho}\right) \frac{1}{\Delta}\nonumber \\
	& \leq C_1\left(\frac{1}{ C_{\alpha} \sqrt{\log (n)} \log^{\xi}(T) } + \frac{1}{C_{\alpha}^2 \log^{1/2 + 2\xi}  (T)}\right)  \rightarrow 0 \label{eq:slow.rate},
	\end{align}
	as	$T \to \infty$ (with all the remaining parameters also possibly changing in accordance to any scaling compatible with  \Cref{assume:phase}). The last expression also shows that, if $n$ diverges as $T$ grows unbounded, the parameter $\xi$ can be taken to be $0$ in \Cref{assume:phase} and consistent localization is still guaranteed.
More interestingly,  \eqref{eq:slow.rate}  continues to hold also when $n \asymp 1$, so that consistent localization is possible even when the number of nodes remains bounded. Of course, this is in striking contrast with the problem of consistent estimation of the edge probabilities -- or, more generally, of an underlying graphon -- which requires $n \rightarrow \infty$. 


We remark that, while \Cref{thm-1} shows that the NBS algorithm is consistent, we make no claim as to whether the localization rate is optimal.   In the next  section we will propose a two-step algorithm for change point localization that is not only consistent but nearly minimax rate-optimal under more favorable scalings on the parameters than the ones considered in \Cref{thm-1}. 


We conclude this section with few technical remarks on the assumptions of \Cref{thm-1}. In order for the NBS algorithm to be consistent, the threshold parameter $\tau$ needs to belong to an appropriate range: see \eqref{eq:input.parameter}. Such choice essentially guarantees that $\tau$ is both large enough to avoid false positives and small enough to never miss any true change points, both events occurring with high probability. Next,
the condition in \eqref{eq-thm1-statement-cr} requires that each of the random intervals fed to the NBS algorithm is not too large, compared to the minimal spacing parameter $\Delta$.  Without assuming \eqref{eq-thm1-statement-cr}, and using the trivial bound $C_R \leq T/\Delta$, it can be shown that the NBS will achieve a larger localization error of
\[
			\epsilon = C_1 \log(T)\left( \frac{\sqrt \Delta }{\kappa_0n\rho } + \frac{ \sqrt{ \log(T)} }{\kappa_0^2n\rho}\right) \left(\frac{T}{\Delta}\right)^2,
		\]
		under the scaling 
\[
			\kappa_0 \sqrt{\rho} \ge C_\alpha  \sqrt{\frac{1}{n\Delta}}  \log^{1 + \xi} (T) \sqrt{\frac{T}{\Delta}},
		\]
		which is stronger than the one in \Cref{assume:phase}. Assumption \eqref{eq-thm1-statement-cr} about the length of the random time intervals used as input to the algorithms is of somewhat technical nature, but it appears necessary to yield the localization error in \eqref{eq:epsilon.thm1}. Indeed, this condition,  or analogous ones requiring some knowledge of  $\Delta$, are commonly assumed  in the literature for change point localization to derive theoretical guarantees for WBS-style methods: see, e.g., \cite{fryzlewicz2014wild}, \cite{wang2016high},  \cite{wang2018univariate}, \cite{baranowski2019narrowest}, \cite{anastasiou2019detecting} and \cite{eichinger2018mosum}.
		Finally, the parameter $M$,  the number of random intervals used by the procedure, affects the results  through the probability lower bound in \eqref{eq-thm1-result}.  In order to guarantee that the probability tends to 1, one needs that 
		\[
			M \gtrsim \frac{T}{\Delta}\log\left(\frac{T}{\Delta}\right).
		\]

 \section{Optimal localization}\label{sec-local}

In the previous section we saw how the NBS algorithm  can consistently estimate the locations of the change points for the dynamic network model of \Cref{assume:model} under nearly any scaling for which this task is feasible, albeit possibly not in an optimal manner.  In this section, we are to show that under stronger, but still fairly general, conditions on both the model and the scaling, a two-step procedure that first applies the NBS and then refines the resulting estimators of  the locations of the change points, will achieve a minimax optimal localization rate.  The additional step beyond the NBS is named local refinement (LR) and is detailed in \Cref{algorithm:RI}. 


	
\begin{algorithm}[!ht]
	\begin{algorithmic}
		\INPUT Symmetric matrix $A \in \mathbb{R}^{n \times n}$, $\tau_2, \tau_3 > 0$.
		\State $(\kappa_i(A), v_i) \leftarrow $ the $i$th eigen-pair of $A$, with $|\kappa_1(A)| \geq \cdots |\kappa_n(A)|$
		\State $A' \leftarrow \sum_{i:|\kappa_i (A)|\ge\tau_2} \kappa_i(A) v_iv_i^{\top}$
		\State $\usvt(A, \tau_2,\tau_3) \leftarrow (A''_{ij})$ with
			\[
				(A'')_{ij} \leftarrow \begin{cases}
						(A')_{ij}, & \text{if} \quad |(A'_{ij})| \le \tau_3\\
					\sign ((A')_{ij})\tau_3, & \text{if} \quad |(A'_{ij})| > \tau_3\\
				\end{cases}
			\]
		\OUTPUT $\usvt(A, \tau_2,\tau_3)$.
		\caption{Universal Singular Value Thresholding. $\usvt(A, \tau_2,\tau_3)$}
		\label{algorithm:USVT}
	\end{algorithmic}
\end{algorithm} 

\begin{algorithm}[!ht]
	\begin{algorithmic}
		\INPUT $\{A(t)\}_{t=1}^{T}, \{B(t)\}_{t=1}^{T} \in \mathbb{R}^{n\times n}$, $\tau_2, \tau_3$, $\{ \nu_{k}\}_{k=1}^{K} \subset \{1, \ldots, T-1\}$, $ \nu_0 =1$, $\nu_{K+1} =T+1$.		
		\For{$k = 1, \ldots, K$}  
			\State $[s,e] \leftarrow [2^{-1}(\nu_{k-1} + \nu_{k}) , \, 2^{-1}(\nu_{k} + \nu_{k+1})] $
			\State $\widetilde{\Delta}_k \leftarrow \sqrt {\frac{(e-\nu_k)(\nu_k-s)}{e-s} }  $
			\State $\widehat \Theta_k \leftarrow \usvt(\widetilde B^{s , e} ( \nu_k) , \tau_2,\tau_3\widetilde{\Delta}_k)$
			\State $b_k \leftarrow \argmax_{s\le t \le e  } ( \widetilde A^{ s , e} (t),\widetilde \Theta_k)$
		\EndFor
		\OUTPUT $\{b_k\}_{k=1}^K$.
		\caption{Local Refinement}
		\label{algorithm:RI}
	\end{algorithmic}
\end{algorithm}


The LR algorithm takes as input  two identically distributed sequences of networks fulfilling \Cref{assume:model} (obtained for instance by sample splitting), along with a sequence $\{ \nu_k \}_{k=1}^K$ of initial change point estimates that are sufficiently close to the locations of the true change points, in way made precise in \eqref{eq:pre estimate} below. In particular, this preliminary estimates may be computed on the same data. The procedure then inspects all the triplets of consecutive change point estimators one at a time (with the time points $1$ and $T+1$ as two dummy change points, for notational consistency). For each such triplet, the LR utilizes the  universal singular value thresholding (USVT) algorithm \citep{Chatterjee2015} to construct a more accurate estimator of a local CUSUM matrix of the expected adjacency matrix at the middle point estimator. This estimator is in turn used to probe nearby locations in order to refine the original estimator of the location of the middle change point location. This results in a provably more precise estimator of that location.  From a computational standpoint,  \Cref{algorithm:RI} is parallelizable in the sense that we can deal with each $k\in \{1, \ldots, K\}$ separately. 

The signal-to-noise ratio conditions under which the LR improves upon the NBS are stronger than the ones that guarantee consistency of the latter, and are imposed in order to ensure that the USVT procedure is effective \citep[see, e.g.][]{Xu2017}. We formalize them next.


\begin{assumption}\label{assume:phase 2}
	Let $\{\Theta(t)\}_{t=1}^T$ be defined as in \Cref{assume:model}. For some $0 < r \leq n$,
		\[
			\max_{k=1,\ldots,K} \rank \left( \Theta (\eta_{k} ) -\Theta (\eta_{k}-1 ) \right) \le r. 
		\]
Furthermore, for a constant $C_\alpha>0$ and any $\xi>0$,
		\begin{equation}\label{eq:scaling.rank}
			  \kappa_0 \sqrt{\rho}  \ge C_\alpha \frac{\log^{1 + \xi} (T)}{\sqrt{\Delta}}\sqrt{\frac{r}{n}}.
		\end{equation}
\end{assumption}

The fixed quantity  $\xi>0$ in in the previous assumtion is required only for the case of $r \asymp n \asymp 1$ and can be set to zero in all other scenarios. The parameter $r$ controlling the maximal rank of the difference of consecutive expected adjacency matrices is, like all the other parameters, also allowed to change with $T$.
The first condition in \Cref{assume:phase 2} is about the model itself and requires that, in addition to all the properties listed in \Cref{assume:model}, the difference between any two different consecutive expected adjacency matrices is of low rank. 
Using the fact that,  for any matrices $A, B \in \mathbb{R}^{n \times n}$ of rank $r_1$ and $r_2$ respectively, it holds that 
	\[
		\rank(A - B) = \min\{r_1 + r_2, \, n\},
	\]
we see that \Cref{assume:phase 2} indirectly constraints the ranks of $\{\Theta(t)\}_{t=1}^T$. In particular, if $\Theta(\eta_k)$ and  $\Theta(\eta_{k-1})$ are the expected adjacency matrices of stochastic block models with $M_1$ and $M_2$ communities respectively, then $\rank \left( \Theta (\eta_{k} ) -\Theta (\eta_{k-1} ) \right) \le \min\{M_1 + M_2, \, n\}$. 

 \Cref{assume:phase 2} is compatible with  a broad range of parameter scalings. Focusing on the rank parameter, we highlight two extreme cases.
\begin{itemize}
\item  When $r \asymp 1$, the scaling \eqref{eq:scaling.rank} match the one in \Cref{assume:phase}. 
\item On the other hand, if the change points are far from each others so that $\Delta \asymp T$ and again $\kappa_0 \sqrt{\rho}  \asymp n^{-1/2}$, then as long as $r \lesssim  T\log^{-(2 + 2 \xi)} (T)$, then \Cref{assume:phase 2} holds.  This includes the situation where $T\log^{-(2 + 2 \xi)}(T) \geq n$, which essentially leaves the order of magnitude of $r$ unconstrained (though, of course, necessarily, $r \leq n$.)
\end{itemize}


\subsection{Upper and lower bounds on the localization error}
\label{subsection:LR upper}
The next theorem derives improved localization rates for the LR procedure under and is the main result of this section.

\begin{theorem}\label{theorem:localization 1}
Assume the model described in \Cref{assume:model} and the conditions of \Cref{assume:phase 2}.
There exist absolute positive constants $C$, $C_\epsilon$, $C_2$ and $C_3$ such that if $\{ \nu_k \}_{k=1}^K \subset (2, \ldots, T)$ is an increasing sequence satisfying
\begin{equation}\label{eq:pre estimate}
 \max_{k=1,\ldots,K} |\nu_k-\eta _k| <\Delta /6,
 \end{equation}
 then  the collection of the estimated change points
	$\mathcal B=\{\widehat  \eta_k \}_{k=1}^{ K}$	returned by the LR procedure with input
	parameters  $(0, T)$, $\{ \nu_k \}_{k=1}^K$,
\begin{align*}
		\tau_2 = (3/4)(C\sqrt{n\rho} + C_{\varepsilon}\log(T))
		\ \text { and } \
		\tau_3  =\rho,
	\end{align*}    is such that 
\begin{align*}
		& \mathbb{P}\Bigl\{   \max_{k=1,\ldots,K} |\eta_k-\widehat  \eta_k| \le \epsilon \Bigr\} \geq 1 - 2T^{3-3C_{\varepsilon}/4} - 4T^{3-3C_3^2/8},
	\end{align*}
	where
	\begin{equation}\label{eq:loc.rate.2}
\epsilon =  C_2 \frac{\log^2(T)}{ \kappa_0^{2}n^{2}\rho}.		
	\end{equation}

%
\end{theorem}

The proof of \Cref{theorem:localization 1} is given in \Cref{app:main-proof}. The values and dependence among the constants can be tracked throughout and, just like with \Cref{thm-1}, demand that the constant $C_\alpha$ in the signal-to-noise ratio condition \eqref{eq:scaling.rank} is chosen large enough. 

It is immediate to see that \Cref{theorem:localization 1} offers stronger consistency guarantees than \Cref{thm-1}. Indeed, using \Cref{assume:phase 2} along with the assumption that $\rho \geq \frac{\log (n)}{n}$, we see that the localization rate implied by \eqref{eq:loc.rate.2} is 
\begin{equation}\label{eq:rate.fast}
\frac{\epsilon}{\Delta} \leq \frac{1}{ C_\alpha^2 \log^{2 \xi}(T) r \log n } \rightarrow 0,
\end{equation}
as $T \rightarrow \infty$. This upper bound on the localization error is of smaller order than the one in \eqref{eq:slow.rate} afforded by \Cref{thm-1}. 
Furthermore, as remarked above, change point consistency is still guaranteed even as $n \asymp 1$. On the other hand, if $n$ is diverging in $T$, we may set $\xi = 0$ in  \Cref{assume:phase 2}.

To gain a further appreciation for the type of improvement  \Cref{theorem:localization 1} delivers over \Cref{thm-1}, assume that $r \asymp n$. Then, according to \Cref{thm-1}, in order for the NSB procedure to yield the same localization error as in \Cref{theorem:localization 1} it appears necessary to strengthen the signal-to-noise ratio requirement to be 
	\[
			\kappa_0\sqrt{n \rho \Delta} \gtrsim \sqrt{n} \log^{1 + \xi}(T)
		\]	
		instead of just $\kappa_0\sqrt{n \rho \Delta} \gtrsim \log^{1 + \xi}(T) $.

In addition to \Cref{assume:phase 2}, \Cref{theorem:localization 1} further requires that the sequence $\{ \nu_k \}_{k=1}^K$  of preliminary estimates used as an input to the  procedure to be within a constant fraction of $\Delta$ from the true change points; see \eqref{eq:pre estimate}. Notice that this assumption may be satisfied even if the ratio $\max_{k=1,\ldots,K} |\nu_k-\eta _k|$ is not a vanishing fraction of $\Delta$, thus failing to fulfill \Cref{def:consistency}. Of course, the change point estimators obtained using the NBS algorithm satisfy \eqref{eq:pre estimate} with high probability and for all large enough $T$, as demonstrated above in  \Cref{thm-1}, and therefore can be used as inputs to the LR algorithm.

Finally, the choices of threshold parameters $\tau_2$ and $\tau_3$ stem from the analysis of the USVT procedure for network estimation in \cite{Xu2017}. In particular, the parameter $\tau_2$ serves as a cutoff for the upper bound of the operator norm difference between the sample and population version of certain matrices of interest.

 \vskip 3mm
 
 In the second result of the section we prove that the localization rate demonstrated in \Cref{theorem:localization 1} is nearly minimax optimal, save for a term poly-logarithmic in $T$.


\begin{lemma}\label{lem-3.3-lower}
Let $\{A(t)\}_{t=1}^T$ be a sequence of independent inhomogeneous Bernoulli networks satisfying \Cref{assume:model} with $K = 1$ (i.e. there exists one and only one change point).  Let $P_{\kappa_0, \Delta, n, \rho}^T$
		denote the corresponding joint distribution. Consider the class of distributions 
		\[
			\mathcal Q = \left \{P_{\kappa_0, \Delta, n, \rho}^T: \kappa_0 \le 1/2, \, \rho \leq 1/2   \right\} .
		\]
	Then,
		\[
			\inf_{\hat \eta} \sup_{P \in \mathcal Q} \mathbb{E}_P (| \hat \eta -\eta |)\ge \max \{c\kappa_0^{-2}n^{-2}\rho^{-1}, 1/2\}.
		\]
\end{lemma}

The family of distributions $\mathcal{Q}$ allows for a wide range of changes. Indeed, the constraints that $\kappa_0 \leq 1/2$ is fairly general and,  in particular,  include the challenging scenario where all edge probabilities change at the change points.   The constant $1/2$ is arbitrary and can be replaced by any constant between 0 and 1.

 \subsection{Sparse stochastic block model}\label{section:SBM}
 
In \Cref{theorem:localization 1} we show that, for network models with rank constraints, combining  the NBS and the LR algorithms yields nearly optimal localization under the low rank assumption and the scaling described in \Cref{assume:phase 2}.  Low rank network models include a wide range of common network models, e.g. the Erd\H{o}s--R\'enyi random graph model \citep{ErdosRenyi1959}, stochastic block models \citep[e.g.][]{HollandEtal1983} and random dot product models \citep{YoungScheinerman2007}.  However, in  these models, it is  often also assumed that no self-loops are allowed, i.e.~the diagonal entries of the adjacency matrices are always $0$. As a result, the low rank assumption no longer holds. In this section we show that, for the case of stochastic block models,  this  issue can be overcome and that the guarantees of  \Cref{theorem:localization 1} hold also in this case.
For completeness, we include the definition of a sparse stochastic block model and some of its properties.  
 
\begin{definition}[Sparse Stochastic Block Model]\label{def-ssbm}
	A network is from a sparse stochastic block model with size $n$, sparsity parameter $\rho$, membership matrix $Z \in \{ 0,1\}^{n \times s}$ and connectivity matrix $Q \in [0, 1]^{r \times r}$, if the expected adjacency matrix  satisfies 
		\[
			\mathbb{E} (A)= \rho ZQZ^{\top}- \mathrm{diag}\bigl(\rho ZQZ^{\top}\bigr) .
		\]	
 		Each of the rows of the membership matrix $Z$ contains only one non-zero entry; moreover, $Z$ is a column full rank matrix, i.e.~$\mathrm{rank}(Z) =r $.  In particular, 
	$\mathrm{rank}(ZQZ^{\top}) \leq r$, with identity holding when $Q$ is a full rank matrix.
\end{definition}

In order to accommodate for the lack of self-loops we rely on a new set of conditions, described next. 

\begin{assumption}\label{assume:model 2}
	Let $\{A(t)\}_{t=1}^T \in \mathbb{R}^{n\times n}$ be a sequence of independent adjacency matrices satisfying the dynamic network model of \Cref{assume:model}.  
	Assume that, for all $k =1,\ldots,K$,  
		\[
			\Theta (\eta_{k} ) -\Theta (\eta_{k}-1 ) =\Lambda(k) -\diag(\Lambda(k)),
		\]
 		where $\Lambda(k) =Z_kQ_kZ_k^{\top}$,  $Z_k$ is a membership matrix such that $\rank (Z_k) \le r$, and $Q_k $ is a connectivity matrix.  
 		Furthermore, for a constant $C_\alpha>0$ and any $\xi>0$,
		\[
			\kappa_0 \sqrt{\rho}  \ge C_\alpha  \frac{\log^{1 +  \xi}(T)}{\sqrt{\Delta}}\sqrt{\frac{r}{n}}.
		\]
\end{assumption} 

 \Cref{assume:model 2} differs from \Cref{assume:phase 2} only in the how it constraints the difference of the expected adjacency matrices. Indeed, under \Cref{assume:model 2}, $\Theta(\eta_k) -\Theta(\eta_{k}-1) $ is typically not a low rank matrix, and therefore \Cref{assume:phase 2} would not hold.  Aside from this, the signal-to-noise condition is identical in the two sets of assumptions. 


Now, unlike in the problem of recovering the community assignment in a stochastic block model, where zeroing out the diagonal entries of the low rank matrix corresponding to the expected adjacency matrix is essentially inconsequential, in the localization problem this is not the case.  To see this, observe that if  the time interval $(s+1,\ldots,e)$ contains one change point $\eta_k$, then for $t \in (s+1,\ldots, e-1)$,
	\[
		\widetilde \Theta^{s, e} (t) = \begin{cases}
			\sqrt {\frac {t-s}{(e-s)(e-t)} }(e-\eta_k) (\Lambda(k) -\diag(\Lambda(k)) ), & \text{if } t\le \eta_k, \\
			\sqrt {\frac {e-t}{(e-s)(t-s)} }(\eta_k-s) (\Lambda(k) -\diag(\Lambda(k)))  , &  \text{if } t\ge\eta_k.
		\end{cases}
	\]
	In particular, at $t=\eta_k$,
	\[
		\left \| \sqrt {\frac {(t-\eta_k)(e-\eta_k) }{(e-s)} } \diag(\Lambda(k))  \right\|_{F} \lesssim \rho\sqrt{n} \sqrt {\min\{ e-\eta_k,\eta_k-s\}}  ,
	\]
a quantity that depends on the spacing between change points and may potentially be quite large. In order to handle such issue we make the following assumption.


\begin{assumption}\label{assume:SBM frobenius norm} 
For each $k=0,\ldots,K$,  set
		\[
			\Theta (\eta_{k} )  =\Gamma(k) -\diag(\Gamma(k)),
		\]
 		where $\Gamma(k) = Z'_kQ'_kZ_k^{'\top}$,  $Z'_k$ is a membership matrix and $Q'_k $ is a connectivity matrix.  
 	For an absolute constant $C_{\Gamma} > 0$, it holds that  
 	\[
 	\|\Gamma(k)\|_\mathrm{F} \ge C_{\Gamma} \| \diag (\Gamma(k))\|_{\mathrm{F}}.
 	\]
\end{assumption}

Since $\|\Gamma(k)\|_\mathrm{F}$ is of order no larger than $\rho n$ and $ \| \diag (\Gamma(k))\|_{\mathrm{F}}$ is of order no larger than $\rho\sqrt{n}$, overall \Cref{assume:SBM frobenius norm} is a mild condition. Of course, if $ \Gamma(k)$ is a diagonally-dominant matrix, then it is unclear how to estimate $\Gamma (k) $ because in the no-self-loop networks, the diagonals of the adjacency matrices are always $0$.  

\begin{theorem}\label{thm-sbm-main} 
	In \Cref{theorem:localization 1}, if \Cref{assume:phase 2} is replaced by \Cref{assume:model 2} and \Cref{assume:SBM frobenius norm}, then the same conclusion still holds.
\end{theorem}

The proof of \Cref{thm-sbm-main} can be found in \Cref{app:main-proof}.  The main difference between this proof and the proof of \Cref{theorem:localization 1} is the treatment on the diagonal entries under \Cref{assume:SBM frobenius norm}.

\section{Illustrative Simulations}\label{sec:simulations}
In this section we will present the results of various illustrative simulations intended to  corroborate the theory developed in the paper and to demonstrate the type of improvements  the LR delivers over the NBS.  As for this, we will use well-tuned tuning parameters, which will be reported. 

We point out that we could not find a methodology for the problem of multiple change point localization in network models with which to directly compare the NBS and LR. 
We have looked into existing methods for multiple change point localization that have been proposed for change point localization in settings different than dynamic network models, such as  the ones put forward in \cite{keshavarz2018sequential}, \cite{Cho2015}, \cite{cho2015multiple} and \cite{wang2016high}, among others. However, none of these procedures could be successfully deployed in the simulation settings described below. For this reason, we do not report the results of these comparison.


We consider the following three simulation settings.   All settings have equally-spaced change points, therefore the total number of time points $T = (K+1)\Delta$.\\
\noindent {\bf Setting (i).} We set $\Delta = 60, 80, 120, 200$, $K = 2$, $n = 150$ and $\rho = 0.02$.  Each network is generated from a balanced 3-community stochastic block model.  At the change points, the connectivity matrices are
	\[
		Q_1 = \rho \left(\begin{array}{ccc}
			0.6 & 1 & 0.6 \\
			1 & 0.6 & 0.5 \\
			0.6 & 0.5 & 0.6 
			\end{array}\right), \quad 
		Q_2 = \rho \left(\begin{array}{ccc}
			0.6 & 0.5 & 0.6 \\
			0.5 & 0.6 & 1 \\
			0.6 & 1 & 0.6 
			\end{array}\right) \quad \text{and} \quad
		Q_3 = Q_1,
	\]
	respectively.\\

\noindent	{\bf Setting (ii).} We set $\Delta = 60, 80, 120, 200$, $K = 2$, $n = 150$, $\rho = 0.015$ and the connectivity matrix be
	\[
		Q = \rho \left(\begin{array}{ccc}
			0.25 & 0.5 & 0.25 \\
			0.5 & 1 & 0.5 \\
			0.25 & 0.5 & 0.25
			\end{array}\right).
	\]
	Each network is generated from a balanced 3-community stochastic block model.  At the change points, membership are reshuffled randomly.\\

	\noindent	{\bf Setting (iii).} We set $\Delta = 80$, $K = 2$, $n = 150, 180, 210, 240$, $\rho = 0.01$.  Each network is generated from a balanced 3-community stochastic block model.  At the change points, the connectivity matrices are
	\[
		Q_1 = \rho \left(\begin{array}{ccc}
			0.9 & 0.8 & 0.3 \\
			0.8 & 0.3 & 0.3 \\
			0.3 & 0.3 & 0.3 
			\end{array}\right), \quad 
		Q_2 = \rho \left(\begin{array}{ccc}
			0.3 & 0.3 & 0.7 \\
			0.3 & 0.6 & 0.3 \\
			0.7 & 0.3 & 0.3  
			\end{array}\right) \quad \text{and} \quad
		Q_3 =  \rho \left(\begin{array}{ccc}
			0.3 & 0.3 & 0.3 \\
			0.3 & 0.3 & 0.6 \\
			0.3 & 0.6 & 0.1 
			\end{array}\right),
	\]
	respectively.

For each of the above  settings we simulated a dynamic network realization and applied both the NBS and LR 200 times.  In fact, we have applied a simplified version of the NBS algorithm based on the BS procedure \citep[see, e.g.][]{vostrikova1981detection} instead of WBS. Since the number of change points is small, it can be shown that the guarantees of \Cref{thm-sbm-main} hold true even for this simpler, computationally less demanding algorithm\footnote{In general, however, when the number of change points increases with $T$, BS is  sub-optimal  compared to WBS. Note that the default choices in R packages based on \cite{Cho2015}, \cite{cho2015multiple} and \cite{wang2016high} are all based on BS instead of WBS.}.  

To evaluate the performance of the algorithms, for each simulation we recorded
	\begin{itemize}
	\item $d(\widehat{S}, S)/T$, the  Hausdorff distance between the set of change point 	estimators and the set of the true change points, normalized by $T$,
	\item $|\widehat{K} - K|$, the absolute difference between the numbers of the change point estimators and the true change points,
	\item and Prop,  the proportion of simulations (out of 200) for which $\widehat{K} = K$.
	\end{itemize}

	Table~\ref{tab-set1-1} presents the results in the form of mean(standard error).  The columns labeled by sub. $d(\widehat{S}, S)/T$ displays the results only for the simulations in which  $\widehat{K} = K$.
	All the numerical analysis were conducted on machines with CPU Intel(R) Xeon(R) CPU E5-2670 0 @ 2.60GHz.	
	
Since the LR is a local refinement to the NBS, the columns corresponding to the LR algorithm report, by construction, the same $\widehat{K}$ and therefore the same correct proportion.  Due to \eqref{eq:pre estimate}, which requires the LR to be deployed only as a refinement of an estimator that returns the correct number of change points, in order to show the improvement afforded by the LR we only considered  simulations in which the NBS outputs  the correct number of change points.  
	


\begin{table*}[htbp]
\centering
\begin{tabular}{ccccccc}
	 & \multicolumn{2}{c}{$d(\widehat{S}, S)/T$} & $|\widehat{K} - K|$ & Prop & \multicolumn{2}{c}{sub. $d(\widehat{S}, S)/T$} \\
	 & NBS & LR & & & NBS & LR\\
	 \hline \hline
	 \multicolumn{7}{c}{Setting (i)} \\ [3pt]
	 $T = 180$ & 0.164(0.010) & 0.130(0.011) & 0.955(0.062) & 0.400 & 0.043(0.005) & 0.008(0.004) \\
	 $T = 240$ & 0.113(0.009) & 0.078(0.009) & 0.820(0.063) & 0.485 & 0.023(0.002) & 0.000(0.000) \\
	 $T = 360$ & 0.049(0.006) & 0.027(0.006) & 0.450(0.051) & 0.675 & 0.010(0.001) & 0.000(0.000) \\
	 $T = 600$ & 0.019(0.003) & 0.003(0.001) & 0.265(0.036) & 0.770 & 0.004(0.000) & 0.000(0.000) \\ [5pt]
	 \multicolumn{7}{c}{Setting (ii)} \\ [3pt]
	 $T = 180$ & 0.033(0.003) & 0.004(0.002) & 0.195(0.033) & 0.830 & 0.021(0.002) & 0.000(0.000) \\
	 $T = 240$ & 0.013(0.002) & 0.001(0.000) & 0.070(0.018) & 0.930 & 0.009(0.001) & 0.000(0.000) \\
	 $T = 360$ & 0.006(0.001) & 0.001(0.000) & 0.070(0.018) & 0.930 & 0.003(0.000) & 0.000(0.000) \\
	 $T = 600$ & 0.002(0.000) & 0.000(0.000) & 0.055(0.016) & 0.945 & 0.001(0.000) & 0.000(0.000) \\ [5pt]
	 \multicolumn{7}{c}{Setting (iii)} \\ [3pt]
	 $n = 150$ & 0.115(0.010) & 0.095(0.010) & 0.415(0.038) & 0.610 & 0.029(0.004) & 0.014(0.005)\\
	 $n = 180$ & 0.027(0.003) & 0.008(0.003) & 0.250(0.034) & 0.775 & 0.012(0.001) & 0.000(0.000)\\
	 $n = 210$ & 0.013(0.002) & 0.000(0.000) & 0.165(0.027) & 0.840 & 0.004(0.001) & 0.000(0.000)\\
	 $n = 240$ & 0.013(0.002) &	0.000(0.000) & 0.165(0.026) & 0.835 & 0.002(0.000) & 0.000(0.000)
\end{tabular}
 \caption{Simulation results for both the NBS and LR. \label{tab-set1-1}	}
\end{table*}

As for tuning parameters, recall that we have one tuning parameter $\tau_1$ for the NBS, and two tuning parameters, $\tau_2$ and $\tau_3$, for the LR. The choices of  tuning parameters in these three different settings are given in Table \ref{tab-tuning}, where Inf is equivalent to no entrywise truncation in the USTV step, and $M$ is the number of communities in the stochastic block model.  In selecting the tuning parameters we have used the true number of communities $M$; of course, in practice, this quantity needs to be estimated from the data \citep[e.g.][]{chen2018network}. Finally we estimate $\rho$ using $\hat{\rho}$, defined to be the 95\% quantile of  
	\[
		\left\{T^{-1}\sum_{t = 1}^T A_{ij}(t), \, 1 \leq i, j \leq n \right\}.
	\]

\begin{table*}[htbp]
\centering
\begin{tabular}{cccc}
	Setting  & $\tau_1$ & $\tau_2$ & $\tau_3$ \\
	\hline \hline 
	(i) & $n\hat{\rho}\log^2(T)/21$ & $Mn\hat{\rho}$ & $\hat{\rho}$ \\
	(ii) & $n\hat{\rho}\log^2(T)/20$ & $Mn\hat{\rho}$ & Inf \\
	(iii) & $3n\hat{\rho}/4$ & $Mn\hat{\rho}$ & Inf 
\end{tabular}	
\caption{Tuning parameter choices.  \label{tab-tuning}}	
\end{table*}

It can be seen from Table~\ref{tab-set1-1} that,	with these choices of the tuning parameters, the performance of both the NBS and LR improves as $T$, $n$ and $\rho$ increase.  In addition, the LR significantly outperforms the NBS. 

For all the settings described above, we have also conducted additional simulations with an omnibus default choice for the tuning parameter which does not require knowledge of $M$ : $\tau_1 = n\hat{\rho} \log^2(T)/20$. The results are shown in Table~\ref{tab-ex1-2}.  Due to the default choice of the tuning parameter, it is not easy to show how the performance changes with different model parameters. Therefore we only collect the NBS results to demonstrate that we can achieve good performances in terms of $|\widehat{K} - K|$, $d(\widehat{S}, S)/T$ and Prop, with easily chosen tuning parameter.

\begin{table*}[h]
\centering
\begin{tabular}{ccccc}
	 & $d(\widehat{S}, S)/T$ & $|\widehat{K} - K|$ & Prop & Time(second/repetition)\\
	 \hline \hline
	 \multicolumn{5}{c}{Setting (i)} \\ [3pt]
	 $T = 180$ & 0.166(0.010) & 1.025(0.062) & 0.360 & 1.607(0.030) \\
	 $T = 240$ & 0.121(0.010) & 0.760(0.061) & 0.520 & 3.104(0.055) \\
	 $T = 360$ & 0.042(0.006) & 0.285(0.044) & 0.805 & 7.126(0.060) \\
	 $T = 600$ & 0.011(0.002) & 0.125(0.023) & 0.875 & 20.837(0.149) \\ [5pt]
	 \multicolumn{5}{c}{Setting (ii)} \\ [3pt]
	 $T = 180$ & 0.332(0.000) & 0.970(0.012) & 0.030 & 1.061(0009) \\
	 $T = 240$ & 0.444(0.000) & 0.955(0.015) & 0.045 & 2.032(0.028) \\
	 $T = 360$ & 0.667(0.000) & 0.985(0.009) & 0.015 & 3.950(0.023) \\
	 $T = 600$ & 1.111(0.000) & 1.000(0.000) & 0.000 & 10.994(0.022) \\ [5pt]
	 \multicolumn{5}{c}{Setting (iii)} \\ [3pt]
	 $n = 150$ & 0.154(0.013) & 0.415(0.038) & 0.610 & 1.861(0.004)\\
	 $n = 180$ & 0.050(0.006) & 0.255(0.035) & 0.770 & 2.683(0.010)\\
	 $n = 210$ & 0.015(0.002) & 0.195(0.032) & 0.825 & 4.936(0.021)\\
	 $n = 240$ & 0.009(0.001) & 0.210(0.033) & 0.815 & 8.785(0.039)
\end{tabular}
\caption{Simulation results for the  NBS  with a default tuning parameter. \label{tab-ex1-2}}	
\end{table*}

It can be seen in \Cref{tab-ex1-2} that with this default choice of tuning parameter, the NBS is still producing good results.

\section{Discussion}\label{sec-disc}
We have studied the change point localization problem in sparse  dynamic network settings. We have proposed two computationally-efficient algorithms based on CUSUM statistics: Network Binary Segmentation (NBS) and Local Refinement (LR).  The NBS is able to localize multiple change points consistently under virtually all parameter scalings for which this task is feasible. The LR guarantees sharper localization errors under slightly stronger scalings and is nearly minimax rate-optimal under those scalings. Our results are applicable to a wide class of dynamic network models and, in particular to the ones assuming a sequence of time-varying stochastic block models.

While we are  able to demonstrate a nearly optimal localization procedure only under a certain low rank assumption (see \Cref{assume:phase 2}), it remains an open problem to design a computationally efficient algorithm that is provably optimal across all scalings  for which consistent localization is possible, described in \Cref{assume:phase}.



The assumptions used in this paper can be possibly generalized in a few directions.  If one wishes to relax the independence across time and/or within networks, or replace the Bernoulli assumption with other distributional assumptions (e.g.~sub-Gaussian), then it will be necessary to change in the proofs of                                                                                                                                                                                                                                                                                                                                                                                                                                                                                                                                                                                                                                                                                                                                                                                                                                                                                                                                                                                                                                                                                                                                                                                                                                                                                                                                                                                                                                                                                                                                                                                                                                                          the concentration inequalities and the corresponding large probability events. This in turn may lead to different scaling requirements for consistency and optimality, as well as possibly different localization error bounds.


It is worth noting that, assuming a stochastic block model at each time point, replacing the USVT algorithm used in the LR  procedure with an NP-hard graphon-based algorithm  \citep[see, e.g.][]{pensky2016dynamic,GaoEtal2015} will produce the nearly optimal rate \eqref{eq:lr.rate} under the scaling
\begin{equation}\label{eq-intro}
\rho \kappa_0^2 \gtrsim  \frac{\log^{2 + 2 \xi}(T)}{\Delta} \frac{\bigl(1 + r^2/n\bigr)}{n},
\end{equation}
which is  weaker than the scaling we assume for our polynomial time algorithms (NBS and LR), namely \eqref{minimax.scaling}.   
%
Equations~\eqref{eq:phase.consistent}, \eqref{minimax.scaling} and \eqref{eq-intro} reveal that 
\begin{itemize}
\item [(i)]	in the very sparse regime, i.e.~$r \lesssim \sqrt{n}$, there is no gap between the scaling  \eqref{eq-intro}  required by NP-hard algorithms and the scaling \eqref{eq:phase.consistent};
\item [(ii)] in the moderately sparse regime, i.e.~$\sqrt{n} \lnsim r \lnsim n$, then there is a gap between statistical and computational limits;
\item [(iii)] in the very dense regime, i.e.~$r \asymp n$, \eqref{minimax.scaling} and \eqref{eq-intro} are the same, which means NP-hard algorithms are not gaining over polynomial methods.
\end{itemize}
These observations is consistent with similar phenomena observed in other statistical problems, see e.g.~\cite{zhang2012communication}, \cite{loh2013regularized}, to name but a few.

To summarize, we have the following \Cref{tab-summary}.
\begin{table*}[htbp]
\centering
\begin{tabular}{ccc}
\hline
Rate & Scaling & Algorithm\\
	\hline
	$\epsilon/T = o(1)$ & $\rho \kappa^2_0 \gtrsim   \frac{ \log^{2 + 2 \xi}(T)}{\Delta}\frac{1}{n}$ & Poly \\ [3pt]
	\hline
	$\epsilon/T = \epsilon_{\mathrm{opt}}\frac{\log^2(T)}{T} $ & $\rho \kappa_0^2  \gtrsim \frac{\log^{2 + 2 \xi} (T) }{\Delta}\frac{r}{n}$  & Poly\\
	& $\rho \kappa_0^2 \gtrsim  \frac{\log^{2 + 2\xi}(T)}{\Delta} \frac{\bigl(1 + r^2/n\bigr)}{n}$ & NP \\[3pt]
	\hline
\end{tabular}	
\caption{Summary of our rates results.  \label{tab-summary}}	
\end{table*}

\section*{Acknowledgments} We would like to thank an anonymous reviewer and the associate editor for constructive comments that led to improvements in the presentation of the paper.

\bibliographystyle{ims}
\bibliography{citation} 

\appendix

\section{Proofs of Theorems 1, 2 and 3}
\label{app:main-proof}

For simplicity, we set
\[
		\|\Theta (\eta_{k} ) -\Theta (\eta_{k}-1 ) \|_{\mathrm{F}} = \kappa_k > 0, \text{ for any }  k=1, \ldots, K.
\]

\begin{proof}[Proof of \Cref{thm-1}]

The value of the constants in statememt of the theorem can be tracked in the proof.  
The hierarchy can be abstracted as follows: first, $c$ and $c_T$ are chosen such that \eqref{eq-thm1-result} tends to 0 as $T \to \infty$; then, $C_{\beta}$ can be chosen depending on $c$ and $c_T$;  the constant $c_2$ therefore depends on $C_{\beta}$ and $C_{\alpha}$; finally, a sufficiently large $C_1 > 0$ is chosen and depends on all the aforementioned constants and $C_R$.   In particular, increasing $C_{\alpha}$ would decrease the lower bound of $C_1$. 

	As the random intervals $\{(\alpha_m,\beta_m)\}_{m=1}^M$ are generated independently from the data, we will assume throughout the proof that the event $\mathcal M$ defined in \eqref{event-M} in \Cref{app:ancillary} holds.  By \Cref{lemma:random interval}, the probability of the complementary event is smaller than 
	\[
\exp\left\{\log\left(\frac{T}{\Delta}\right) - \frac{M\Delta}{4C_R T} \right\},
	\]
	which vanishes provided that 
	\[
	M \gtrsim \bigl(T/\Delta\bigr) \log\bigl(T/\Delta\bigr).
	\]


	For $0\leq s < t < e \leq T$, we consider the event
		\begin{equation}\label{eq-event-A}
		\mathcal A (s,t,e) = \left \{ \left |  (\widetilde A^{s,e}(t) , \widetilde B^{s,e}(t) )   -\|\widetilde \Theta^{s,e}(t)  \|_{\mathrm{F}}^2 \right | \le  C_{\beta}\log(T) \left( \| \widetilde \Theta^{s,e}(t) \|_\mathrm{F} + \log^{1/2}(T)\rho n\right) \right\}. 
		\end{equation}
		Due to \Cref{lemma:exponential for sparse bernulli}, it holds that $\mathbb{P}(\mathcal A(s,t,e)^c) 	\le 6T^{-c_T}+ 2T^{-c}$ for some $c, c_T > 3$, and, by a union bound argument, 
		\[
		\mathbb{P}(\mathcal A) = \mathbb{P}\left (\bigcup_{1\le s \le t \le e \le T } \mathcal A(s,t,e) \right )\ge 1 - (6T^{3-c_T}+ 2T^{3-c})  .
		\]
		All the analysis in the rest of this proof is conducted on the event $\mathcal{A}\cap \mathcal{M}$.

	\vskip 3mm
	The general strategy of the proof is to utilize a standard induction-like argument that is commonly used in proving the consistency of change point estimators; see, e.g. \cite{fryzlewicz2014wild}, \cite{wang2016high} and \cite{WangEtal2017}. Of course the specific details and technicalities of this argument  are new and challenging in our problem. In a nutshell, we will show that,
on the event $\mathcal{A}\cap \mathcal{M}$ and assuming that the algorithm has not made any mistakes so far in the detection and localization of change points, the procedure will also correctly identify any undetected change point and estimate its location within an error of $\epsilon$, if such an undetected change point exists.
Towards that end, it suffices to consider any generic time interval $(s, e) \subset (0, T)$ that satisfies 
		\[
			\eta_{r-1} \le s\le \eta_r \le \ldots\le \eta_{r+q} \le e \le \eta_{r+q+1}, \quad q\ge -1
			\]
			and
			\[
			\max \{ \min \{ \eta_r-s ,s-\eta_{r-1}  \}, \min \{ \eta_{r+q+1}-e, e-\eta_{r+q}\} \} \le \epsilon,
		\]
		where $q = -1$ indicates that there is no change point contained in $(s, e)$ and $\epsilon$ is given in \eqref{eq:epsilon.thm1}.

	Observe that 
		\begin{align} \label{eq:MWBS bound of epsilon}
			\epsilon =C_1 \log(T)\left( \frac{\sqrt \Delta }{\kappa_0n\rho } + \frac{\log^{1/2}(T)}{\kappa_0^2n\rho}\right) \le C_1\left(\frac{\Delta}{ C_{\alpha} \log^{1/2} (n) \log^{\xi} (T) } + \frac{\Delta}{C_{\alpha}^2\log^{1/2 + 2 \xi} (T)}\right), 
		\end{align}
		where the inequality follows from \Cref{assume:model} part {\it 1.} and \Cref{assume:phase}. 	
	Therefore, using the previous bound, 
		\[
		\epsilon \le 2C_1 \Delta \max \left\{ \frac{1}{ C_{\alpha} \log^{1/2} (n) \log^{\xi} (T) }, \frac{1}{C_{\alpha}^2\log^{1/2 + 2 \xi} (T)} \right\} \le  \Delta /4,
		\]
		by appropriately assuming $C_\alpha$ to be large enough.  It, therefore, has to be the case that, for any change point $\eta_p\in (0, T)$, either $|\eta_p -s |\le \epsilon$ or $|\eta_p-s| \ge \Delta - \epsilon \geq 3\Delta /4 $. This means that $ \min\{ |\eta_p-e|, |\eta_p-s| \}\le \epsilon$ indicates that $\eta_p$ is a change point that has been previously detected and estimated within an error of magnitude $\epsilon$ in the previous induction step, even if $\eta_p\in (s, e)$.  Below we will say that a change point  $\eta_p$ in $[s,e]$ is undetected if $ \min\{ \eta_p -s, \eta_p-e\} \ge 3\Delta/4 $.

	In order to complete the induction step, it suffices to show that NBS$((s, e), \{(\alpha_m,\beta_m)\}_{m=1}^M, \tau)$ (i) will not detect any new change point in $(s,e)$ if all the change points in that interval have been previously detected, and (ii) will find a point $b$ in $(s,e)$  such that $|\eta_p-b|\le \epsilon$ if there exists at least one undetected change point in $(s, e)$.

	\vskip 3mm
	\noindent{\bf Step 1.}
	Suppose that there does not exist any undetected change points within $(s, e)$. Then, for any $(s_m', e_m') = (\alpha_m, \beta_m) \cap (s, e)$, one of the following situations must hold:
	\begin{itemize}
		\item [(a)]	there is no change point within $ (s_m', e_m')$;
		\item [(b)] there exists only one change point $\eta_{r} $ within $(s_m', e_m')$ and $\min \{ \eta_{r}-s_m' , e_m'-\eta_{r}\}\le \epsilon $ or
		\item [(c)] there exist two change points $\eta_{r} ,\eta_{r+1}$ within $(s_m', e_m')$ and $\max \{ \eta_{r}-s_m' , e_m'-\eta_{r+1}\}\le \epsilon $.
	\end{itemize}
	
	We will analyze situation (c) only, as the other two cases are similar and in fact simpler.  Observe that if (c) holds, then  by \eqref{eq:MWBS bound of epsilon} and \eqref{eq:input.parameter},
		$$\epsilon \le 64^{-1}\Delta \le 64^{-1}(e_m'-s_m') ,$$
		where the second inequality is fulfilled by choosing a  sufficiently large $C_\alpha $.  Therefore,  the interval $$[s_m,e_m]= [s_m' +64^{-1} (e_m'-s_m') ,  e_m'- 64^{-1} (e_m'-s_m')], $$
		contains no change points.  To see this, notice that, on the event $\mathcal A$, $\widetilde \Theta^{s_m,e_m}(t) = 0$  for all $t \in (s_m, e_m)$, as there is no change point in $ [s_m,e_m]$. Furthermore, by \Cref{lemma:exponential  for sparse bernulli},  there exists a large enough constant $C_{\beta} > 0$ such that 
		$$\max_{s_m  < t < e_m }  ( \widetilde A^{s_m,e_m}( t),\widetilde B^{s_m,e_m} (t)) \le C_\beta \rho n\log^{3/2}(T).$$
	Thus, with the input parameter $\tau$ satisfying 
		\[
		\tau \geq C_{\beta}\rho n \log^{3/2}(T),
		\] 
		we conclude that NBS$((s, e), \{(\alpha_m,\beta_m)\}_{m=1}^M,
	\tau)$  will always correctly reject the existence of undetected change points.

	\vskip3mm
	\noindent{\bf Step 2.} 	Suppose now that there exists a change point $\eta_p\in (s, e)$ such that $ \min\{ \eta_p -s, \eta_p-e\} \ge 3\Delta/4$.  Let $a_m, b_m$ and $m*$ be defined as in NBS$((s, e), \{(\alpha_m,\beta_m)\}_{m=1}^M, \tau)$.  On the event $\mathcal M$, for any $\eta_p\in (s, e)$ such that $ \min\{ \eta_p -s, e - \eta_p\} \ge 3\Delta/4$, there exists an interval $[s_m', e_m'] $ containing only one change point $ \eta_p$ such that 
		\[
		\eta_p -3 \Delta/4 \le s_m'  \le  \eta_p - \Delta/2 \quad \text{and} \quad  \eta_p + \Delta/2 \le e_m' \le  \eta_p +3 \Delta/4. 
		\]
		Therefore, if $[s_m,e_m] =[s_m' +64^{-1} (e_m'-s_m') , e_m'- 64^{-1} (e_m'-s_m')]$, then one has that
		\begin{equation}\label{eq:sm}
			\eta_p - \Delta 3/4 \le s_m  \le  \eta_p - \Delta/8 \quad \text{and} \quad  \eta_p + \Delta/8 \le e_m  \le  \eta_p + \Delta 3/4. 
		\end{equation}
		Next, on the event $\mathcal A$, it holds that
		\[
			( \widetilde A^{s_m,e_m}( \eta_p),\widetilde B^{s_m ,e_m } ( \eta_p)) \ge   \| \widetilde \Theta^{s_m   ,e_m    } (\eta_p)\|_{\mathrm{F}}^2 -   C_\beta \log(T)( \log^{1/2}(T)\rho n + \| \widetilde \Theta^{s_m    ,e_m   } (\eta_p)\|_{\mathrm{F}}).
		\]
		It then follows from \Cref{lemma:one change point population size} that
		\[
			\| \widetilde \Theta^{s_m,e_m} (\eta_p)\|_{\mathrm{F}}^2  = \frac{ (\eta_p - s_m)(e_m - \eta_p)  }{e_m - s_m  } \kappa_p^2 \ge \min\{ e_m - \eta_p, \eta_p - s_m \}\kappa_p^2 \ge \kappa_p^2\Delta/8, 
		\]
		where the last inequality stems from \eqref{eq:sm}. Thus, due to \Cref{assume:model} part{\it 1.} and \Cref{assume:phase}, we conclude that
		\[
			 \kappa^2_p\Delta/16 \geq \kappa^2_0n^2\rho^2\Delta/16 \geq C_{\alpha}^2/16 n\rho \log^{2+2\xi}(T) > C_{\beta}n\rho\log^{3/2}(T),
		\]
		and
		\begin{equation}\label{eq-step2-2}
		\kappa_p\sqrt{\Delta}/4 \geq \kappa_0n\rho\sqrt{\Delta}/4 \geq C_{\alpha}/4\sqrt{n\rho}\log^{1+\xi}(T)\geq C_{\alpha}/4\log^{1/2}(n)\log^{1+\xi}(T) > 2C_{\beta}\log(T),
		\end{equation}
provided that, for $n, T \geq 2$,
		\begin{equation}\label{eq-c-beta-c-alpha}
			C_{\beta} < \min\left\{8^{-1}C_{\alpha}\log^{\xi}(T)\log^{1/2}(n), \, C_{\alpha}^2/16\log^{1/2 + 2\xi}(T)\right\}.
		\end{equation}
		We remark that as for the hierarchy of all the absolute constants involved, \eqref{eq-c-beta-c-alpha} is a constraint on $C_{\alpha}$.
		Thus, with a large enough $C_{\alpha}$, there exists an absolute constant $c_2 > 0$, such that
		\[
			( \widetilde A^{s_m,e_m}( \eta_p),\widetilde B^{s_m,e_m} ( \eta_p)) \ge  c_2 \kappa_p^2 \Delta.
		\]
		By the definition of $m^*$, one then obtain the inequality
		\begin{align}\label{eq:consistent MWBS size of one change point 2} 
			a_{m*} =( \widetilde A^{s_{m*},e_{m*}} (b_{m*}),\widetilde B^{s_{m*},e_{m*}} (b_{m*})) \ge c_2 (\kse)^2 \Delta,
		\end{align}
		where $\kse= \max\{\kappa_k: \,\min\{ \eta_p -s , e -\eta_p \} \ge 3\Delta /4\}$. Thus, with input parameter $\tau$ satisfying
		\[
		\tau < c_2\kappa_0^2n^2\rho^2\Delta,
		\] 
		The NBS can consistently detect the existence of undetected change points.

	\vskip 3mm
	\noindent{\bf Step 3.}  Assume next that there exists at least one undetected change point $\eta_p\in (s, e)$ such that $ \min\{ \eta_p -s, \eta_p-e\} \ge 3\Delta/4$.  Let $a_m, b_m$ and $m*$ be defined as in \Cref{algorithm:MWBS}.

	To complete the induction step and therefore the proof, it suffices to show that there exists a (necessarily undetected) change point $\eta_p\in[s_{m*},e_{m*}]$ such that 
\begin{equation}\label{eq:first}
		 \min\{ \eta_p -s, \eta_p-e\} \ge 3\Delta/4
\end{equation}
and that
\begin{equation}\label{eq:second}
|b_{m*}-\eta_p|\le \epsilon.
\end{equation}
In this step we will prove that \eqref{eq:first} holds.
	Denote 
		\[
			[s_{m*},e_{m*}]= [s_{m*}' +64^{-1} (e_{m*}'-s_{m*}') , e_{m*}- 64^{-1} (e_{m*}'-s_{m*}')]. 
		\]	
		Suppose for the sake of contradiction that
		\begin{align}\label{eq: size of population with small spacing} \max_{s_{m*} < t < e_{m*} }\| \widetilde \Theta ^{s_{m*},e_{m*}}(t)\|_{\mathrm{F}}^2  < c_2(\kse)^2  \Delta/2. 
		\end{align}
		Then
		\begin{align*}
			& \max_{s_{m*} < t < e_{m*}  }( \widetilde A^{s_{m*},e_{m*}}( t),\widetilde B^{s_{m*} ,e_{m*} } ( t))  \\
			&\le   \max_{s_{m*} < t < e_{m*}  } \| \widetilde \Theta ^{s_{m*}   ,e_{m*}    } (t)\|_{\mathrm{F}}^2 +   C_\beta \log (T)( \log^{1/2}(T)\rho n +\max_{s_{m*} < t < e_{m*}  } \| \widetilde \Theta^{s_{m*}    ,e_{m*}   } (t)\|_{\mathrm{F}}), \\
			& \leq c_2(\kse)^2 \Delta/2+   C_\beta \log^{3/2}(T)\rho n+   C_\beta \log (T) \sqrt {c_2 /2} \kse\sqrt \Delta \\
			&< c_2(\kse)^2 \Delta/2 +c_2(\kse)^2 \Delta/4 +c_2(\kse)^2 \Delta/4=  c_2(\kse)^2 \Delta,
		\end{align*}
		where the first inequality is due to the definition of the event $\mathcal A$, the second inequality follows from \eqref{eq: size of population with small spacing} and the third  inequality from \Cref{assume:phase}, for an appropriately large $C_\alpha$.  This contradicts  \eqref{eq:consistent MWBS size of one change point 2}. Therefore
		\begin{align}\label{eq: size of population with large spacing} 
			\max_{s_{m*} < t < e_{m*} }\| \widetilde \Theta ^{s_{m*},e_{m*}}(t)\|_{\mathrm{F}}^2  \ge  c_2(\kse)^2  \Delta/2. 
		\end{align} 

	Observe that if $[s_{m*},e_{m*}] $ contains two change points, then $ e_{m*}-s_{m*} \ge \Delta$ and if $ [s_{m*},e_{m*}] $ contains one change point $\eta$, then  it has to be the case that $\min\{ \eta -s_{m*} , e_{m*} -\eta \} \ge  c_2  \Delta/2 $, as otherwise by \Cref{lemma:one change point population size}, 
		\[
			\max_{s_{m*} < t <  e_{m*} }\| \widetilde \Theta^{s_{m*},e_{m*}}(t)\|_{\mathrm{F}}^2 =	\|\widetilde \Theta^{s_{m*},e_{m*}}(\eta) \|_{\mathrm{F}}^2    \le c_2(\kse)^2  \Delta/2,
		\]
		which contradicts \eqref{eq: size of population with large spacing}. 

	Therefore, since $e_{m*} -s_{m*} \ge  c_2\Delta/2$,
	the bound \eqref{eq:MWBS bound of epsilon} implies that 
		\begin{equation}\label{eq-derive-delta}
			\epsilon \le C_1\left(\frac{\Delta}{C_{\alpha} \log^{1/2} (n) \log^{\xi} (T) } + \frac{\Delta}{C_{\alpha}^2\log^{1/2 + 2 \xi} (T)}\right)   \le 64^{-1} (e_{m*}'-s_{m*}'),
		\end{equation}
		where the second inequality follows if  $C_\alpha$ is sufficiently large. By a similar argument as in {\bf Step 1}, $[s_{m*}, e_{m*}] $ contains no detected  change points.  Observe that by \eqref{eq:consistent MWBS size of one change point 2}, $[s_{m*},e_{m*}]$ contains at least one undetected change point. 
	
	\vskip 3mm
	\noindent {\bf Step 4.}
	In the final step of the proof we will show that \eqref{eq:second} occurs. To that end, we will apply \Cref{lemma:wbs multiple d}. Let
		\begin{equation}\label{eq-lambda-S4}
			\lambda = \max_{s_{m*} <  t <  e_{m*}} \bigl|( \widetilde A^{s_{m*},e_{m*}}(t) ,\widetilde B^{s_{m*},e_{m*}}  (t))    -\|\widetilde \Theta^{s_{m*},e_{m*}}(t) \|_{\mathrm{F}}^2 \bigr|.
		\end{equation}
		Observe that \eqref{eq: size of population with large spacing} and \eqref{eq-step2-2} imply that 
		\[
			c_3 \max_{s_{m*} < t < e_{m*} }  \| \widetilde  \Theta ^{s_{m*} , e_{m*}} (t) \|_{\mathrm{F}}^2 /2> C_\beta \log( T) \max_{s_{m*}< t < e_{m*} }  \| \widetilde   \Theta ^{s_{m*},e_{m*} } (t )\|_{\mathrm{F}},
		\]
		and
		\[
			c_3 \max_{s_{m*} < t < e_{m*} }  \| \widetilde  \Theta ^{s_{m*} , e_{m*}} (t) \|_{\mathrm{F}}^2 /2> C_{\beta}\log^{3/2}(T)\rho n,
		\]
		for a sufficiently large $c_3 > 0$.  Then, due to the definition of the event $\mathcal A$,
		\begin{equation}\label{eq-wbs-noise-proof}
			\lambda \le  C_\beta \log(T)\left( \log^{1/2}(T)\rho n + \max_{s_{m*} <  t <  e_{m*}}  \| \widetilde \Theta^{s_{m*}    ,e_{m*}   } (t)\|_{\mathrm{F}}\right)\le c_3 \max_{s_{m*} < t <  e_{m*}}  \| \widetilde\Theta ^{s_{m*}    ,e_{m*}   } (t)\|^2_{\mathrm{F}}.
		\end{equation}
		Since \eqref{eq:wbs size of sample} follows from \eqref{eq:consistent MWBS size of one change point 2}, \eqref{eq:mcusum noise} follows from \eqref{eq-lambda-S4}, and \eqref{eq:wbs noise} follows from  \eqref{eq-wbs-noise-proof}, all the conditions in \Cref{lemma:wbs multiple d} hold.  \Cref{lemma:wbs multiple d} implies that there exists an undetected change point $\eta_p$  within $[s,e]$ such that 
		\[
			|\eta_{k} -b |\le \frac{C_3\Delta\lambda}{ \| \widetilde \Theta^{s_{m*},e_{m*}} (\eta_k)\|_{\mathrm{F}}^2}  \quad \text{and} \quad  \|\widetilde \Theta^{s_{m*},e_{m*}} (\eta_k )  \|_{\mathrm{F}}^2 \ge c'\max_{s_{m*} \le t \le e_{m*} } \| \widetilde \Theta^{s_{m*},e_{m*}} (t)\|_{\mathrm{F}}^2. 
		\]
		and this combining  with \eqref{eq: size of population with large spacing} provides that
		\[
		|\eta_k - b|\leq \frac{2C_3C_{\beta}}{c_2(c')^2}\frac{\log^{3/2}(T)}{\kappa_0^2n\rho} + \frac{\sqrt{2}C_3C_{\beta}}{c'\sqrt{c_2}}\frac{\sqrt{\Delta}\log(T)}{\kappa_0n\rho} \leq C_1\log(T)\left(\frac{\log^{1/2}(T)}{\kappa_0^2n\rho} + \frac{\sqrt{\Delta}}{\kappa_0n\rho}\right),
		\]
		where $C_1 > \frac{2C_3C_{\beta}}{c_2(c')^2} + \frac{\sqrt{2}C_3C_{\beta}}{c'\sqrt{c_2}}$ and $c' < 2\log(2)C_{\beta}/c_3$.  This completes the induction.		
\end{proof}

 \begin{proof}[Proof of \Cref{theorem:localization 1}]
The dependence among the constants involved in \Cref{theorem:localization 1} is as follows.  Firstly, $C$ and $C_{\varepsilon}$ are chosen to guarantee that $2T^{3-3C_{\varepsilon}/4} \to 0$.  Secondly, $C_3$ is chosen such that $4T^{3-3C_3^2/8} \to 0$. In particular, we may take $C > 64 \times 2^{1/4}e^2$, $C_{\varepsilon} > 12$ and $C_3 > 2\sqrt{2}$. Finally, the leading constant $C_2 > 0$ in the error bound depends on all the aforementioned constants and the signal-to-noise ratio constant $C_{\alpha}$ in \Cref{assume:phase 2}, which should be chosen to be sufficiently large.

 For convenience, we have broken down the proof in five steps, each of which is applied to every $k \in \{1, \ldots, K\}$.  Before proceeding to the details, we have an overview of all steps.
 	
 	In {\bf Step 1}, we are to show that each working interval $(s, e)$ contains one and only one true change point, and the two endpoints are well separated;  {\bf Step 2} shows that the population CUSUM statistics within each working interval has good performances; the reasoning of the choices of the parameters in Algorithms~\ref{algorithm:USVT} and \ref{algorithm:RI}, and the good performances of the sampler CUSUM statistics in large probability events, will be detailed in {\bf Step 3}; additional probability controls regarding data splitting are demonstrated in {\bf Step 4}; and finally to show the localization rates, we are to transfer the network CUSUM statistics into a univariate case in {\bf Step 5}.

 	\vskip 3mm
 	\noindent {\bf Step 1.}  By \eqref{eq:pre estimate}, $\eta_k \in [\nu_{k-1}  , \nu_{k+1} ]  $ and 
		\begin{align*}
			\eta_k -  \nu_{k-1} &\ge \eta_k - \eta_{k-1} -| \eta _{k-1} -  \nu_{k-1} | \ge \Delta - \Delta/6 \ge 5\Delta /6, \\
			\nu_{k+1} -\eta_k & \ge \eta_{k+1} - \eta_{k} -| \eta_{k+1} -   \nu_{k+1} | \ge \Delta - \Delta/6 \ge 5\Delta /6.  
		\end{align*}	
		Similar calculations show also  that
		\[
			\min\{ \nu_ k -\nu_{k-1},  \nu_{k+1} -\nu_k  \} \ge  2\Delta/3 .
		\]
		Therefore, it holds that
			\[
				1/2 \min\{ \nu_ k -\nu_{k-1},  \nu_{k+1} -\nu_k  \} \ge \Delta/6.
			\]
		As a result, the interval
			\[
				[s,e] =[\nu_{k-1} + 1/2 ( \nu_{k} - \nu_{k-1}) , \nu_{k+1} -1/2 ( \nu_{k+1} -  \nu_{k})]
			\]
		contains only one change point $\eta_k$.  We have that
			\[
				\nu_k-s   = (1-1/2)(\nu_{k} - \nu_{k-1}) \ge  (1 - 1/2)2\Delta /3  = \Delta /3,
			\]
			and $e - \nu_k \geq \Delta/3$. Therefore, $ \min\{e-\nu_k ,\nu_k-s  \} \ge \Delta/3. $ 

	\vskip 3mm
	\noindent {\bf Step 2.}  Let $\Lambda(k) = \Theta(\eta_k) -\Theta(\eta_{k-1})$.  Then, by \Cref{lemma:one change point population size},
     	\[
			\|\widetilde  \Theta ^{s, e} (  t) \|_{\mathrm{F}}^2= \begin{cases}
			\frac {t-s}{(e-s)(e-t)}  (e-\eta_k)^2  \|\Lambda(k)  \|_{\mathrm{F}}^2 , & t\le \eta_k, \\
			\frac {e-t}{(e-s)(t-s)}  (\eta_k -s)^2  \|\Lambda(k)  \|_{\mathrm{F}}^2   , &  t\ge\eta_k.
			\end{cases}
		\]
     	Next, we set
     	\[
     		\widetilde{\Delta}_k  =  \sqrt{\frac{(\nu_k-s)(e-\nu_k)} {e-s}}
     	\]
     	and, without loss of generality, we may assume that $ \nu_k  \le \eta_k$.  Since 
     	\[
     	\widetilde{\Delta}_ k   \ge \min\{ \nu_k-s ,e-\nu_k \}/2  \ge  \Delta /6,
     	\] 
     	we obtain that
			\begin{align}
				\|\widetilde  \Theta ^{s, e} (  \nu_k) \|_{\mathrm{F}}^2 & =  \frac {\nu_k-s}{(e-s)(e-\nu_k)}  (e-\eta_k)^2  \|\Lambda(k)  \|_{\mathrm{F}}^2 = \widetilde{\Delta}_k^2  \left (\frac{e-\eta_k}{e-\nu_k}  \right)^2\kappa_k^2 \nonumber\\
				&=  \widetilde{\Delta}_k^2\left(1 - \frac{\eta_k -\nu_k}{e-\nu_k}\right)^2 \kappa_k^2 \geq \frac{\Delta}{6}\left(1 - \frac{\Delta/6}{\Delta/3}\right)^2\kappa_k^2  \geq \Delta\kappa^2_k/24. \label{eq:estimation population size}
			\end{align}

	\vskip 3mm
	\noindent {\bf Step 3.}  We next apply  \Cref{prop:operator norm} by letting $\varepsilon = C_{\varepsilon}\log(T)$, with $C_{\varepsilon} > 12$.  Define the event  
		\[
			\mathcal A =\left\{ \sup_{0\le s < t < e\le T}  \|  \widetilde A^{s,e} (t)  - \widetilde \Theta^{s,e} (t)\|_{\op} \le C \sqrt {n \rho } + C_{\varepsilon}\log(T)\right\},
     	\]
     	where $C > 64 \times 2^{1/4}e^2$.  Due to \Cref{prop:operator norm}, we have $ \mathbb{P} \left(  \mathcal A \right) \ge  1- 2T^{3-C_{\varepsilon}/4}$.
     	
     We then apply \Cref{lemma:USVT}. Set $\tau_2 = (3/4)(C\sqrt{n\rho} + C_{\varepsilon}\log(T))$, and define
		\[
			\mathcal B =  \left \{ \sup_{0 \leq s < t < e \le T}  \| \usvt (\widetilde A^{s,e} (t) ,\tau_2,\infty )- \widetilde \Theta^{s,e}(t) \|_\mathrm{F} \le 3\sqrt{r}\bigl(C\sqrt{n\rho} + C_{\varepsilon}\log(T)\bigr) \right\}.
		\]
		In order to apply \Cref{lemma:USVT}, let $A = \widetilde A^{s,e} (t)$, $B =\widetilde \Theta^{s,e} (t)$ and $\tau = \tau_2$.  We then have $\mathbb{P}(\mathcal B) \ge 1- 2T^{3-C_{\varepsilon}/4}$. 
		
		Let 
		\begin{align}\label{eq:usvt on first samples}  
			\widehat A^{s,e} (\nu_k) =\usvt (\widetilde A^{s,e} (\nu_k) ,\tau_2,\tau_3\widetilde{\Delta}_k).
		\end{align}
		Since $ \nu_k  \le \eta_k$, for any $i, j = 1, \ldots, n$, it holds that 
		\[
			\widetilde \Theta^{s,e}_{ij} (\nu_k) = \sqrt {  \frac {\nu_k-s}{(e-s)(e-\nu_k)}  } (e-\eta_k)   \Lambda_{ij}(k) \le  \widetilde{\Delta}_k \rho  \frac{e-\eta_k}{e-\nu_k}  \le  \widetilde{\Delta}_k\rho = \widetilde{\Delta}_k\tau_3.
		\]

		On the event $ \mathcal B$, 
		\[
			\| \widehat A^{s,e} (\nu_k)- \widetilde \Theta^{s,e}(\nu_k) \|_\mathrm{F}  \le \| \usvt (\widetilde A^{s,e} (\nu_k) ,\tau_2,\infty )- \widetilde \Theta^{s,e}(\nu_k) \|_\mathrm{F} \le 3\sqrt{r}\bigl(C\sqrt{n\rho} + C_{\varepsilon}\log(T)\bigr).
		\]
		By the triangle inequality and \Cref{assume:phase 2}, we have that
		\begin{align}\label{eq:sample size frobenius}
			\|  \widehat A^{s,e} (\nu_k) \|_\mathrm{F}  \ge \| \widetilde \Theta^{s,e}(\nu_k) \|_\mathrm{F} - 3\sqrt{r}\bigl(C\sqrt{n\rho} + C_{\varepsilon}\log(T)\bigr) \ge c_1' \sqrt {\Delta} \kappa_k,
		\end{align}
		where
		\[
		c_1' \leq 1/\sqrt{24} - \frac{3C}{C_{\alpha}\log^{1+\xi}(2)} - \frac{3C_{\varepsilon}}{C_{\alpha}\log^{1/2 + \xi}(2)},
		\]
		for any $n, T \geq 2$.  As a consequence,
		\begin{align*}
			 & 2 \left(\frac{\widetilde \Theta^{s,e}(\nu_k)  } {\|\widetilde \Theta^{s,e}(\nu_k)  \|_\mathrm{F}}, \frac{\widehat A^{s,e} (\nu_k) }{ \|\widehat A^{s,e} (\nu_k)  \|_\mathrm{F}}  \right)  = 2 - \left \| \frac{\widetilde \Theta^{s,e}(\nu_k)  }{\|\widetilde \Theta^{s,e}(\nu_k)  \|_\mathrm{F}} - \frac{\widehat A^{s,e} (\nu_k)}{\|\widehat A^{s,e} (\nu_k) \|_\mathrm{F} }\right\|_\mathrm{F}^2  \\
			\ge &  2-  4  \left( \frac{ \| \widetilde \Theta^{s,e}(\nu_k)  - \widehat A^{s,e} (\nu_k)\|_\mathrm{F} }{ \max \left \{  \| \widetilde \Theta^{s,e}(\nu_k)\|_\mathrm{F},  \|\widehat A^{s,e} (\nu_k)\|_\mathrm{F} \right\}}  \right)^2 \ge 2- \frac{9r \bigl(C\sqrt{n\rho} + C_{\varepsilon}\log(T)\bigr)^2}{(c'_1)^2\kappa^2_k   \Delta} \ge 1,
		\end{align*}
		where the second inequality follows from the definition of the event $\mathcal B$ and from \eqref{eq:estimation population size}, while the last inequality follows from \Cref{assume:phase 2} with a sufficiently large $C_\alpha$.  Therefore, 
    	\begin{equation}\label{eq:population size after projection}
    		(\widetilde \Theta^{s,e}(\nu_k)  , \widehat A^{s,e} (\nu_k) / \|\widehat A^{s,e} (\nu_k) \|_{\mathrm{F}} )\ge  \|\widetilde \Theta^{s,e}(\nu_k) \|_{\mathrm{F}} /2  \ge  (4\sqrt{6})^{-1}\sqrt {\Delta} \kappa_k,
    	\end{equation}
    	where in the last inequality we have used again \eqref{eq:estimation population size}. 

	\vskip 3mm
	\noindent {\bf Step 4.} Since $\{ B(t)\}_{t=1}^T $ is independent of $\{ A(t)\}_{t=}^T $, the distribution of $\{ B( t)\}_{t=1}^T $ does not change on the event $\mathcal B$.  Observe that, from \eqref{eq:usvt on first samples}, 
		\[
			\| \widehat A^{s,e} (\nu_k)\|_{\infty} \le \widetilde{\Delta}_k \tau_3  = \widetilde{\Delta}_k\rho.
		\]
		In combination with \eqref{eq:sample size frobenius}, the previous inequality implies that 
		\[
			(e-s)^{-1/2}\|\widehat A^{s,e} (\nu_k)\|_{\infty} / \| \widehat A^{s,e} (\nu_k)\|_F  \le \frac{\rho}{c_1'\sqrt{\Delta}\kappa_k} .
		\]	
	Using this bound along with  \Cref{lemma:berstein for zero one}, we obtain that, for any $\varepsilon > 0$,
		\[
			\mathbb{P} \left (\left |\frac{1}{\sqrt {e-s }} \sum_{t=s+1} ^e \left( \Theta (t)-B(t), \, \widehat A^{s,e} (\nu_k) / \| \widehat A^{s,e} (\nu_k)\|_\mathrm{F} \right)  \right| \ge \varepsilon  \right) \le 2\exp\left( \frac{-3/2\varepsilon^2}{3\rho + \varepsilon \rho/(c_1'\kappa_k\sqrt{\Delta})}\right).
		\]
		Setting $\varepsilon = C\sqrt {\rho } \log (T) $, with $C > 2\sqrt{2}$, we finally obtain the probabilistic bound
		\begin{align}\label{eq:1d bound 1}
			\mathbb{P}\left (   \left |\frac{1}{\sqrt {e-s }}  \sum_{t=s} ^e \left( \Theta (t) -B(t),\, \widehat A^{s,e} (\nu_k) / \| \widehat A^{s,e} (\nu_k)\|_\mathrm{F} \right )  \right| \ge C\sqrt {\rho } \log (T)  \right) \le 2T^{-3C^2/8}.
		\end{align}
		Similar arguments also show that
		\begin{align}\label{eq:1d bound 2}
			\mathbb{P} \left (   \left |  \left( \widetilde \Theta^{s,e}(t) -\widetilde B^{s,e}(t),\, \widehat A^{s,e} (\nu_k) / \| \widehat A^{s,e} (\nu_k)\|_{\mathrm{F}} \right )  \right| \ge C\sqrt {\rho } \log (T)  \right) \le 2T^{-3C^2/8}.
		\end{align}
		
	\vskip 3mm	
	\noindent {\bf Step 5.}  Consider the one dimensional time series $y(t) = (B(t), \widehat A^{s,e} (\nu_k) / \| \widehat A^{s,e} (\nu_k)  \|_\mathrm{F} )$.  Conditional on $\{A(t)\}_{t=1}^T$,  on the event $\mathcal B$, it holds that
		\[
		t \in [s,e] \mapsto	f(t) : = \mathbb{E}(y(t)) = (\Theta(t), \widehat A^{s,e} (\nu_k) / \| \widehat A^{s,e} (\nu_k) )\|_\mathrm{F}) 
		\]
		is a piecewise constant function  with only one change point, namely $\eta_k$.  Due to \eqref{eq:population size after projection}, it holds that 
		\[
			| \widetilde f^{s,e}(\eta_k )| = |(\widetilde \Theta^{s,e}(\eta_k), \widehat A^{s,e} (\nu_k) / \| \widehat A^{s,e} (\nu_k)  \|_\mathrm{F} )|  \ge |(\widetilde \Theta^{s,e}(\nu_k), \widehat A^{s,e} (\nu_k) / \| \widehat A^{s,e} (\nu_k)  \| _\mathrm{F})|    \ge (4\sqrt{6})^{-1}\sqrt \Delta\kappa_k, 
		\]
		and, by \eqref{eq:1d bound 1} and \eqref{eq:1d bound 2},
		\[
			 \mathbb{P} \left (  \sup_{s\le t\le e} \left |\frac{1}{\sqrt {e-s }}  \sum_{t=s} ^e \bigl(x(t) -f(t)\bigr)  \right| \ge C\sqrt {\rho } \log (T)  \right) \le 2T^{-c}
			 \]
			 and
\[
 \mathbb{P} \left (   \sup_{s\le t\le e} \left |   \widetilde x^{s,e}(t) -\widetilde f^{s,e}(t)  \right| \ge C\sqrt {\rho } \log (T)  \right) \le 2T^{-c},
		\]
		where $c = 3(C^2/8-1) > 0$.  We then apply Lemma 12 in \cite{WangEtal2017} by setting $\lambda = C\sqrt{\rho}\log(T)$.  It follows that $b_k = \arg\max_{s < t < e} | \widetilde x^{s,e}(t) |$ is an undetected change point such that, for a large enough constant $C_2 > 0$, 
		\[
			| b_k - \eta_k|  \le C_2 \frac{\rho (\log T)^2}{\kappa_k^2}. 
		\]
     \end{proof}

\begin{proof}[Proof of \Cref{thm-sbm-main}]
	In the proof of \Cref{theorem:localization 1}, note that arguments in {\bf Steps 1} and  {\bf 2} still hold under Assumptions in this theorem, and arguments in {\bf Steps 4} and {\bf 5} will still hold if the conclusions in {\bf Step 3} still hold.

	Let $[s,e]$ be defined as that in the proof of \Cref{theorem:localization 1}. We apply  \Cref{prop:operator norm} by letting $\varepsilon = C_{\varepsilon}\log(T)$, with $C_{\varepsilon} > 12$.  Define the event  
		\[
			\mathcal A' =\left\{ \sup_{0 \le s < t < e\le T}  \|  \widetilde A^{s,e} (t)  - \widetilde \Theta^{s,e} (t)\|_{\op} \le C \sqrt {n \rho } + C_{\varepsilon}\log(T)\right\},
     	\]
     	where $C > 64 \times 2^{1/4}e^2$.  Due to \Cref{prop:operator norm}, we have $ \mathbb{P} \left(  \mathcal A' \right) \ge  1- 2T^{3-C_{\varepsilon}/4}$.
     	
	For $t \in \{1, \ldots, T\}$, define $\Gamma(t)$ to be the block structure matrix satisfying
		\[
			\Gamma(t) - \diag(\Gamma(t)) = \Theta(t);
		\]
		in addition, for any $s < t< e$, define
		\[
			\widetilde{\Gamma}^{s, e}(t) = \sqrt{\frac{e-t}{(e-s)(t-s)}} \sum_{i = s+1}^t \Gamma(i) - \sqrt{\frac{t-s}{(e-s)(e-t)}} \sum_{i = t+1}^e \Gamma(i).
		\]	
	
	By \Cref{lemma:SBM frobenius norm}, on the event $\mathcal A'$, it holds that 
		\begin{align*}
			& \mathcal B' = \bigg\{\sup_{0\leq  s < t < e\le T} \| \usvt(\widetilde A^{s,e} (t), \tau_2,\infty)-\widetilde \Gamma^{s,e}( t ) \|_{\mathrm{F}}^2 \\
			\le &  9r \bigl(C\sqrt{n\rho} + C_{\varepsilon}\log(T)\bigr)^2 + 512 \| \diag ( \widetilde \Gamma^{s,e}(\nu_k) )\|_\mathrm{F}^2\bigg\}.
		\end{align*}
	
	Let 
		\[
			\widehat A^{s,e} (\nu_k) =\usvt (\widetilde A^{s,e} (\nu_k) ,\tau_2, \widetilde{\Delta}_k\tau_3 ).
		\]

	Observe that since $ \nu_k  \le \eta_k$ and $\|\widetilde \Lambda^{s,e} (\nu_k)\|_{\infty} \le \widetilde{\Delta}_k \tau_3$, on the event $\mathcal B'$ it holds that
		\begin{align*}
			& \| \widehat A^{s,e} (\nu_k)- \widetilde \Gamma^{s,e}(\nu_k) \|_{\mathrm{F}}   \le \| \usvt (\widetilde A^{s,e} (\nu_k) ,\tau_2,\infty )- \widetilde \Gamma^{s,e}(\nu_k) \|_{\mathrm{F}} \\
			\le & 3\sqrt{r}\bigl(C\sqrt{n\rho} + C_{\varepsilon}\log(T)\bigr) + 16\sqrt{2}\| \diag (\widetilde \Gamma^{s,e}(\nu_k) )\|_{\mathrm{F}}.
		\end{align*}
	Since $[s,e]$ contains only one change point $\eta_k$, by \Cref{assume:SBM frobenius norm} and \Cref{lemma:one change point population size},
		\begin{align}
			\|  \widehat A^{s,e} (\nu_k) \|_{\mathrm{F}} & \ge \| \widetilde \Gamma^{s,e}(\nu_k) \|_{\mathrm{F}} - 3\sqrt{r}\bigl(C\sqrt{n\rho} + C_{\varepsilon}\log(T)\bigr) - 16\sqrt{2}\| \diag (\widetilde \Gamma^{s,e}(\nu_k) )\|_{\mathrm{F}} \nonumber \\
			& \ge (1 - 16\sqrt{2}/C_{\Gamma})  \| \widetilde \Gamma^{s,e}(\nu_k) \|_{\mathrm{F}} - 3\sqrt{r}\bigl(C\sqrt{n\rho} + C_{\varepsilon}\log(T)\bigr) \nonumber \\
			& \ge \frac{1 - 16\sqrt{2}/C_{\Gamma}}{1 + C_{\Gamma}} \| \widetilde \Theta^{s,e}(\nu_k) \|_{\mathrm{F}} - 3\sqrt{r}\bigl(C\sqrt{n\rho} + C_{\varepsilon}\log(T)\bigr) \ge  c_1' \sqrt {\Delta} \kappa_k, \label{eq-longgg}
		\end{align}
		with $c_1' > 0$ by choosing proper constants.  \Cref{eq-longgg} follows from the fact that
		\[
			\| \widetilde \Theta^{s,e}(\nu_k) \|_{\mathrm{F}} \leq \|\widetilde \Gamma^{s,e}(\nu_k)\|_{\mathrm{F}} + \|\diag(\widetilde \Gamma^{s,e}(\nu_k))\|_{\mathrm{F}} \leq (1 + C_{\Gamma}) \|\widetilde \Gamma^{s,e}(\nu_k)\|_{\mathrm{F}}.
		\]
		As a consequence,
		\begin{align*}
			& 2\left(\frac{\widetilde \Theta^{s,e}(\nu_k)} {\|\widetilde \Theta^{s,e}(\nu_k) \|_{\mathrm{F}}}, \, \frac{\widehat A^{s,e} (\nu_k) }{ \|\widehat A^{s,e} (\nu_k)  \|_{\mathrm{F}}}  \right) = 2 - \left \| \frac{\widetilde \Theta^{s,e}(\nu_k)  }{\|\widetilde \Theta^{s,e}(\nu_k)  \|_{\mathrm{F}}} - \frac{\widehat A^{s,e} (\nu_k)}{\|\widehat A^{s,e} (\nu_k) \|_{\mathrm{F}}}\right\|_{\mathrm{F}}^2 \\
			= & 2 - \frac{\|\|\widehat A^{s,e} (\nu_k)\|_{\mathrm{F}}\Theta^{s,e}(\nu_k) -\| \widetilde \Theta^{s,e}(\nu_k) \|_{\mathrm{F}} \widehat A^{s,e} (\nu_k) \|_{\mathrm{F}}^2}{\| \widetilde \Theta^{s,e}(\nu_k) \|_{\mathrm{F}}^2  \|\widehat A^{s,e} (\nu_k)\|_{\mathrm{F}}^2 }\\
			\geq & 2 - \frac{ \| \widetilde \Theta^{s,e}(\nu_k)  - \widehat A^{s,e} (\nu_k)\|_{\mathrm{F}}^2}{\| \widetilde \Theta^{s,e}(\nu_k) \|_{\mathrm{F}}^2} - \frac{\left|\|\widehat A^{s,e} (\nu_k)\|_{\mathrm{F}}^2 - \|\widetilde \Theta^{s,e}(\nu_k)\|_{\mathrm{F}}^2\right|}{\| \widetilde \Theta^{s,e}(\nu_k) \|_{\mathrm{F}}^2}\\
			\ge & 2-  2  \frac{ \| \widetilde \Theta^{s,e}(\nu_k)  - \widehat A^{s,e} (\nu_k)\|_{\mathrm{F}}^2 }{\| \widetilde \Theta^{s,e}(\nu_k) \|_{\mathrm{F}}^2}   \\
			\ge &  2- 2\left(\frac{9r\bigl(C\sqrt{n\rho} + C_{\varepsilon}\log(T)\bigr)^2}{(c_1')^2\kappa_k^2\Delta}  + \frac{513 \|\diag (\widetilde \Lambda^{s,e}(\nu_k) )\|_{\mathrm{F}}}{\|\widetilde \Theta^{s,e}(\nu_k) \|_{\mathrm{F}}}   \right)  \ge 1,
		\end{align*}
		where the second inequality follows from \eqref{eq:estimation population size} and the event $\mathcal B'$, and the last inequality follows from  \Cref{assume:model 2} and \eqref{eq-longgg}.  Therefore 
		\begin{equation*}
			(\widetilde \Theta^{s,e}(\nu_k)  , \widehat A^{s,e} (\nu_k) / \|\widehat A^{s,e} (\nu_k) \|_{\mathrm{F}} )\ge  1/2\|\widetilde \Theta^{s,e}(\nu_k) \|_{\mathrm{F}} \geq c''\sqrt{\Delta}\kappa_k.
		\end{equation*}
%
		Thus all the conclusions in  {\bf Step 3} of  the proof of \Cref{theorem:localization 1} still hold.	
	\end{proof}

\section{Proofs of Lemmas 3 and 4}\label{sec-lower-proof}
In this subsection, we provide proofs of \Cref{lemma:lower bound testing} in \Cref{section:consistent testing} and \Cref{lem-3.3-lower} in \Cref{subsection:LR upper}, which provide the minimax lower bounds for detection and localization respectively.  In addition, \Cref{lemma:TV} is used in the proofs of Lemmas~\ref{lemma:lower bound testing} and \ref{lem-3.3-lower}.

\begin{lemma} \label{lemma:TV}
	Let $ \Theta \in \mathbb R^{n\times n}$ such that $\Theta_{ij} =\rho$ for all $1\le i,j\le n$, where $0 < \rho < 1/2$.  Let $A$ be an adjacency matrix of an inhomogeneous Bernoulli network with independent edges such that $ \mathbb{E}(A)=\Theta$. For any $v_b ,v_c\in [-\sqrt{\rho}, \sqrt{\rho}]^{ n}$, let $B$ and $C$ be adjacency matrices of inhomogeneous Bernoulli networks with independent edges such that $\mathbb{E}(B) =v_bv_b^{\top}+ \Theta$ and $\mathbb{E}(C) =v_cv_c^{\top} +\Theta$.  Let $P_A, P_B, P_C$ be the distributions of $A$, $B$ and $C$. Then
		\[
			\mathbb{E}_{P_A} \left(\frac{dP_B}{dP_A} \frac{dP_C}{dP_A}\right) \le\exp\left( \frac{ (v_b^{\top} v_c )^2}{\rho(1-\rho)}     \right).
		\]

	Let $A' = A-\mathrm{diag}(A)$, $B'=B-\mathrm{diag} ( B)$ and $C'=C-\mathrm{diag}(C)$.  Then
		\[
			\mathbb{E}_{P_{A'}} \left(\frac{dP_ {B'}}{dP_{A'}} \frac{dP_{C'}}{dP_{A'}}\right) \le\exp\left( \frac{ (v_b^{\top} v_c )^2}{\rho(1-\rho)}     \right).
		\]
\end{lemma}

\begin{proof}
	Let $\Gamma =v_bv_b^{\top}$ and $\Lambda=v_cv_c^{\top}$.
		\begin{align*}
			& \mathbb{E}_{P_A} \left(\frac{dP_B}{dP_A} \frac{dP_C}{dP_A}\right)  = \prod_{1\le i,j\le n} \left(  \frac{(\Gamma_{ij}+\rho) (\Lambda_{ij} +\rho)}{ \rho} + \frac{ (1-\Gamma_{ij} -\rho) (1-\Lambda_{ij} -\rho)}{ (1-\rho)}\right) \\
			= &\prod_{1\le i,j\le n} \left(  1+ \frac{\Gamma_{ij}\Lambda_{ij} }{ \rho (1-\rho)} \right) \le  \prod_{1\le i,j\le n} \exp\left( \frac{\Gamma_{ij}\Lambda_{ij} }{ \rho (1-\rho)}\right) = \exp\left( \frac{ (\Gamma, \Lambda)}{\rho(1-\rho)}     \right) =  \exp\left( \frac{ (v_b^{\top}v_c)^2}{\rho(1-\rho)}     \right).
		\end{align*}
	
	Note that
		\begin{align*}
			& \mathbb{E}_{P_{A'}} \left(\frac{dP_ {B'}}{dP_{A'}} \frac{dP_{C'}}{dP_{A'}}\right) = \prod_{ i\not = j} \left(  \frac{(\Gamma_{ij}+\rho) (\Lambda_{ij} +\rho)}{ \rho} + \frac{ (1-\Gamma_{ij} -\rho) (1-\Lambda_{ij} -\rho)}{ (1-\rho)}\right) \\
			= & \prod_{i\not = j} \left(  1+ \frac{\Gamma_{ij}\Lambda_{ij} }{ \rho (1-\rho)} \right) \le \prod_{1\le i , j\le n} \left(  1+ \frac{\Gamma_{ij}\Lambda_{ij} }{ \rho (1-\rho)} \right) \le  \prod_{1\le i,j\le n} \exp\left( \frac{\Gamma_{ij}\Lambda_{ij} }{ \rho (1-\rho)}\right) = \exp\left( \frac{ (v_b^{\top}v_c)^2}{\rho(1-\rho)} \right),
		\end{align*}
		where the fist inequality follows from the observation that $\Gamma_{ii} = (v_b)_i^2 \ge 0$ and  $\Lambda_{ii}= (v_c)_i^2 \geq 0$.
\end{proof}

\begin{remark}
	Let  $\Theta_{ij} = \rho +(vv^{\top})_{ij}$, where $v\in \{ \pm \sqrt{\kappa_0\rho } \}^{n}$, $0 < \rho < 1/2$ and $0 < \kappa_0 < 1$, then the community labels can be decided according to the vector $\mathrm{sign}(v)$. More precisely let 
		\[
			\mathcal C_1 = \{ i:\, v_i >0 \} \, \mbox{ and } \, \mathcal C_2 = \{ i:\, v_i<0 \}. 
		\]
		The probability within  $\mathcal C_1$ or $\mathcal C_2$ is $\rho (1+\kappa_0)$.  The  probability between $\mathcal C_1$ and $\mathcal C_2$ is $\rho (1-\kappa_0)$.
\end{remark}

\begin{proof}[{\bf Proof of \Cref{lemma:lower bound testing}}]
		Without loss of generality, suppose that $L = 4^{-1}T\log^{-1}(T)$ is an integer.  For $l \in \{1, \ldots, L\}$, $v \in \{1, -1\}^n$, let $\widetilde{P}^l$ be the joint distribution of a collection of independent adjacency matrices $\{A(t)\}_{t=1}^T$ such that
			\[
				\mathbb{E}\{(A(t))_{ij}\} = \rho/2, \quad i, j \in \{1, \ldots, n\}, \quad t \in \{1, \ldots, T\} \setminus \{(l-1)\log(T) + 1, \ldots,l \log(T)\},
			\]
			and 
			\[
				\mathbb{E}\{(A(t))_{ij}\} = \rho/2 + \frac{\rho^{1/2}(vv^{\top})_{ij}}{n^{1/2}}, \quad i, j \in \{1, \ldots, n\}, \quad t \in \{(l-1)\log(T) + 1, \ldots,l \log(T)\}.
			\]	
			Let $\widetilde{Q}^l = \widetilde{P}^{4L-l}$,
			\[
				\widetilde{P} = \frac{1}{L} \sum_{l = 1}^L  \widetilde{P}^l \quad \mbox{and }\quad  \widetilde{Q} =\frac{1}{L} \sum_{l = 1}^L  \widetilde{Q}^l.
			\]
			
			Note that for each $l \in \{1, \ldots, L\}$, $\widetilde{P}^l$ has two change points and $\Delta = \log(T)$.  Furthermore, 
			\[
				\max_{t = 1, \ldots, T}\|\mathbb{E}(A(t))\|_{\infty} \leq \rho \quad \mbox{and }\quad \kappa_0 = (n\rho)^{-1/2}.  
			\]
			As a result,
			\[
				\kappa_0\sqrt{n\rho \Delta} = \log^{1/2}(T),
			\]
			which implies that $\widetilde{P}^l \in \mathcal{P}$ and therefore $\widetilde{Q}^l \in \mathcal{P}$ for all $l$.
			
			It follows from Le Cam's lemma that
			\[
				\inf_{\widehat{\eta}} \sup_{ P \in \mathcal{P} } \mathbb{E}_P(H(\widehat{\eta}, \eta(P)))  \geq \frac{T}{4}\{1 - d_{\mathrm{TV}}(\widetilde{P}, \widetilde{Q})\}.
			\]
			
			Let $P^l$ be a finite-dimensional distribution of $\widetilde{P}^l$ consisting of only the first $T/2$ time points and $P_0$ be the joint distribution of a collection of independent adjacency matrices $\{B(t)\}_{t=1}^T$ such that
			\[
				\mathbb{E}\{(B(t))_{ij}\} = \rho/2, \quad i, j \in \{1, \ldots, n\}, \quad t \in \{1, \ldots, T/2\}.
			\]
			Due to the symmetry, we have
			\[
				d_{\mathrm{TV}}(\widetilde{P}, \widetilde{Q}) \leq 2 d_{\mathrm{TV}}(P, P_0),
			\]
			where 
			\[
				P = \frac{2}{L}\sum_{l = 1}^{L/2} P^l.
			\]
			Since $d_{\mathrm{TV}}(\cdot, \cdot) \leq \sqrt{\chi^2(\cdot, \cdot)}$, it suffices to bound $\chi^2(P, P_0)$.  We have
			\begin{align*}
				\chi^2(P, P_0) & = \left(\frac{2}{L}\right)^2 \left[\sum_{l = 1}^{L/2} \mathbb{E}_{P_0} \left(\frac{dP^l}{dP_0}\frac{dP^l}{dP_0}\right) + (L/2)(L/2 - 1)\right] - 1 \\
				& \leq \left(\frac{2}{L}\right)^2 \left[L/2 \exp(2\rho)^{\log(T)} + (L/2)(L/2 - 1)\right] - 1  \\
				& = (2/L)\left\{\exp\left(2\rho\log(T) \right) - 1\right\} = 8\log(T)(T^{2\rho - 1} - T^{-1}),
			\end{align*}
			where the inequality follows from \Cref{lemma:TV}.
			Therefore, since $\rho < 1/2$, there exits a sufficiently large $T$ such that $\log(T)(T^{2\rho - 1} - T^{-1}) = 64^{-1}$ and that concludes the proof.
	\end{proof}

\begin{proof}[{\bf Proof of \Cref{lem-3.3-lower}}]
	Let $\Theta (1),\Theta(2) \in \mathbb R^{n\times n}$ be such that for all $i, j = 1, \ldots, n$, $\Theta_{ij}(1) =\rho/2 $ and that  $\Theta_{ij}(2) = \rho/2 + \kappa_0\rho$.  Since $\kappa_0\le 1/ 2$, it holds that $\| \Theta (2) \|_{\infty}\le \rho$.

	For $\delta>0$ to be chosen later, let $P_1^\delta$ be the joint distribution of a collection of independent  adjacency matrices $\{A(t )\}_{t=1}^T$ such that 
		\[
			\mathbb{E} (A (t)) = \begin{cases}
				\Theta(1), & \mbox{if } t\le T/2+\delta, \\
				\Theta(2), & \mbox{if } t > T/2 +\delta.
			\end{cases}
		\]
		Let $P_2^\delta$ be the joint distribution of a collection of independent  adjacency matrices $\{B(t)\}_{t=1}^T$  such that 
		\[
			\mathbb{E} (B(t) ) = \begin{cases}
				\Theta(1), & \mbox{if } t\le T/2, \\
				\Theta(2), & \mbox{if } t > T/2.
			\end{cases}
		\]
		Then we have,
		\begin{align*}
			& 2d_{TV}^2 (P_1, P_2) \le  KL (P_1, P_2) \\
			= & \delta n^2 \left(  (\rho/2 + \kappa_0\rho) \log\left( \frac{\rho/2 + \kappa_0\rho}{\rho/2 }  \right)  	+  (1-\rho/2- \kappa_0\rho)\log\left( \frac{  1-\rho/2 - \kappa_0\rho}{1- \rho/2 } \right) \right)	\\
			\le & \delta n^2\left( (\rho/2 + \kappa_0\rho)\frac{ \kappa_0\rho}{\rho/2 } 	+ (1-\rho/2- \kappa_0\rho)\frac{ -  \kappa_0\rho}{1- \rho/2 } \right) 	\\
			= & \delta n^2\left(\kappa_0\rho + 2 \kappa^2_0 \rho -\kappa_0\rho +  \kappa_0^2 \rho^2(1-\rho/2)^{-1} \right) \le 4\delta \kappa_0^2n^2\rho = 4\delta\kappa^2_0n^2\rho.
		\end{align*}
	Since  
		\[
			\inf_{\hat \eta} \sup_{P\in \mathcal P} \mathbb{E}_P (| \hat \eta -\eta |)\ge  \delta (1- d_{TV} (P_1,P_2 )),
		\]
		taking $\delta =  \frac{1}{8\kappa^2n^2\rho}$, we have
		\[
			\inf_{\hat \eta} \sup_{P\in \mathcal P} \mathbb{E} _P (| \hat \eta -\eta |)\ge \frac{1}{16\kappa^2_0n^2\rho}.
		\]
\end{proof}

\section{Proofs of technical results used in Theorem 1 }\label{app:proofa}
 Throughout this section, for notational convenience we set $p = n^2$ and assume $\rho \sqrt{p} \geq \log(p)$.  We admit the discrepancy with \eqref{eq-assume-sparse} -- where we require $\rho n \geq \log(n)$. This will only affect the constants.

Observe that in \Cref{section:consistent testing}, no additional structure is imposed on the adjacency matrix.  In addition, for two matrices $A, B \in \mathbb{R}^{n\times n}$, we have
	\[
	(A, B) = \bigl\{\mathrm{vec}(A)\bigr\}^{\top}\mathrm{vec}(B),
	\]
	where $\mathrm{vec}(\cdot)$ is the vectorized version of a matrix by stacking the columns thereof.  It, therefore, suffices to view $A$ as a sparse Bernoulli vector with $p=n^2$ entries.  The assumptions below are vector versions of \Cref{assume:model}.  We  include them here for brevity.  

\begin{assumption}\label{assume:model bernoulli}
	Let $X(1),\ldots, X(T) \in \mathbb{R}^p$ be independent random vectors with independent Bernoulli entires. Suppose that the $i$th coordinate $X_i(t)$ of $X(t)$ satisfies $\mathbb{E}(X_i(t))=\mu_i(t)$ and that 
		\[
		\max_{1\le t\le T} \left\|\mu(t)\right\|_{\infty} \le \rho .
		\] 
\end{assumption}

Note that in fact if $A$ is an adjacency matrix of an inhomogeneous Bernoulli network defined in \Cref{def-1}, then due to symmetry, there are in fact $p = n(n-1)/2$ independent entries.  In this section, for notational simplicity, we let $p = n^2$ which has the same order as $n(n-1)/2$.

\begin{assumption}
	\label{assume:model bernoulli population}
	Let  $\{\eta_k\}_{k=0}^{K+1} \subset \{1, \ldots, T+1\}$ be a collection of change points, such that  $1 = \eta_0 < \eta_1 < \ldots < \eta_K \leq T < \eta_{K+1} = T + 1$ and, for $t=2,\ldots,T$, 
	\[
\mu(t) \neq \mu(t-1) \quad  \text{if and only if} \quad t \in \{ \eta_1,\ldots,\eta_K\}.
	\]
	Assume the spacing $\Delta $ satisfy that
	\[
	\Delta := \min_{k = 1, \ldots, K+1} \{\eta_k-\eta_{k-1}\} \leq T,
	\]
	and the normalized jump size $\kappa_0$ satisfies
	\[
	\inf_{k = 1, \ldots, K}\|\mu(\eta_k)-\mu(\eta_{k} -1)\| = \inf_{k = 1, \ldots, K}\kappa_{k} \geq \kappa_0\rho\sqrt{p} > 0.
	\]
\end{assumption}

\subsection{Probability bounds}
\label{section:probability bounds}

In this subsection, our task is to provide a probability bound for the event $\mathcal{A}(s, e, t)$ defined in \eqref{eq-event-A} to hold.  The result is formally stated in \Cref{lemma:exponential for sparse bernulli}, and necessary technical details are provided in Lemmas~\ref{lemma:berstein for zero one} and \ref{lemma:condition part}.

Suppose $\{ w_t\}_{t=1}^T \subset \mathbb R$ satisfies
	\begin{equation}\label{eq-wt-1}
	\sum_{t=1}^T w_t^2 =1. 
	\end{equation}

\begin{lemma}\label{lemma:berstein for zero one}
	Suppose \Cref{assume:model bernoulli} holds. Let $ v \in  \mathbb R^p$ and $\{w_t\}_{t=1}^T \subset \mathbb{R}$ satisfy \eqref{eq-wt-1}. Then for any $\varepsilon > 0$, we have
	$$\mathbb{P} \left (  \left |  \sum_{i=1} ^p v_i \sum_{t=1}^T w_t (X_i(t) -\mu_i(t)) \right| \ge \varepsilon  \right) \le 2\exp\left( 
	-\frac{3/2\varepsilon^2}{3\rho \| v\|_2^2  +\varepsilon  \max_{i=1}^p | v_i|  \max_{t=1}^T |w_t|}
	\right).$$
\end{lemma}
\begin{proof}
	Observe that  
	$$  \mathbb{E} \left( \sum_{i=1} ^p v_i \sum_{t=1}^T w_t (X_i(t) -\mu_i(t)) \right )^2  = \sum_{i=1}^p \sum_{t=1}^Tv_i^2 w_t^2\mathbb{E}(X_i(t) -\mu_i(t))^2  \le \rho \| v\|_2^2,$$
	due to the independence assumption and the fact that $\sum_{t=1}^Tw_t^2 = 1$, and that 
	$$ \max_{\stackrel{t = 1, \ldots, T}{i = 1, \ldots, p}}|w_tv_i(X_i(t) -\mu_i(t))  | \le  \max_{i=1}^p |v_i | \max_{t=1}^T|w_t |,$$
	since $X_i(t)$ is a Bernoulli random variable with mean $\mu_i(t)$.
	The desired result follows from Bernstein inequality.
\end{proof}

\begin{lemma}\label{lemma:condition part}
	Assume that the collection $\{Y(t)\}_{t=1}^T$ satisfies \Cref{assume:model bernoulli}.  Let $v = \sum_{t=1}^T w_t(Y(t) -\mu(t)) \in \mathbb{R}^p$.  Then there exists $C>0$ depending on $c > 0$ such that
		\[
			\mathbb{P}\left(\max_{1\le i\le p}v_i \ge C \sqrt {\log(p) \vee \log(T) } \right) \le T^{-c},
		\]
		and
		\[
			\mathbb{P}\left( \|v\| \ge C\sqrt {\log(p) \vee \log(T)} + \sqrt{\rho p}   \right) \le T^{-c},
		\]
\end{lemma}

\begin{proof}
	For the first part observe that it follows from Lemma~5.9 in \cite{vershynin2010introduction} that there exists some absolute constant $C_1 > 0$ such that
		\[
		\|v_i \|_{\psi_2}^2 \le C_1\sum_{t=1}^T w_t^2 \|Y_i (t) -\mu_i(t) \|_{\psi_2}^2 \le 2C_1,
		\]
		where $\|\cdot\|_{\psi_2}$ is the Orlicz norm \citep[e.g. Definition~5.7 in][]{vershynin2010introduction}, and the second inequality follows from $\|Y_i (t) -\mu_i(t) \|_{\psi_2}^2 \leq 2$ and $\sum_{t=1}^Tw_t^2 = 1$.  Therefore for each $i = 1, \ldots, p$, $v_i$ is sub-Gaussian and there exist a constant $c > 0$ and a large enough $C >0$ depending on $c$ and $C_1$ such that 
		$$ \mathbb{P}\left(v_i \ge C\sqrt {\log(p) \vee \log(T) }  \right) \le (p \vee T) ^{-c-1}.$$
		Since
		\[
		p(p \vee T) ^{-c-1} \leq \begin{cases}
 			T^{-c}, & p \leq T,\\
 			p^{-c} \leq T^{-c}, & p \geq T,
			\end{cases}
		\]
		the desired result follows from a union bound argument.

	For the second part, define $F(x_1, \ldots, x_p) =\|x\|$ and $G_i(y_1,\ldots,  y_t) =\sum_{t=1}^T w_t (y_t -\mu_i(t))$, $i = 1, \ldots, p$.  Since for all $i$, both $F$ and $G_i$  are one Lipschitz function, $\|v\|$ is a one Lipschitz function of $\{ \{Y_i(t)\}_{i=1}^p \}_{t=1}^T$.  It follows from the proof of Corollary 4 in \cite{samson:2000} that, for any $\varepsilon > 0$, 
		\[
		\mathbb{P}\left(\|v\| > \mathbb{E} \| v\|  + \varepsilon \right )\le \exp\left( -\varepsilon^2/2\right). 
		\]
	Since $\mathbb{E }\| v\|\le \sqrt { \sum_{i=1}^ p \mathbb{E}(v_i^2)}  \le \sqrt {\rho p}, $  the desired results follows by by taking $ \varepsilon =C\sqrt {\log (p) \vee \log(T)}$.

\end{proof}
 
\begin{lemma} \label{lemma:exponential for sparse bernulli}
	Let $\{X(t)\}_{t=1}^T$ and $\{Y(t)\}_{t=1}^T$ be two independent copies, both of which satisfying \Cref{assume:model bernoulli}.  Suppose in addition that 
		\[
		\rho  \sqrt{p} \ge  \log(p) .
		\]
		For $\{w_t\}_{t=1}^T$ satisfying $\sum_{t=1}^Tw_t^2 = 1$, let $\widetilde{X} = 	\sum_{t=1}^Tw_tX(t)$, $\widetilde{Y} = 	\sum_{t=1}^Tw_tY(t)$ and $\widetilde{\mu} = \sum_{t=1}^Tw_t\mu(t)$.  There exists $C_{\beta} > 0$ depending on $c$ and $c_T$ such that
		\[
			\mathbb{P}\left(  \left | \sum_{i=1}^p \widetilde X_i \widetilde Y_i  -\sum_{i=1}^p \widetilde  \mu_i^2 \right | \ge  C_{\beta}\log(T) \left( \|\widetilde \mu\| + \log^{1/2}(T)\rho\sqrt {  p}\right) \right) \le 6T^{-c_T}+ 2T^{-c},
		\]
		where $C_{\beta} > \max\{ 4c_T/3, \sqrt{3c_T(C+1)^2 + C^2}\}$, and $C, c$ are from \Cref{lemma:berstein for zero one}.
\end{lemma}

\begin{proof}
	Note that $ \sum_{i=1}^p \widetilde X_i \widetilde Y_i  -\sum_{i=1}^p \widetilde  \mu_i^2 = I + II + III$, where
	\begin{align*}
	I = &\sum_{i=1}^p  (\widetilde X_i -\widetilde \mu_i ) (\widetilde Y_i -\widetilde \mu_i ), \quad II =  \sum_{i=1}^p(\widetilde X_i -\widetilde \mu_i )\widetilde \mu_i \quad \mbox{and} \quad III =\sum_{i=1}^p(\widetilde Y_i -\widetilde \mu_i )\widetilde \mu_i.
	\end{align*}
	It suffices to bound $I$ and $II$, due to the fact that $\{X(t)\}_{t=1}^T$ and $\{Y(t)\}_{t=1}^T$ are iid.

	As for $I$, for any $i = 1, \ldots, p$, let $v_i = \sum_{t=1}^T w_t (Y_i(t) -\mu_i(t)) $. Conditional on $\{Y(t)\}_{t = 1}^T$, it follows from \Cref{lemma:berstein for zero one} that for any $\varepsilon > 0$, we have
		\[
		\mathbb{P}_{X|Y}\left( \left|\sum_{i=1}^p v_i \sum_{t=1}^T w_t (X_i(t) -\mu_i(t)) \right| \ge \varepsilon  \right) \le 2\exp\left(-\frac{3/2\varepsilon^2}{3\rho\|v\|^2 + \varepsilon \max_{i }| v_i|}\right),
		\]
		due to the fact that $\max_t|w_t| \leq 1$.  By \Cref{lemma:condition part}, there exist $C, c > 0$ such that
		\[
		\mathbb{P}_Y \left(\max_{i = 1, \ldots, p}|v_i| \ge C\sqrt {\log (p) \vee \log(T)} \right) \le T^{-c}, 
		\]
		and that 
		\[
		\mathbb{P}_Y\left( \|v\| \ge C\sqrt {\log(p) \vee \log(T)} + \sqrt{\rho p}   \right) \le T^{-c}.
		\]
		Thus for any $\varepsilon > 0$, it holds that
		\begin{align*}
		& \mathbb{P}_{X, Y}\left (  \left |  \sum_{i=1} ^p v_i \sum_{t=1}^T w_t (X_i(t) -\mu_i(t)) \right| \ge  \varepsilon  \right) \\
		\le &  2\exp\left( -\frac{3/2\varepsilon ^2}{3\rho\bigl(C\sqrt{\log(p) \vee \log(T)} + \sqrt{\rho p}\bigr)^2  + C\varepsilon  \sqrt {\log(p) \vee \log(T)} }\right) + 2T^{-c}.
		\end{align*}
		
	Since $\rho \sqrt  p \ge  \log(p)$, by taking $ \varepsilon= C ''\rho \sqrt {p} \log^{3/2}(T)$ for sufficiently large 
		\[
		C'' \geq \sqrt{3c_T(C+1)^2 + C^2},
		\]
		it holds that
		\[
		\mathbb{P}(|I | \ge C''\rho \sqrt p \log^{3/2}(T) ) \le 2T^{-c_T} +2T^{-c}.
		\]

	Observe that $III$ is identically distributed as $II$. For $II$, observe that for $\varepsilon > 0$, it follows from \Cref{lemma:berstein for zero one},
		\[
		\mathbb{P} \left (  \left |  \sum_{i=1} ^p  \widetilde \mu_i \sum_{t=1}^T w_t (X_i(t) -\mu_i(t)) \right| \ge \varepsilon  \right) \le 2\exp\left(-\frac{3/2\varepsilon^2}{3\rho \| \widetilde \mu\|^2  +\varepsilon \max_{i }| \widetilde \mu_i| \max_t |w_t |}\right) .
		\]
	
	Let  $\varepsilon = C'\|\widetilde \mu\| \log (T)$, with $C' > 4c_T/3$,
		\[
			3\rho \| \widetilde \mu\|^2  + \varepsilon \max_{i }| \widetilde \mu_i| \max_t |w_t |  \le 3\rho \| \widetilde \mu\|^2 + \varepsilon\rho  \le  3\| \widetilde \mu\|^2 + \varepsilon  \le 3/(2c_T) \epsilon^2/\log(T) .  
		\]
		Therefore $P(| II| \ge C' \|\widetilde \mu\|\log (T) )  \le 2T^{-c_T}$.
	
\end{proof}


\subsection{Localization}

This is the key lemma used in the proof of \Cref{thm-1} to localize the change points.  We deliberately present this lemma with seemingly low-level conditions, in order for us to directly check the conditions in the proof of \Cref{thm-1}.

\begin{lemma}\label{lemma:wbs multiple d} 

	Assume $\{X_t\}_{t=1}^T$ and $\{Y_t\}_{t=1}^T$ be two independent copies $\mathbb{E}(X_t) = \mathbb{E}(Y_t) = \mu(t)$ such that \Cref{assume:model bernoulli population} holds.  
	
	Let $[s_0, e_0]$ be any interval with $e_0 - s_0\le C_R\Delta$ and containing at least one change point $\eta_r$ such that 
	\[
	\eta_{r-1} \le s \le \eta_r \le \ldots\le \eta_{r+q} \le e \le \eta_{r+q+1},
	\quad q\ge 0
	\]
	and that $\min\{s_0 -\eta_r, e_0 -\eta_{r+q} \}\ge \Delta /2$.  Denote  $\kse= \max\{\kappa_p: \,\min\{\eta_p - s_0, e_0 -\eta_p \} \ge \Delta /16\}$.
	Consider any generic $[s,e] \subset [s_0,e_0]$  such that 
	$[s,e]$ contains at least one change point.
	Let $b \in \arg \max_{s < t < e} ( \widetilde X^{s,e} (t),\widetilde Y^{s,e} (t) )  $.
	For some $c_4>0$ and $\lambda>0$, suppose that
	\begin{align}
		&( \widetilde X^{s,e} (b),\widetilde Y ^{s,e}(b)   )  \ge c_4 (\kse)^2 \Delta  \label{eq:wbs size of sample} \\
		&\sup_{s < t < e} |( \widetilde X^{s,e}(t) ,\widetilde Y^{s,e}  (t))    -\|\widetilde \mu^{s,e}(t) \|_2^2 |  \le \lambda 
		\label{eq:mcusum noise}
	\end{align}
	If there exists a sufficiently small absolute constant $ c_5>0 $ satisfying
		\[
		c_5 < \min\left\{\frac{c_3}{2C_R^2 + 2c_3}, \, \frac{1}{2+32C_R^2/\min\{1/4, 1/2-2c_3\}}\right\}
		\]
		with $c_3$ defined in \Cref{lemma:Venkatraman}, such that
	\begin{equation}\label{eq:wbs noise}
		\lambda\le c_5 \max_{s<t<e} \| \widetilde \mu^{s,e}(t)\|_2^2
	\end{equation}
	then there exists a change point $\eta_{k} \in (s, e)$  such that 
	\[
	\min \{e-\eta_k,\eta_k-s\}  >  \Delta /4,  \quad
	|\eta_{k} -b |\le \frac{C_3\Delta\lambda}{ \| \widetilde \mu^{s,e} (\eta_k)\|_2^2}  \quad 
	 \text{and} \quad  \|\widetilde \mu^{s,e} (\eta_k )  \|_2^2 \ge (1-2c_5)\max_{s<t <e} \| \widetilde \mu^{s,e} (t)\|_2^2,
	\]
	where $C_3 = 2C_R^2/\min\{1/4, 1/2-2c_3\}$.
\end{lemma}

\begin{proof}
	For any $t\in \{s+1, \ldots, e-1\}$, denote $\widetilde Z^{s,e}(t)=  ( \widetilde X^{s,e} (t) ,\widetilde Y^{s,e} (t)  )$.  It follows from \Cref{eq:monotonicity of mcusum} that without loss of generality, we can assume $\|\widetilde \mu^{s,e}(t) \|^2  $ is locally decreasing at $b$.  Observe that  this implies that  there exists a change point $\eta_k \in [s,b]$, since otherwise $\|\widetilde \mu^{s,e}(t) \|_2^2$  is increasing on $[s,b] $ as a consequence of  \Cref{lemma:mcusum property near the boundary}.  Therefore, we have
		\[
		s\le \eta_k\le b \le  e .
		\]
		Observe that 
		\begin{align}
			\widetilde  Z^{s,e} (b) \ge \max_{s<t <e} \| \widetilde \mu^{s,e} (t)\|^2 -\lambda  \ge c_5^{-1}(1-c_5) \lambda,  \label{eq:mcusum signal noise 1}  
		\end{align}
		which follows from \eqref{eq:mcusum noise} and \eqref{eq:wbs noise}, and
		\begin{align}
			\| \widetilde \mu^{s,e} (b)\|^2 &\ge \widetilde  Z^{s,e} (b) -\lambda   \ge \max_{s<t <e} \| \widetilde \mu^{s,e} (t)\|^2 -2\lambda \ge c_5^{-1}(1-2c_5)\lambda, \label{eq:mcusum signal noise 3}
		\end{align}
		which follows from \eqref{eq:mcusum signal noise 1}.  We, consequently, have
		\begin{align} 
			\| \widetilde \mu^{s,e} (b)\|^2  \ge \widetilde Z^{s,e}(b)   -\lambda 
		\ge  (1-c_5(1-c_5)^{-1})\widetilde Z^{s,e}(b)   
		> \widetilde Z^{s,e}(b) /2
		\ge  (c_4/2)  (\kse)^2  \Delta . \label{eq:wbs size of change point}
		\end{align}
		where the second inequality follows from \eqref{eq:mcusum signal noise 1} and the last inequality follows from \eqref{eq:wbs size of sample}.

	Since $s\le \eta_k\le b \le  e $ and $\|\widetilde \mu^{s,e}(t) \|_2^2$ is locally decreasing at $b$, by \Cref{eq:monotonicity of mcusum}, $\|\widetilde \mu^{s,e}(t) \|_2^2$ is decreasing within $[\eta_k,b] $.  Therefore
		\begin{align} \label{eq:mcusum decreasing}
			\|\widetilde \mu^{s,e}(\eta_k) \|^2 &\ge \| \widetilde \mu^{s,e} (b)\|^2.
		\end{align}
		\Cref{eq:mcusum decreasing} combining with \eqref{eq:mcusum signal noise 3} gives
		\[
			\|\widetilde \mu^{s,e} (\eta_k )  \|^2 \ge (1-2c_5)\max_{s<t <e} \| \widetilde \mu^{s,e} (t)\|^2. 
		\]

	\vskip 3mm
	\noindent {\bf Step 1.} 
	In this step, it will be shown that $\min\{ \eta_k -s , e -\eta_k \} \ge \min\{1, c_4 \}\Delta /16$. 

	Suppose $\eta_k$ is the only change point in $(s, e)$.  It must hold that $\min\{ \eta_k -s , e -\eta_k \} \ge   \min\{1, c_4\}\Delta /16$, otherwise by \Cref{lemma:one change point population size}, 
		\[
			\|\widetilde \mu^{s,e}(\eta_k) \|^2 = \frac{(\eta_k - s)(e - \eta_k)}{e-s}\kappa_k^2  <  \frac{c_4 }{16} \kappa_k^2\Delta \le \frac{c_4}{2}  (\kse)^2\Delta,
		\]
		which contradicts \eqref{eq:wbs size of change point}.
	
	Suppose $[s, e]$ contains at least two change points. For the sake of contradiction, suppose $\min\{ \eta_k -s , e -\eta_k \} < \min\{1, c_4 \}\Delta /16$.  Reversing the time series if necessary, it suffices to consider  
		\begin{align}\label{eq:first change point mcusum} 
			\eta_k -s <  \min\{1,c_4  \}\Delta /16.
		\end{align} 
		Observe that \eqref{eq:first change point mcusum} implies that $\eta_k$ is the first change point in $[s,e]$.  Therefore 
		\begin{align*}
			\|\widetilde \mu^{s,e}(\eta_k) \|^2  &\le \frac{1}{8} \| \widetilde \mu^{s,e} (\eta_{k+1})\|^2 +4\kappa_r^2  (\eta_k -s) \le \frac{1}{8}\max_{s < t < e} \|\widetilde \mu^{s,e}(t) \|^2+\frac{c_4}{4}\kappa_k ^2\Delta \\
			& \le  \frac{1}{8} (1-2c_5)^{-1} \| \widetilde \mu^{s,e}(b) \|^2+\frac{1}{2}\| \widetilde \mu^{s,e}(b) \|^2 < \| \widetilde \mu^{s,e}(b) \|_2,
		\end{align*}
		where the first inequality follows from \Cref{lemma:mcusum boundary bound} and \eqref{eq:first change point mcusum}, the second inequality  follows from \eqref{eq:first change point mcusum}, the third inequality follows from  \eqref{eq:mcusum signal noise 3} and \eqref{eq:wbs size of change point}, and the fourth inequality follows from $c_5 < 3/8$.  This contradicts \eqref{eq:mcusum decreasing}.

	\vskip 3mm
	\noindent{\bf Step 2.} In order to apply  \Cref{lemma:Venkatraman}, it suffices to check that \eqref{eq:ven 2} for $\widetilde \mu^{s,e}(t)$. Observe that  
		\begin{align*}  
			& \max_{s<t <e} \| \widetilde \mu^{s,e} (t)\|^2  -	\|\widetilde \mu^{s,e} (\eta_k )  \|^2 \le 2c_5 \max_{s<t <e}  \| \widetilde \mu^{s,e} (t)\|^2  \le  2c_5(1-2c_5)^{-1}	\|\widetilde \mu^{s,e} (\eta_k )  \|^2  \\
			&\le \frac{2c_5C_R^2}{c_3 (1-2c_5) } c_3\|\widetilde \mu^{s,e} (\eta_k )  \|^2   \Delta^2 (e-s)^{-2} \le c_3\|\widetilde \mu^{s,e} (\eta_k )  \|^2   \Delta^2 (e-s)^{-2},
		\end{align*}
		where $c_3$ is defined as in \eqref{eq:ven 2}, the first and the second inequality follow from \eqref{eq:mcusum signal noise 3},  the third inequality follows from $e-s\le C_R \Delta$ and the last inequality hold for sufficiently small
		\[
			0 < c_5 < \frac{c_3}{2C_R^2 + 2c_3}.
		\]
		Let $c$ be defined in \Cref{lemma:Venkatraman}. Since $ e-s \le C_R\Delta$,
		\[
			\frac{2\lambda (e-s)^2}{c\Delta \|	\widetilde \mu^{s,e}(\eta_k) \|^2 } \le 2C_R^2 \frac{\lambda \Delta }{c c_5^{-1} (1-c_5) \lambda   } < \Delta /16, 
		\]
	where the first inequality follows from \eqref{eq:mcusum signal noise 3} and the last inequality holds for sufficiently small $c_5$ satisfying
		\[
			c_5 < \frac{1}{2 + 32C_R^2/c}.
		\] 
		By \Cref{lemma:Venkatraman} if  $d$ is chosen such that   
		\begin{align}\label{eq:choice of d in derivative} 
	  		d -\eta_k = \frac{2\lambda (e-s)^2}{c\Delta \|	\widetilde \mu^{s,e}(\eta_k) \|^2 } < \Delta /16,
		\end{align}
		and that 
		\begin{align}\label{eq:venkatraman ucusum}
			\|\widetilde \mu^{s,e}(\eta_k) \|^2  -	\| \widetilde \mu^{s,e}  (d)\|^2    > c \|\widetilde \mu^{s,e}(\eta_k) \|^2  |d-\eta_k | \Delta (e-s)^{-2}\ge 2 \lambda,
		\end{align}
		where the first inequality follows from \Cref{lemma:Venkatraman} and the second inequality follows from  \eqref{eq:choice of d in derivative}.

	For the sake of contradiction, suppose $b\ge d $. Then 
		\[
			\|\widetilde \mu^{s,e}(b) \|^2 \le  \|\widetilde \mu^{s,e}(d) \|^2 < \|\widetilde \mu^{s,e}(\eta_k) \|^2 -2\lambda  \le \max_{s < t < e} \|\widetilde \mu^{s,e}(t) \|^2  -2\lambda \le \max_{s < t < e}\widetilde Z(t)  +\lambda-2 \lambda  =\widetilde Z(b) -\lambda ,
		\]
		where the first inequality follows from \Cref{eq:monotonicity of mcusum},  which ensures that  $\|\widetilde \mu^{s,e}(t) \|^2 $ is decreasing on $[\eta_{k},b] $ and $d\in [\eta_{k},b]$, the second inequality follows from \eqref{eq:venkatraman ucusum}.  This is a contradiction to \eqref{eq:mcusum noise}.  Thus $b\le d$ and so 
		\begin{align*}
			0\le  b-\eta_k   \le d-\eta_k \le  \frac{2\lambda (e-s)^2}{c\Delta \|	\widetilde \mu^{s,e}(\eta_k) \|^2 }    \le \frac{2C_R^2}{c} \frac{\Delta  \lambda}{   \|	\widetilde \mu^{s,e}(\eta_k) \|^2  } 
		\end{align*}
		where  the third inequality follows from $ e-s\le C_R\Delta$. 
	
\end{proof}

\section{Proofs of technical results used in Theorems 2 and 3}

\subsection{Matrix estimation}
We first establish some results concerning matrix estimation.

\begin{lemma}\label{prop:operator norm}

\noindent {\bf 1.}	Let $\{ A (t)\}_{t=1}^T $ be a collection of independent matrices with independent Bernoulli entries satisfying
	\[
		\max_{1\le t\le T}\|\mathbb{E}A(t)\|_{\infty}\le \rho,
	\]	
	with $n\rho \geq \log(n)$. Let  $\{w(t)\}_{t=1}^T \subset \mathbb R $ be a collection of scalars such that $\sum_{t=1}^T w(t)^2 =1$ and $\sum_{t=1}^T w(t) =0$.  Then there exists an absolute constant $C > 32 \times 2^{1/4}e^2$ such that 
		\begin{align}\label{eq:operator norm bounds} 
			\mathbb{P}\left( \left\| \sum_{t=1}^T w(t) A(t) - \mathbb{E} \left( \sum_{t=1}^T w(t) A(t)  \right)\right\|_{\mathrm{op}}  \ge C \sqrt{n\rho} + \varepsilon  \right)\le  \exp(-\varepsilon^2/2).
		\end{align}

\vskip 3mm
\noindent {\bf 2.} If $\{ A (t)\}_{t=1}^T $ are symmetric matrices, then \eqref{eq:operator norm bounds} still holds.
\end{lemma}
\begin{proof}
	Observe that the conclusion in	{\bf 2} is a consequence of that in {\bf 1}, as if $A(t)$ is symmetric, then $A(t) =A'(t) +A''(t) $, where $A'(t)$ is the upper diagonal matrix of $A(t)$ including the diagonal and $A''(t)$ is the lower diagonal matrix of $A(t)$.  Therefore the conclusion in {\bf 2} follows by applying the conclusion in {\bf 1}  to $A'(t)$ and $A''(t)$.  In the rest of this proof, we will only consider {\bf 1}.

	Let $B(t) = A(t) - \mathbb{E}(A(t))$ and $\widetilde B =\sum_{t=1}^Tw(t)B(t) $.  The function 
		\[
			H(B(1), \ldots, B(T))  = \left\| \sum_{t=1}^T w(t)B(t)\right\|_{\op} = \bigl\|  \widetilde B \bigr\|_{\op}
		\] 
		is one-Lipschitz, therefore by Corollary 4 in \cite{samson:2000}, one has for any $\varepsilon > 0$, 
		\[
			\mathbb{P}\left\{ \left\|\sum_{t=1}^Tw(t)B(t) \right\|_{\op} \ge \mathbb{E} \left(\left\|\sum_{t=1}^Tw(t)B(t) \right\|_{\op} \right  )+\varepsilon \right\}\le \exp(-\varepsilon^2/2).
		\]
		To complete the argument, it suffices to bound $\mathbb{E} \left(\|\sum_{t=1}^Tw(t)B(t) \|_{\op} \right)$.  By \Cref{lemma:Tomozei}  and for all $t \in \{1, \ldots, T\}$, the entries of $w(t)B(t)$ are bounded on $[-\rho,1]$, there exists a collection of random matrices $\{Z(t)\}_{t=1}^T\subset \mathbb R^{n\times n}$ such that $\mathbb{E}(Z| B) =B$, where $Z=  (Z(1),\ldots, Z(T))$ and $B=  (B(1),\ldots,B(T))$, and that $(Y_t)_{ij}= (1-\rho)(Z(t))_{ij} + \rho$ are mutually independent Bernoulli random variables with parameter $\rho$.  Denote $G(B) =\|\sum_{t=1}^Tw(t)B(t) \|_{\op}$. Then
		\begin{align*}
			\mathbb{E}\left(\left\|\sum_{t=1}^Tw(t)B(t) \right\|_{\op}\right) = & \mathbb{E} \left(G(B  ) \right) = \mathbb{E} \left( G(\mathbb{E}( Z | B ) ) \right) \le \mathbb{E} \left( \mathbb{E}  (G(Z) ) |B\right)  = \mathbb{E}\left(\left\| \sum_{t=1}^T w(t) Z(t)\right\|_{\op}\right) \\
			= &\frac{1}{1-\rho } \mathbb{E}\left(\left\|\sum_{t=1}^T w(t)Y(t) \right\|_{\op}\right),
		\end{align*}
		where $G$ being convex is used in the inequality and $\sum_{t=1}^T w(t)=0$ is used in the last equality.  Since the entries of $\sum_{t=1}^T w(t)Y(t)$ are independent and identically distributed, by  \Cref{lemma:iid edge}, 
		\[
			\mathbb{E}\left(\left\|\sum_{t=1}^Tw(t)B(t) \right\|_{\op}\right)  \le C\sqrt {\rho n},
		\]
		where $C > 32 \times 2^{1/4}e^2$.
\end{proof}

 \begin{lemma}\label{lemma:Tomozei}
 	Let $X \in [-\rho, 1]$ be a centered Bernoulli random variable. Then there exists a random variable $Y$ such that 
 	\begin{itemize}
 		\item $\mathbb{E}(Y|X) = X$, and
 		\item $(1-\rho)Y +\rho$ is a Bernoulli random variable with parameter $\rho$.
 	\end{itemize}
 \end{lemma}
 
 \begin{proof}
 	The proof is taken from the proof of Lemma~2 in \cite{tomozei2014distributed}, by letting 
 		\[
 		Y = 1 - \mathbbm{1}\{X \leq (1+\rho)U - \rho\},
 		\]
 		where $U$ is a Uniform$[0, 1]$ random variable independent with $X$.
 \end{proof}

\begin{lemma}\label{lemma:iid edge}
	Let $\{ A(t)\}_{t=1}^T$ be a collection of independent adjacency matrices whose entries are independent Bernoulli random variables with parameter $\rho$ satisfying with $n\rho\ge c_2\log (n)$, $c_2 > 4$, and let $B_t=A_t-E(A_t)$. Suppose $\{w_t\}_{t=1}^T \subset \mathbb R $ be a collection of scalar such that $\sum_{t=1}^T w_t^2 =1$. Then there exists an absolute constant $C > 32 \times 2^{1/4}e^2$ such that 
		\[
			\mathbb{E} \left( \left\|\sum_{t=1}^T  w_tB_t\right\|_{\op}\right)\le C\sqrt {n \rho}.
		\]
\end{lemma}
\begin{proof}
	Let $\widetilde B =\sum_{t=1}^Tw_tB(t)$.  To bound $\mathbb{E}( \|\widetilde B \|_{op})$, since the entries of $\widetilde B$ are independent and identically distributed with $\mathbb{E}(\widetilde B) =0$, by Corollary 2.2 in \cite{seginer2000expected}, one has
		\[
			\mathbb{E} \left( \|\widetilde B \|_{\op} \right)  \le C_1 \mathbb{E}\left( \max_{1\le i\le n} \| \widetilde  B_{i*}\|\right),
		\]
		where $C_1 = 16 \times 2^{1/4}e^2$.  For any $i \in \{1, \ldots, n\}$,  since $\| \widetilde  B_{i*}\|$ is one-Lipschitz convex function, by Corollary 4 in \cite{samson:2000}, it holds that for any $\varepsilon > 0$,
		\[
			\mathbb{P}\left(  \|\widetilde  B_{i*}\| \ge \mathbb{E} \|\widetilde  B_{i*}\| +\varepsilon  \right) \le  \exp(-\varepsilon^2/2).
		\]
		Since
		\[
			(\mathbb{E} \|\widetilde  B_{i*}\|)^2 \le \mathbb{E} ( \|\widetilde  B_{i*}\|^2)  = \sum_{t=1}^T w_t^2 \mathbb{E} (\| (B(t))_{i*}\|^2)  + \sum_{s\not = t} w_sw_t \mathbb{E}( B_s , B_t) = \sum_{t=1}^T w_t^2 n\rho(1-\rho) \le  n\rho,
		\]
		one has
		\begin{equation}\label{eq:column bound}
			\mathbb{P}\left(  \|\widetilde  B_{i*}\| \ge \sqrt {n\rho} +\varepsilon  \right)\le  \exp(-\varepsilon^2/2).
		\end{equation}
		
	Using the above display, it follows that
	\begin{align*}
			& \mathbb{E}\left(\max_{1\le i\le n}  \|\widetilde  B_{i*}\|\right)   = \int_{0}^\infty \mathbb{P}\left(\max_{1\le i\le n}  \|\widetilde  B_{i*}\|_2 \ge \varepsilon\right)\,d \varepsilon \le \int_0^{2\sqrt {\rho n} }  1\,d \varepsilon  + \int_{2\sqrt {\rho n}}^{\infty} n\mathbb{P}( \|\widetilde  B_{1*}\| \ge  \varepsilon)\,d \varepsilon  \\
			&=2\sqrt {\rho n}  + \int_{\sqrt {\rho n}}^{\infty} n\mathbb{P}( \|\widetilde  B_{1*}\| \ge  \varepsilon +\sqrt{\rho n})\,d \varepsilon \le 2\sqrt {\rho n}  + \int_{\sqrt {\rho n}}^{\infty} n \exp (-\varepsilon^2/2)\,d \varepsilon \\
			& \le 2\sqrt {\rho n}  + \frac{1}{\sqrt{\rho n}} \int_{\sqrt {\rho n}}^{\infty} n\varepsilon  \exp (-\varepsilon^2/2)\,d \varepsilon \le 2\sqrt {\rho n}  + n^{1-c_2/2} \frac{1}{\sqrt{c_2 \log(n)}} < C_2\sqrt {\rho n},
	\end{align*}	
	where $C_2 > 2$, the first inequality follows from the observation that $\|\widetilde  B_{i*}\|$ are identically distributed, the second inequality follows from \eqref{eq:column bound} and the last inequality follows from $\rho n\ge c_2\log (n) $, $c_2 > 2$.
\end{proof}
	
Lemmas~\ref{lemma:USVT}	and \ref{lemma:USVT 2} are from Lemma~1 in \cite{Xu2017}.
	
\begin{lemma}
	\label{lemma:USVT}
	Let $A, B\in \mathbb{R}^{n \times n}$ be two symmetric matrices with $\|A - B\|_{\mathrm{op}}< \tau/(1+\delta)$, $\tau >0$.  Then for a fixed $\delta < 1$, 
	we have 
	\[
		\| \usvt (A, \tau, \infty) - B\|_\mathrm{F}^2 \le 16\min_{0\leq r \leq n}\left\{r\tau^2+ (1+\delta)^2\delta^{-2} \sum_{i = r+1}^n \lambda_i^2(B)\right\}.
	\]
\end{lemma}


\begin{lemma}\label{lemma:USVT 2}
	Let $A$ and $B$ be defined as in \Cref{lemma:USVT}, and that $\|B\|_{\infty}\le \tau' $, then
		\[
			\| \usvt (A, \tau,  \tau') - B\|_\mathrm{F}^2 \le 16\min_{0\leq r \leq n}\left\{r\tau^2+ (1+\delta)^2\delta^{-2} \sum_{i = r+1}^n \lambda_i^2(B)\right\}.
		\]
\end{lemma}

\subsection{Proofs of technical results used in \Cref{thm-sbm-main}}
\begin{lemma}\label{lemma:SBM frobenius norm}
	Suppose $A, \Gamma \in \mathbb{R}^{n \times n}$ are symmetric matrices, satisfying that the entries of $A$ are Bernoulli random variables, $\|\Gamma\|_{\infty}\le \rho$ and $\| A-(\Gamma-\diag (\Gamma))\|_{\op}\le (1+\delta)\tau $.  Then 
		\[
			\| \usvt(A,\tau,\infty)- \Gamma\|_{\mathrm{F}}^2 \le 16\min_{0\leq r \leq n}\left\{r\tau^2 + 2(1+\delta)^2\delta^{-2} \sum_{i = r+1}^n \lambda_i\right\} +32(1+\delta)^2\delta^{-2}\| \mathrm{diag} (\Gamma)\|_\mathrm{F}^2,
		\]
	where $\{ \lambda_i\}_{i=1}^n$ are the eigenvalues of $\Gamma$ ordered in decreasing absolute values.

	\end{lemma}

\begin{proof}
	Let $\{ \lambda'_i\}_{i=1}^n$ be the eigenvalues of  $\Gamma-\diag (\Gamma)$ ordered in absolute value, $\{ \lambda_i\}_{i=1}^n$ be the eigenvalues of  $\Gamma$ ordered in absolute value and $ \{ v_i\}_{i=1}^n$ be the eigenvectors  of  $\Lambda$.  Observe that for any orthonormal basis $ \{ u_i\}_{i=1}^n$, and any $r = 1, \ldots, n-1$,
		\[
			\sum_{i=r+1}^n (\lambda'_i)^2  \le \sum_{i=r+1} ^n u_i^{\top}(\Gamma-\diag (\Gamma) )^2u_i.
		\]
		By \Cref{lemma:USVT}, one has
		\[
			\| \usvt (A, \tau, \infty) - (\Gamma-\diag (\Gamma)) \|_{\mathrm{F}}^2 \le 16\min_{0\leq r \leq n}\left\{r\tau^2+ (1+\delta)^2\delta^{-2} \sum_{i = r+1}^n (\lambda'_i)^2\right\}.
		\]
		For any $r = 1, \ldots, n$,
		\begin{align*}
			& \sum_{i = r+1}^n (\lambda_i')^2 \le  \sum_{i=r+1} ^n v_i^{\top}(\Gamma-\diag (\Gamma) )^2v _i  -  v_i^{\top}\Gamma^2 v_i +\sum_{i = r+1}^n \lambda_i^2  \\  
			= & \sum_{i=r+1} ^n v_i^{\top} \left( -2\Gamma  \diag(\Gamma) + \diag(\Lambda)^2 \right)v_i +\sum_{i = r+1}^n \lambda_i^2   \le \sum_{i=r+1} ^n \|\Gamma v_i \|_2^2  +2 v_i^{\top} \diag(\Gamma)^2 v_i +\sum_{i = r+1}^n \lambda_i^2   \\
			\le & 2\sum_{i = r+1}^n  \lambda_i^2    + 2\| \diag (\Gamma)\|_{\mathrm{F}}^2,
		\end{align*}
		which leads to the desired results.
	\end{proof}

\section{Properties of the population CUSUM statistics}
\label{section:properties of cusum}

Recall that in \Cref{def-1} we introduced a general version of CUSUM statistics, which can be applied to various types of data.  In Sections~\ref{sec-vec-cusum} and \ref{sec-1d-cusum}, we apply \Cref{def-1} to vectors and scalars respectively.

\subsection{Vector CUSUM}\label{sec-vec-cusum}

\begin{assumption}\label{assum-multi-cusum}
	Let $\{ V(t)\}_{t=1}^T\subset  \mathbb  R^{p}$. Assume there exists $\{ \nu_m\}_{m=0}^M \subset \{1, \ldots, T\}$ such that $1 = \nu_0 < \nu_1 < \ldots < \nu_M \leq T < \nu_{M+1} = T + 1$ and, for $t=2,\ldots,T$, 
		\[
	V(t) \neq V(t-1) \quad  \text{if and only if} \quad t \in \{ \nu_1,\ldots,\nu_M\}.
	\]
	Let $\inf_{m = 1, \ldots, M}\|V (\nu_m) -V (\nu _{m}-1) \| = \inf_{m = 1, \ldots, M} \kappa_m \ge \kappa = \kappa_0\rho\sqrt{p}$.
\end{assumption}

The results in this subsection are used in the proofs of the main theorems.  Below, $\{V(t)\}_{t=1}^T$ corresponds to $\{\mu(t)\}_{t=1}^T$ as defined in \Cref{assume:model}, and $\kappa = \kappa_0\rho\sqrt{p}$ (see \Cref{assume:model}).  For brevity, we introduce new notation in this subsection such that it is self-contained within this subsection.

For $0 \leq s < t < e \leq T$, denote the CUSUM statistics
	\begin{equation}\label{eq-v-cusum}
	\widetilde V^{s,e}(t) =   \sqrt {\frac{e-t}{(e-s)(t-s)} } \sum_{r=s+1}^t V(r)  -\sqrt {\frac{t-s}{(e-s)(e-t)} } \sum_{r=t+1}^e V(r) .
	\end{equation}
	For simplicity denote  $ \widetilde V(t) =\widetilde V^{0,T}(t)$. It is desired to show that this vector version CUSUM statistics have the same properties as the univariate CUSUM statistic. 

\begin{remark}
	The CUSUM statistic defined in \eqref{eq-v-cusum} is translational invariant. In other words, let $W\in \mathbb R^p$ and $U(t) =V(t)+W$ for all $t$, then
		\[
		\widetilde V(t) =\widetilde U(t).
		\]
		Consequently it can be assumed that $ \sum_{t=1}^T V(t) =0$, and
	\begin{align}\label{eq:translation free}
	\widetilde V(t)  = \left(\sum _{r=1}^tV(r) -\frac{t}{ T} \sum_{r=1}^T V(r)  \right)/ \sqrt {  \frac{t(T-t)}{T} } =  \left(\sum _{r=1}^tV(r) \right)/ \sqrt {  \frac{t(T-t)}{T} }.
	\end{align}
\end{remark}

\begin{proposition}\label{eq:monotonicity of mcusum}
	The quantity $\|\widetilde V (t)\|^2$ is maximized at the change points.  For $t\in [\nu_{m-1},\nu_{m}]$, $\|\widetilde V (t)\|^2$ is either monotone or decreases and then increases.
\end{proposition}

\begin{proof} 
	Let $t\in   (\nu_{m-1},\nu_{m})$. By Equation (2.7) of Lemma~2.2 in  \cite{venkatraman1992consistency}, for every $j = 1, \ldots, p$, $\widetilde V_j (t)$ can be continuously extended to the function 
		\[
		f_j(x)  =\frac{a_j-b_jx}{\sqrt {x(1-x)}},
		\]
		where $x =t/T$, $a_j$ and $b_j$ are defined similarly as in Lemma~2.2 in \cite{venkatraman1992consistency}.  Thus it suffices to show that for $x\in (c,d)$ where $0\le c\le d \le1$, the function
		\[
		f(x)=\sum_{j=1}^{p}\frac{(a_j-b_jx)^2}{x(1-x)}
		\]
		is maximized at either $c$ or $d$.

	Let  
		\[
		f'(x)=\sum_{j=1}^n \frac{-(2a_jx -b_jx -a_j)(b_jx-a_j)}{(x-1)^2x^2} = \frac{g(x)}{(x-1)^2x^2}.
		\]
		The desired result follows if $f'(x)$ is either nonpositive, or nonnegative or that there exists $x_0\in (0,1) $ such that 
		\begin{equation}f'(x)\begin{cases}
			\le 0 \quad \text{when} \quad x \le x_0\\
			\ge 0\quad \text{when} \quad x \ge x_0
			\end{cases}\label{eq:derivative cases}
		\end{equation}
	Since $(x-1)^2 x\ge 0 $ for all $x\in (0,1)$. Observe that $g$ is quadratic and that $g(0) = -\sum_{i=1}^na_i^2\le 0 $ and $g(1) = (b_ix-a_i)^2 \ge0 $.  Therefore $g(x)$ can have at most one root in $(c,d)$. If $g(x)$ has no root in $(c,d)$, then $g(x)$ is either positive or negative.  If $g(x)$ has a root $x_0\in (c,d)$, then \eqref{eq:derivative cases} holds.
\end{proof}

\begin{lemma}
	Suppose there exists a change point $\nu \in (0, T)$ such that any other change point $\nu'$ within $(0,T)$ satisfies $ \min \{ |\nu' -\nu| \}\ge \Delta$. Then
		\[
		\max_{0 < t < T} \| \widetilde V (t)\|^2 \ge \frac{\|V(\nu) -V(\nu+1) \|^2   \Delta^2}{48T} .
		\]
\end{lemma}

\begin{proof}
	Denote $\kappa=\|V(\nu) -V(\nu+1) \|$.

	\noindent{\bf Step 1.} Let 
		\begin{align*}
			I_1&=\left\{ i  : \, \left|\sum_{r=1}^{\nu -\Delta}  V_i(r) \right| \ge  \Delta|V_i (\nu) - V_i(\nu+1)| /4\right\},\\
			I_2&=\left\{ i  : \, \left|\sum_{r=1}^{\nu }  V_i(r) \right| \ge  \Delta|V_i (\nu) - V_i(\nu+1)| /4\right\},\\
			I_3&=\left\{ i  : \, \left|\sum_{r=1}^{\nu +\Delta}  V_i(r) \right| \ge  \Delta|V_i (\nu) - V_i(\nu+1)| /4\right\},\\
		\end{align*}
		Then by \Cref{lemma:1d cusum population},  $I_1\cup I_2\cup I_3 =\{1,\ldots,p\}$.  We have
		\[
		\sum_{l=1}^3\left\{ \sum_{ i \in I_l}(V_i(\nu)- V_i(\nu+1)) ^2 \right\}    \ge \sum_{i=1}^ p (V_i(\nu)- V_i(\nu+1)) ^2  = \kappa^2,
	 	\]
		which implies that
		\[
		\max_{l=1,2,3}\left\{  \sum_{ i \in I_l}(V_i(\nu)- V_i(\nu+1)) ^2  \right\} \ge\kappa^2/3.
		\]
	Without loss of generality, suppose $ \sum_{ i \in I_1}(V_i(\nu)- V_i(\nu+1)) ^2  \ge \kappa ^2/3.$ Then
		\begin{align*}
			\max_{1 < t < T}\|\widetilde V (t) \|^2 &\ge \|\widetilde V (\nu- \Delta)\|^2 = \frac{T } {(\nu -\Delta)(T- (\nu-\Delta))}\left\|\sum_{r=1}^{\nu -\Delta} V(r)\right\|^2\\
	&\ge\frac{1  } {T} \sum_{i \in I_1}\left(\sum_{r=1}^{\nu-\Delta} V_i(r)\right)^2\ge\frac{1} {T} 
	\sum_{i\in I_1} \left(  \Delta|(V_i(\nu) - V_i(\nu+1)) | /4\right)^2\\
	&\ge\frac{ \Delta^2} {48 T}\kappa^2,
	\end{align*}
	where the first equality  follows from 
	\eqref{eq:translation free} and the
	second last  inequality follows from the definition of $I_1$.
\end{proof}

\begin{lemma}\label{lemma:Venkatraman}
	Let $[s,e]\subset [0,T]$ be any generic interval containing a change point $\nu$ satisfying
		\[
			\min \{\nu-s , e-\nu \} \ge c_1\Delta.
		\]
		If
		\[
			\| \widetilde V^{s,e}(\nu)\|^2\ge \kappa ^2 \Delta^2 (e-s)^{-1} ,    
		\]
		and 
		there exists a sufficient small absolute constant $c_3 > 0$ such that 
		\begin{align}
			\max_{1 < t < T}  \|\widetilde V ^{s,e} (t)\|^2  -\| \widetilde V^{s,e}(\nu) \|^2   \le  c_3 \| \widetilde V ^{s,e}(\nu)  \|^2  \Delta^2 (e-s)^{-2},\label{eq:ven 2}
		\end{align}
		then there exists an absolute constant $c, c_1>0$ such that $d\in [s,e]$ satisfying $|d-\nu |\le  c_1\Delta/16 $, and
		\[
		\|\widetilde V^{s,e} (\nu)\|^2  - \|\widetilde V^{s,e} (d)\|^2   > c \|\widetilde V^{s,e} (\nu)\|^2  |\nu-d | \Delta(e-s)^{-2},
		\]	
		where $c = \min\{c_1, 1/2-2c_3\}$.
\end{lemma}
\begin{proof}
	Denote $\widetilde V^{s,e} (t) = \widetilde V (t) $ and  $ l = d- \nu$.
	It suffices consider the case of $l\ge0$, as the case of $l\le 0$ follows by reversing the time series.
	Let $\nu'>\nu$ be the next change point. Then either $\nu' =e$ which means that $\nu$ is the last change point,
	or $\nu'<T$ which indicates that $\nu$ is not the last change point. 

	\vskip 3mm
	\noindent {\bf Case 1.} Suppose $ \nu' =T$. Let $i=\nu -s$ and $h=e-\nu $.  For any $u \in \{1,\ldots,p\}$, by Case 1 in Lemma~2.6 of \cite{venkatraman1992consistency}, it holds that 
		\[
		\widetilde V_u(\nu)=\frac{a_u\sqrt{i+h} }{\sqrt {ih}} , \quad \widetilde V_u (\nu+l)=\frac{h-l}{h}\frac{a_u\sqrt{i+h} }{\sqrt {(i+l)(h-l)}}.
		\]
		Thus 
		\begin{align*}
			\widetilde V_u (\nu)^2 -	\widetilde V_u (\nu+l)^2 = l\frac{a_u^2(i+h)}{ih} \frac{h+i}{h(i+l)} = \frac{l(h+i)}{h(i+l)} \widetilde V_u(\nu)^2.
		\end{align*} 
		So
		\[
			\|\widetilde V(\nu) \|^2 -\|\widetilde V (\nu+l)\|^2 = \frac{l(h+i)}{h(i+l)} \|\widetilde V (\nu) \|^2 \ge \frac{l (e-s)}{ (e-s)^2 }\|\widetilde V (\nu) \|^2 \ge \frac{c_1l\Delta}{ (e-s)^2} \|\widetilde V (\nu) \|^2.
		\]
	
	\vskip 3mm
	\noindent {\bf Case 2.} Suppose $\nu'<e$.  Let $i=\nu-s$, $h=\Delta/2$ and $j=e-\nu-h$.  Let $l\le h/2$.  For any $u \in \{1,\ldots, p\} $, by  Case 2 in Lemma 2.6 of \cite{venkatraman1992consistency}, 
		\[
			\widetilde V_u(\nu)=\frac{a_u\sqrt{i+h} }{\sqrt {ih}} ,\quad \widetilde V_u( \nu+h)=\frac{(a_u+h\theta)\sqrt{i+j+h}}{\sqrt{(i+h)j}}  \quad\text{and}   \quad \widetilde V_u (\nu+l)=\frac{(a_u+l\theta)\sqrt{i+j+h}}{\sqrt{(i+l)(j+h-l)}},
		\]
		where $\theta$ is the solution of 
		\[
			\widetilde V_u ^2( \nu+h) -\widetilde V_u ^2(\nu)=  \frac{(a_u+h\theta)^2(i+j+h)}{(i+h)j}- \frac{a_u^2(i+h)}{ih}.
		\]
		Denote $B = \|\widetilde V (\nu+h)\|^2-\|\widetilde V(\nu)\|^2$ and $B_u= \widetilde V_u(\nu+h)^2-  \widetilde V_u (\nu)^2$. Thus by  \eqref{eq:ven 2},
		\begin{align}\label{eq:venkatraman 2.6 noise size}
			B \le  c_3 \| \widetilde V^{s,e} (\nu)  \|_2^2  \Delta^2 (e-s)^{-2}.
		\end{align}
		Then by  \Cref{lemma:Venkatraman 2.6},
		\begin{align*}
			& \| \widetilde V(\nu)\|^2 - \|\widetilde V (\nu+l)\|^2 =  \sum_{u=1}^p \left\{ \widetilde V_u (\nu)^2 -\widetilde V_u (\nu+l)^2 \right\}\\
			\ge & \sum_{u=1}^p \left\{\frac{\widetilde V_u (\nu)^2 (hl-l^2)}{(i+l)(j+h-l)} - B_u \frac{l(i+h)j}{h(i+l)(j+h-l)} \right\} =  \frac{\| \widetilde V (\nu)\|_2^2 l(h-l)}{(i+l)(j+h-l)} -B\frac{l(i+h)j}{h(i+l)(j+h-l)}  \\
			\ge &\frac{\| \widetilde V (\nu) \|_2^2 l  \Delta}{ 2 (e-s)^2} - 2B\frac{ l }{\Delta} \ge (1/2 - 2c_3)\frac{\| \widetilde V (\nu)\|_2^2 l  \Delta}{(e-s)^2}, 
		\end{align*}
		where the last inequality follows from \eqref{eq:venkatraman 2.6 noise size}.
\end{proof}

\begin{lemma}\label{lemma:Venkatraman 2.6} Denote
	$$\Theta_\nu =\frac{a\sqrt{i+j+h}}{\sqrt{i(j+h)}},\quad 
	\Theta_{\nu+h}=\frac{(a+h\theta)\sqrt{i+j+h}}{\sqrt{(i+h)j}} 
	\quad
	\text{and} \quad 
	\Theta_{\nu+l}=\frac{(a+l\theta)\sqrt{i+j+h}}{\sqrt{(i+l)(j+h-l)}} .
	$$ 
	Then
	$$\Theta_{\nu}^2-\Theta_{\nu+l}^2 \ge \frac{\Theta^2_{\nu} (hl-l^2)}{(i+l)(j+h-l)}
	- (\Theta _{\nu+h}^2-\Theta_\nu^2) \frac{l(i+h)j}{h(i+l)(j+h-l)}.$$
\end{lemma}
\begin{proof}Observe that 
	\begin{align*}
	& \Theta_{\nu}^2-\Theta_{\nu+l}^2 = \frac{a^2(i+j+h)}{i(j+h)}- \frac{(a+l\theta)^2(i+j+h)}{(i+l)(j+h-l)} \\
	&= \frac{a^2(i+j+h)}{i(j+h)(i+l)(j+h-l)}\left((i+l)(j+h-l) -i(j+h)\right) -\frac{(2l\theta a +l^2\theta^2)(i+j+h)}{(i+l)(j+h-l)} \\
	& = \frac{a^2(i+j+h)}{i(j+h)(i+l)(j+h-l)} (-il + lj + lh -l^2 )-(2l\theta a +l^2\theta^2)\frac{(i+j+h)}{(i+l)(j+h-l)}.
	\end{align*}
	To bound the term $2l\theta a +l^2\theta^2$,
	let $b=\Theta_{\nu+h}^2 -\Theta_{\nu}^2$. 
	Then 
	$$b=  \frac{(a+h\theta)^2(i+j+h)}{(i+h)j}-
	\frac{a^2(i+j+h)}{i(j+h)}.
	$$
	Therefore
	$$ \frac{bij(i+h)(j+h)}{i+j+h} = (a^2+2h\theta a + h^2\theta^2 ) i(j+h) -a^2 (i+h)j,$$
	which gives
	$$  2h\theta a + h^2 \theta^2 = \frac{bj(i+h)}{i+j+h} + \frac{a^2(j-i)h}{i(j+h)}.$$
	Therefore 
	\begin{align*}
	 2l\theta a +l^2\theta^2  &\le  2l\theta a  +lh\theta^2  = \frac{l}{h} (2h\theta a  +h^2\theta^2) =  \frac{l}{h} \left( \frac{bj(i+h)}{i+j+h} + \frac{a^2(j-i)h}{i(j+h)}\right)
	\end{align*}
	which implies that 
		\begin{align*}
	\Theta_{\nu}^2-\Theta_{\nu+l}^2 
	& = \frac{a^2(i+j+h)}{i(j+h)(i+l)(j+h-l)} (-il + lj + lh -l^2 )-(2l\theta a +l^2\theta^2)\frac{(i+j+h)}{(i+l)(j+h-l)}
	\\
	& \le 
	 \frac{a^2(i+j+h)}{i(j+h)(i+l)(j+h-l)} (-il + lj + lh -l^2 )  
	 \\
	 & -
	 \frac{l}{h} \left( \frac{bj(i+h)}{i+j+h} + \frac{a^2(j-i)h}{i(j+h)}\right)\frac{(i+j+h)}{(i+l)(j+h-l)}
	 \\
	 & =
	 \frac{a^2(i+j+h)}{i(j+h)(i+l)(j+h-l)} (-il + lj + lh -l^2 )  
	 \\
	 & -
	 \frac{lbj(i+h)}{h(i+l)(j+h-l)}  - \frac{a^2(i+j+h)}{i(j+h)(i+l) (j+h-l)}(j-i) l
	\\
	& 
	=\frac{a^2(i+j+h)}{i(j+h)(i+l)(j+h-l)} (lh-l^2) -b \frac{l(i+h)j}{h(i+l)(j+h-l)},
	\end{align*}
	which concludes the proof.

\end{proof}

\begin{lemma}\label{lemma:one change point population size}
	Suppose $[s,e]$ contains one and only one  change point $\eta_k$, then
		\[
		\|\widetilde  V^{s, e} (t) \|^2=\begin{cases}
			\frac {t-s}{(e-s)(e-t)}  (e-\eta_k)^2  \|V(\eta) -V(\eta+1) \|^2 , & t\le \eta_k, \\
			\frac {e-t}{(e-s)(t-s)}  (\eta_k -s)^2  \|V(\eta) -V(\eta+1) \|^2   , &  t\ge\eta_k.
		\end{cases}
		\]
	\end{lemma}

\begin{proof}
	This is a straightforward result from the definitions.	
\end{proof}

\begin{lemma}\label{lemma:mcusum property near the boundary}
	Let $\eta_1$ be the first change point in $\{1, \ldots, T\}$. Then for any $1\le t\le \eta_1$,
		\[
		\|\widetilde V^{0,T}(t) \|^2 =\frac{ t(T-\eta_1)}{\eta_1(T-t)}\|\widetilde V ^{0,T} (\eta_1)\|^2  .
		\]
\end{lemma}

\begin{proof}
	This is a direct consequence of \Cref{lemma:cusum property near the boundary}.
\end{proof}

\begin{lemma}\label{lemma:mcusum boundary bound}
	Let $[s,e]$ contain two or more change points such that 
	\[
	\eta_{r-1} \le s < \eta_r < \ldots < \eta_{r+q} \le e \le \eta_{r+q+1}, \quad q\ge 1.
	\]
	If $\eta_{r}-s \le  c\Delta $ for some $c\le 1/4$ and $\eta_{r+1}-\eta_r \ge \Delta, $
	then
	$$\|\widetilde V^{s,e}(\eta_r ) \|^2 \le 2c  \|\widetilde V ^{s,e}({\eta_{r+1} }) \|^2  +4 \kappa_r^2 (\eta_r-s).$$
	
	If there are two and only two change points, then 
	\[
		\max_{s < t < e} \|\widetilde{V}^{s, e}(t)\|^2 \leq (e - \eta_{r+1}) \kappa_{r+1}^2 + (\eta_r - s)\kappa_r^2.
	\]
\end{lemma}
\begin{proof}
	This follows from a similar calculation as in  \Cref{lemma:cusum boundary bound}.
\end{proof}

\subsection{Univariate CUSUM}\label{sec-1d-cusum}

\begin{assumption}
	Let $\{ f(t)\}_{t=1}^T\subset  \mathbb  R$. Assume there exists a sequence $\{ \nu_m\}_{m=0}^M \subset \{1, \ldots, T\}$ such that $1 = \nu_0 < \nu_1 < \ldots < \nu_M \leq T < \nu_{K+1} = T + 1$ and, for $t=2,\ldots,T$, 
	\[
	f(t) \neq f(t-1) \quad  \text{if and only if} \quad t \in \{ \nu_1,\ldots,\nu_M\}.
	\]
	We  set
	\[
		|f (\nu_{m}) -f(\eta_{m}-1 ) |=\kappa_m \geq \kappa.	
	\]
\end{assumption}

For the same reasons as we described after \Cref{assum-multi-cusum}, in this subsection we use a self-contained notation system, and one can interpret $\kappa = \kappa_0n\rho$ as we used in \Cref{assume:model}.

\begin{lemma}\label{lemma:1d cusum population}
	Suppose $ \nu_m$ is a change point of $\{ f (t)\}_{t=1}^T  $ such that $\min_{m'\not = m} \{ \nu_m - \nu_{m'}  \}\ge \Delta $. Then 
		\begin{equation}\label{eq:size at change point}
			\max\left \{ \left|\sum_{r=1}^{\nu_m-\Delta} f(r) \right|,\quad \left|\sum_{r=1}^{\nu_m } f(r) \right|,\quad \left|\sum_{r=1}^{\nu_m + \Delta} f(r) \right| \right\}  \ge  \Delta|f(\nu_m) -f(\nu_m+1)| /4.
		\end{equation}
\end{lemma}

\begin{proof} For simplicity denote $\nu_m =\nu $.
	Observe that 
	$$ \max\{ |  f(\nu) | , |f(\nu+1) | \} \ge|f(\nu) -f(\nu+1)| /2. $$
	Thus 
	\begin{align}
	\max
	\left \{ \left|\sum_{r=\nu -\Delta}^{\nu}f(r) \right|,
	\left|\sum_{r=\nu+1}^{\nu +\Delta}  f(r)\right| 
	\right \}
	\ge 
	\Delta|f(\nu) -f(\nu+1)| /2 .
	\label{eq:mcusum size 2}
	\end{align}
	Since 
	\begin{align}
	\left|\sum_{r=\nu -\Delta}^{\nu}f(r) \right|  \le \left|\sum_{r=1}^{\nu-\Delta} f(r) \right|
	+
	\left|\sum_{r=1}^{\nu } f(r) \right|
	\quad \mbox{and}\quad 
	\left|\sum_{r=\nu+1}^{\nu +\Delta} f(r)\right|  \le \left|\sum_{r=1}^{\nu }f(r) \right|
	+
	\left|\sum_{r=1}^{\nu + \Delta} f(r) \right| \label{eq:mcusum size 3},
	\end{align}
	we have that \eqref{eq:mcusum size 2} and \eqref{eq:mcusum size 3} directly imply \eqref{eq:size at change point}.
\end{proof}

\begin{lemma}\label{lemma:cusum property near the boundary}
	Let $\eta_1$ be the first change point in $\{2, \ldots, T\}$. Then for any $1\le t < \eta_1$,
		\[
		\widetilde f^{0,T}_t  = \sqrt {\frac{ t(T-\eta_1)}{\eta_1(T-t)}}\widetilde f^{1,T}_{\eta_1} .
		\]
\end{lemma}

\begin{proof}
	Without loss of generality assume $\sum_{t=1}^T f_t =0 $.  Thus $\eta_1 f_1 = \sum_{t=1}^{\eta_1}f_t = - \sum_{t=\eta_1 +1} ^T f_t$.  As a result, for any $1\le t < \eta_1$,
	\begin{align*}
	\widetilde f^{1,T}_{t} 
	&= \sqrt { \frac{T-t}{Tt} } \sum_{i=1}^{t}f_i - \sqrt { \frac{t}{T(T-t)} } \sum_{i=t+1}^{T}f_i
	\\
	&=  \sqrt { \frac{T-t}{Tt} }tf_1 - \sqrt { \frac{t}{T(T-t)} }\left(  (\eta_1- t) f_1 + \sum_{i=\eta_1+1}^T f_i\right)
	\\
	&=  \sqrt { \frac{T-t}{Tt} }tf_1 - \sqrt { \frac{t}{T(T-t)} }\left\{(\eta_1- t) f_1  -\eta_1 f_1\right\} =\frac{(T-t)\sqrt t +t\sqrt t }{\sqrt{T(T-t)}}f_1 =\sqrt{ \frac{Tt}{T-t} }f_1.
	\end{align*}
	
\end{proof}

\begin{remark}
	If there exists $b \in [1,\eta_1]$ such that  $\widetilde f^{1,T}_b >0$, then by \Cref{lemma:cusum property near the boundary}, $ \widetilde f^{1,T}_{\eta_1} >0$.
	Since  for $t \in [1,\eta_1] $,  $ \sqrt {\frac{ t(T-\eta_1)}{\eta_1(T-t)}} $ is an increasing function of $t$,  this also 
	implies $\widetilde f^{1,T}_t >0$ is increasing within $[1,\eta_1]$, as a function of $t$.
	
\end{remark}

\begin{lemma}\label{lemma:cusum boundary bound}
	Let $[s,e]$ contain two or more change points such that 
		\[
		\eta_{r-1} \le s\le \eta_r \le \ldots\le \eta_{r+q} \le e \le \eta_{r+q+1}, \quad q\ge 1.
		\]
		If 
		$\eta_{r}-s \le  c_1^2\Delta $ for some $c_1\le 1/4$ and $\eta_{r+1}-\eta_{r} \ge \Delta $,
		then
		$$|\widetilde f^{s,e}_{\eta_r}| \le c_1  |  \widetilde f^{s,e}_{\eta_{r+1}}| +2\kappa_r  \sqrt {\eta_r -s}. $$
		
		If $[s, e]$ contains two and only two change points $\eta_r$ and $\eta_{r+1}$, then
		\[
			\max_{s < t < e} \left|\widetilde{f}^{s, e}_t\right| \leq \sqrt{e - \eta_{r+1}} \kappa_{r+1} + \sqrt{\eta_r - s} \kappa_r.
		\]
		
\end{lemma}

\begin{proof}
	Consider the sequence $\{g_t\}_{t=s+1}^e $ be such that 
		\[
		g_t=
		\begin{cases}
		f_{\eta_{r+1}}, & \text{if} \quad s +1\le  t\le \eta_{r},
		\\
		f_t, & \text{if} \quad \eta_{r}+1 \le t \le e. 
		\end{cases}
		\]

	For any  $t\ge \eta_r +1$,
		\begin{align*}
		&\widetilde f^{s,e}_{t} - \widetilde g^{s,e}_{t} 
		\\
		=&
	\sqrt { \frac{e-t}{(e-s)(t-s)} }  \left(	\sum_{i=s+1}^{\eta_r} f_{\eta_r} + \sum_{i=\eta_r +1}^{t  } f_{\eta_{r+1}} 
	- \sum_{i=s+1}^{\eta_r} g _{\eta_r}  - \sum_{i=\eta_r +1}^{t  } g_{\eta_{r+1}}  
	 \right) \\
	 &\hspace{5cm}- 	\sqrt { \frac{t-s}{(e-s)(e-t) } }  \left(	\sum_{i=t+1}^{e} f_{t} - \sum_{i=t+1}^{e} g_{t}
	 \right)
	 \\
	 =	&\sqrt { \frac{e-t}{(e-s)(t-s)} } (\eta_r-s) (f_{\eta_{r+1}} -f_{\eta_{r}})  \le \sqrt{\eta_r-s} \kappa_r . 
		\end{align*}
		Thus 
		\begin{align*}
		|\widetilde f^{s,e}_{\eta_r} | &
		\le  |\widetilde g^{s,e}_{\eta_r}  |+  \sqrt{\eta_r-s} \kappa_r \le \sqrt { \frac{(\eta_r-s)  (e-\eta_{r+1})  }{   ( \eta_{r+1}-s)  (e-\eta_r)    } }|\widetilde g^{s,e}_{\eta_{r+1}} |+  \sqrt{\eta_r-s} \kappa_r
		\\
		&\le \sqrt { \frac{c_1^2\Delta }{ \Delta}}|\widetilde g^{s,e}_{\eta_{r+1}} |  +  \sqrt{\eta_r-s} \kappa_r \le c_1 |\widetilde f^{s,e}_{\eta_{r+1}} |  + 2\sqrt{\eta_r-s} \kappa_r,
		\end{align*}
		where the first inequality follows from \Cref{lemma:cusum property near the boundary} and the observation that the first change point of $g_t$ in $[s,e]$ is  $\eta_{r+1}$.

		If there are two and only two change points, then 
		\[
			\max_{s < t < e} \left|\widetilde{f}^{s, e}_t\right| = \max\left\{|\widetilde f^{s,e}_{\eta_r} |, |\widetilde f^{s,e}_{\eta_{r+1}} |\right\} \leq \max_{s < t < e} |\widetilde g^{s,e}_t| + \sqrt{\eta_{r} - s} \kappa_r \leq \sqrt{e - \eta_{r+1}} \kappa_{r+1} + \sqrt{\eta_r - s} \kappa_r.
		\]	
	\end{proof}

 \section{Additional lemmas}\label{app:ancillary}
 \begin{lemma}
 	Suppose $x>0$ and that $x^2 +bx-c\ge0$ where $b, c>0$ and that 
 		\[
 	b\le \sqrt c /4.
 		\] 
 		Then $x\ge  7\sqrt c/8$.
 \end{lemma}

 \begin{proof} 
	We have either $x \ge \frac{-b + \sqrt {b^2  +4c}}{2}$ or $\quad x \le \frac{-b - \sqrt {b^2  +4c}}{2}$.  Since $x, b, c > 0$ and $  b\le \sqrt c /4$, we have
		\[
		x\ge  \frac{-b + \sqrt {b^2  +4c}}{2} \ge 7\sqrt c /8.
		\]
\end{proof}
	
 Let $\{\alpha_m\}_{m=1}^M,\{\beta_m\}_{m=1}^M$ be two sequences independently selected at random from $\{1, \ldots, T\}$, and   
	\begin{equation}\label{event-M}
		\mathcal{M} = \bigcap_{k = 1}^K \bigl\{\alpha_m \in \mathcal{S}_k, \beta_m \in \mathcal{E}_k, \, \mbox{for some }m \in \{1, \ldots, M\}\bigr\}, 
	\end{equation}
	where $\mathcal S_{k}= [\eta_k-3\Delta/4, \eta_k-\Delta/2 ]$ and $\mathcal
	E_{k}= [\eta_k+\Delta/2, \eta_k+3\Delta/4 ]$, $k = 1, \ldots, K$.  In the lemma below, we give a lower bound on the probability of $\mathcal{M}$. 

 \begin{lemma}\label{lemma:random interval}
	For the event $\mathcal{M}$ defined in \eqref{event-M}, we have
 
		\begin{equation}\label{eq:event.E}
			\mathbb{P}(\mathcal M) \geq 1 -\exp\left\{\log\left(\frac{T}{\Delta}\right) - \frac{M\Delta^2}{16 T^2} \right\}.
		 \end{equation}
 
\end{lemma}
	
\begin{proof}
	Since the number of change points are bounded by $T/\Delta $,
		\begin{align*}
			\mathbb{P}\bigl(\mathcal{M}^c\bigr) \leq \sum_{k=1}^K \prod_{m =1}^M \bigl\{1 - \mathbb{P}\bigl(\alpha_m \in \mathcal{S}_k, \beta_m \in \mathcal{E}_k\bigr)\bigr\}	 \leq K (1-\Delta^2/(16T^2))^M \leq (T/\Delta)  (1 - \Delta^2/(16T^2))^M.
		\end{align*}
\end{proof}

\end{document}